\documentclass[sigconf, nonacm]{acmart}

\usepackage{color, colortbl}
\usepackage[english]{babel}
\usepackage{mathtools}
\usepackage[labelfont=bf,width=0.95\textwidth]{caption}
\usepackage{subcaption}
\usepackage{amsmath}
\usepackage{amsfonts}
\usepackage{graphicx}
\usepackage{thm-restate}
\usepackage{bbm}
\usepackage{xspace}
\usepackage{algorithm}
\usepackage{algorithmic}
\usepackage[capitalize,noabbrev]{cleveref}
\usepackage{multirow}
\usepackage{pifont}
\usepackage{balance}

\newcommand{\set}[1]{\{#1\}}
\DeclareMathOperator*{\argmin}{argmin}

\DeclarePairedDelimiter\ceil{\lceil}{\rceil}

\newtheorem{theorem}{Theorem}
\newtheorem{lemma}{Lemma}
\newtheorem{definition}{Definition}
\newtheorem{problem}{Problem}
\newtheorem{proposition}{Proposition}
\newtheorem{corollary}{Corollary}

\newcommand{\revision}[1]{#1} 

\renewcommand{\check}{\ding{51}}
\newcommand{\norm}[1]{\left\lVert#1\right\rVert}
\newcommand{\eat}[1]{}

\newcommand{\dataset}[1]{\textsc{#1}\xspace}
\newcommand{\adult}{\dataset{adult}}
\newcommand{\msnbc}{\dataset{msnbc}}
\newcommand{\fire}{\dataset{fire}}
\newcommand{\salary}{\dataset{salary}}
\newcommand{\nltcs}{\dataset{nltcs}}
\newcommand{\titanic}{\dataset{titanic}}

\newcommand{\workload}[1]{\textsc{#1}\xspace}
\newcommand{\general}{\workload{all-3way}}
\newcommand{\target}{\workload{target}}
\newcommand{\weighted}{\workload{skewed}}

\newcommand{\mech}[1]{\textsf{#1}\xspace}
\newcommand{\gaussian}{\mech{Gaussian}}
\newcommand{\gaussianpgm}{\mech{Gaussian+PGM}}
\newcommand{\pgm}{\mech{Private-PGM}}
\newcommand{\aim}{\mech{AIM}}
\newcommand{\hdmm}{\mech{HDMM+PGM}}
\newcommand{\mwempgm}{\mech{MWEM+PGM}}
\newcommand{\mwem}{\mech{MWEM}}
\newcommand{\mst}{\mech{MST}}
\newcommand{\privmrf}{\mech{PrivMRF}}
\newcommand{\privsyn}{\mech{PrivSyn}}
\newcommand{\rap}{\mech{RAP}}
\newcommand{\gem}{\mech{GEM}}

\newcommand{\privbayes}{\mech{PrivBayes}}
\newcommand{\privbayespgm}{\mech{PrivBayes+PGM}}
\newcommand{\independent}{\mech{Independent}}

\newcommand{\jtsize}{\textsf{JT-SIZE}\xspace}

\newcommand{\dom}{\Omega}

\def\alg{{\mathcal M}}

\newcommand{\fancyPara}[1]{\vspace{4.5pt} \textit{\textbf{#1}}.}

\newcommand\vldbdoi{10.14778/3551793.3551817}
\newcommand\vldbpages{2599 - 2612}
\newcommand\vldbvolume{15}
\newcommand\vldbissue{11}
\newcommand\vldbyear{2022}
\newcommand\vldbauthors{\authors}
\newcommand\vldbtitle{\aim: An Adaptive and Iterative Mechanism for Differentially Private Synthetic Data}
\newcommand\vldbavailabilityurl{}
\newcommand\vldbpagestyle{empty}

\begin{document}
\title{\aim: An Adaptive and Iterative Mechanism for \\ Differentially Private Synthetic Data}

\author{Ryan McKenna}
\author{Brett Mullins}
\affiliation{%
  \institution{University of Massachusetts}
  \streetaddress{140 Governors Drive}
  \city{Amherst}
  \state{Massachusetts}
  \postcode{01002}
}
\email{{ rmckenna, bmullins }@cs.umass.edu}

\author{Daniel Sheldon}
\author{Gerome Miklau}
\affiliation{%
  \institution{University of Massachusetts}
  \streetaddress{140 Governors Drive}
  \city{Amherst}
  \state{Massachusetts}
  \postcode{01002}
}
\email{{ sheldon, miklau }@cs.umass.edu}

\begin{abstract}
We propose \aim, a new algorithm for differentially private synthetic data generation. \aim is a workload-adaptive algorithm within the paradigm of algorithms that first selects a set of queries, then privately measures those queries, and finally generates synthetic data from the noisy measurements.  It uses a set of innovative features to iteratively select the most useful measurements, reflecting both their relevance to the workload and their value in approximating the input data. We also provide analytic expressions to bound per-query error with high probability which can be used to construct confidence intervals and inform users about the accuracy of generated data. We show empirically that \aim consistently outperforms a wide variety of existing mechanisms across a variety of experimental settings.
\end{abstract}

\maketitle

\pagestyle{\vldbpagestyle}
\begingroup\small\noindent\raggedright\textbf{PVLDB Reference Format:}\\
\vldbauthors. \vldbtitle. PVLDB, \vldbvolume(\vldbissue): \vldbpages, \vldbyear.\\
\href{https://doi.org/\vldbdoi}{doi:\vldbdoi}
\endgroup
\begingroup
\renewcommand\thefootnote{}\footnote{\noindent
This work is licensed under the Creative Commons BY-NC-ND 4.0 International License. Visit \url{https://creativecommons.org/licenses/by-nc-nd/4.0/} to view a copy of this license. For any use beyond those covered by this license, obtain permission by emailing \href{mailto:info@vldb.org}{info@vldb.org}. Copyright is held by the owner/author(s). Publication rights licensed to the VLDB Endowment. \\
\raggedright Proceedings of the VLDB Endowment, Vol. \vldbvolume, No. \vldbissue\ %
ISSN 2150-8097. \\
\href{https://doi.org/\vldbdoi}{doi:\vldbdoi} \\
}\addtocounter{footnote}{-1}\endgroup

\ifdefempty{\vldbavailabilityurl}{}{
\vspace{.3cm}
\begingroup\small\noindent\raggedright\textbf{PVLDB Artifact Availability:}\\
The source code, data, and/or other artifacts have been made available at \url{\vldbavailabilityurl}.
\endgroup
}

\section{Introduction} \label{sec:intro}

Differential privacy \cite{dwork2006calibrating} has grown into the preferred standard for privacy protection, with significant adoption by both commercial and governmental enterprises. Many common computations on data can be performed in a differentially private manner, including aggregates, statistical summaries, and the training of a wide variety predictive models. Yet one of the most appealing uses of differential privacy is the generation of synthetic data, which is a collection of records matching the input schema, intended to be broadly representative of the source data. Differentially private synthetic data is an active area of research~ \cite{zhang2017privbayes,chen2015differentially,zhang2019differentially,xu2017dppro,xie2018differentially,torfi2020differentially,vietri2020new,liu2016model,torkzadehmahani2019dp,charest2011can,ge2020kamino,huang2019psyndb,jordon2018pate,zhang2018differentially,tantipongpipat2019differentially,abay2018privacy,bindschaedler2017plausible,zhang2020privsyn,asghar2019differentially,li2014differentially} and has also been the basis for two competitions, hosted by the U.S. National Institute of Standards and Technology~\cite{ridgeway2021challenge}.

Private synthetic data is appealing because it fits any data processing workflow designed for the original data, and, on its face, the user may believe they can perform {\em any} computation they wish, while still enjoying the benefits of privacy protection. Unfortunately, it is well-known that there are limits to the accuracy that can be provided by synthetic data under differential privacy or any other reasonable notion of privacy~\cite{dinur2003revealing}.

As a consequence, it is important to tailor synthetic data to some class of tasks, and this is commonly done by asking the user to provide a set of queries, called the workload, to which the synthetic data can be tailored. However, as our experiments will show, existing workload-aware techniques often fail to outperform workload-agnostic mechanisms, even when evaluated specifically on their target workloads. Not only do these algorithms fail to produce accurate synthetic data, but they provide no way for end-users to detect the inaccuracy. As a result, in practical terms, differentially private synthetic data generation remains an unsolved problem.

In this work, we advance the state-of-the-art of differentially private synthetic data in two key ways. First, we propose a new workload-aware mechanism that offers lower error than all competing techniques. Second, we derive analytic expressions to bound the per-query error of the mechanism with high probability.

Our mechanism, \aim, follows the select-measure-generate paradigm, which can be used to describe many prior approaches.\footnote{Another common approach is based on GANs~\cite{goodfellow2014generative}. Recent research~\cite{tao2021benchmarking} has shown that published GAN-based approaches rarely outperform simple baselines; therefore, we do not compare with those techniques in this paper.}  Mechanisms following this paradigm first \emph{select} a set of queries, then \emph{measure} those queries in a differentially private way (through noise addition), and finally \emph{generate} synthetic data consistent with the noisy measurements. We leverage \pgm~\cite{mckenna2019graphical} for the generate step, as it provides a robust and efficient method for combining the noisy measurements into a single consistent representation from which records can be sampled.

The low error of \aim is primarily due to innovations in the \emph{select} stage. \aim uses an iterative, greedy selection procedure, inspired by the popular \mwem algorithm for linear query answering. Through careful analysis, we define a low-sensitivity quality score function to determine the best marginal to measure next, which takes into account: (i) how well the candidate marginal is already estimated, (ii) the expected improvement measuring it can offer, (iii) the relevance of the marginal to the workload, and (iv) the available privacy budget\eat{which may discourage the selection of large marginals}. This new quality score is accompanied by a host of other algorithmic techniques including adaptive selection of rounds and budget-per-round, intelligent initialization, and new set of candidates from which to select.

In conjunction with \aim, we develop new techniques to quantify the uncertainty in query answers derived from the generated synthetic data. \revision{The bounds on error are useful in practice to understand which queries the synthetic data supports well, and which it does not, and are therefore critical to avoid the mis-use of the data by downstream users, a danger that could limit the adoption of synthetic data~\cite{kingnoisy}.  To the best of our knowledge, \aim is the first synthetic data mechansim equipped with such guarantees.}
The problem of error quantification for data independent mechanisms like the Laplace or Gaussian mechanism is trivial, as they provide unbiased answers with known variance to all queries. The problem is considerably more challenging for data-dependent mechanisms like \aim, where complex post-processing is performed and only a subset of workload queries have unbiased answers. Some mechanisms, like \mwem, provide theoretical guarantees on their worst-case error, under suitable assumptions. However, this is an \emph{a priori} bound on error obtained from a theoretical analysis of the mechanism under worst-case datasets. Instead, we develop an \emph{a posteriori} error analysis, derived from the intermediate differentially private measurements used to produce the synthetic data. Our error estimates therefore reflect the actual execution of \aim on the input data but do not require any additional privacy budget for their calculation. 

This paper makes the following contributions:

\begin{enumerate}
\item In \cref{sec:prior}, we assess the prior work in the field, characterizing different approaches via key distinguishing elements and limitations, which brings clarity to a complex space.
\item In \cref{sec:aim}, we propose \aim, a new mechanism for synthetic data generation that is workload-aware (for workloads consisting of weighted marginals) as well as data-aware.
\item In \cref{sec:uncertainty}, we derive analytic expressions to bound the per-query error of \aim with high probability.  These expressions can be used to construct confidence bounds.
\item In \cref{sec:experiments}, we conduct a comprehensive empirical evaluation and show that \aim consistently outperforms all prior work, improving error over the next best mechanism by $1.6\times$ on average and up to $5.7\times$ in some cases.
\end{enumerate}



\section{Background} \label{sec:background}
\revision{
\begin{table}
\caption{\label{table:notation} \revision{Table of Notation}}
\begin{tabular}{|c|c|}
\hline
\textbf{Symbol} & \textbf{Meaning} \\\hline
$\epsilon, \delta, \rho$ & Privacy parameters \\
$\Omega$ & Domain \\
$d$ & Number of attributes \\
$x$ & Single record in $\Omega$\\
$D$ & Dataset of records in $\Omega$ \\
$r$ & Subset of attributes \\
$M_r$ & Marginal query \\
$n_r$ & Domain size of attributes $r$ \\
$W$ & Workload (marginal queries + weights \\
$c_r$ & Weight on marginal $r$ in the workload \\
$\Delta$ & Sensitivity \\
$p$ & Distribution over $\Omega$ \\
$\mathcal{S}$ & Set of all distributions over $\Omega$ \\
$T$ & Number of rounds of \mwempgm \\
$\alpha$ & Privacy budget split of \aim \\
$w_r$ & Weighted assigned to marginal $r$ by \aim \\\hline
\end{tabular}
\end{table}}

In this section we provide the requisite background on datasets, marginals, and differential privacy needed to understand this work.
\subsection{Data, Marginals, and Workloads} \label{sec:workload}

\paragraph*{\textbf{Data}}

A dataset $D$ is a multiset of $N$ records, each containing potentially sensitive information about one individual.  Each record $x \in D$ is a $d$-tuple $(x_1, \dots, x_d)$. The domain of possible values for $x_i$ is denoted by $\Omega_i$, which we assume is finite and has size $ | \Omega_i | = n_i$.   The full domain of possible values for $x$ is thus $ \Omega = \Omega_1 \times \dots \times \Omega_d$ which has size $\prod_i n_i = n $.  We use $\mathcal{D}$ to denote the set of all possible datasets, which is equal to $ \cup_{N=0}^{\infty} \Omega^N $. 

\paragraph*{\textbf{Marginals}}

A marginal is a central statistic to the techniques studied in this paper, as it captures low-dimensional structure common in high-dimensional data distributions.
A marginal for a set of attributes $r$ is essentially a histogram over $x_r$: it is a table that counts the number of occurrences of each $t \in \Omega_r$.  

\begin{definition}[Marginal]
Let $r \subseteq [d]$ be a subset of attributes, $\Omega_r = \prod_{i \in r} \Omega_i$, $n_r = | \Omega_r |$, and $x_r = (x_i)_{i \in r}$.  The marginal on $r$ is a vector $\mu \in \mathbb{R}^{n_r}$, indexed by domain elements $t \in \Omega_r$, such that each entry is a count, i.e., $\mu[t] = \sum_{x \in D} \mathbbm{1}[x_r = t]$.  We let $M_r : \mathcal{D} \rightarrow \mathbb{R}^{n_r}$ denote the function that computes the marginal on $r$, i.e., $ \mu = M_r(D) $.
\end{definition}

In this paper, we use the term \emph{marginal query} to denote the function $M_r$, and \emph{marginal} to denote the vector of counts $ \mu = M_r(D) $.  With some abuse of terminology, we will sometimes refer to the attribute subset $r$ as a marginal query as well.

\paragraph*{\textbf{Workload}}

A workload is a collection of queries the synthetic data should preserve well.  It represents the measure by which we will evaluate utility of different mechanisms.  We want our mechanisms to take a workload as input and adapt intelligently to the queries in it, providing synthetic data that is tailored to the queries of interest.  In this work, we focus on the special (but common) case where the workload consists of a collection of weighted marginal queries.  Our utility measure is stated in \cref{def:error}.

\begin{definition}[Workload Error] \label{def:error}
A workload $W$ consists of a list of marginal queries $r_1, \dots, r_k$ where $r_i \subseteq [d] $, together with associated weights $c_i \geq 0 $.  The error of a synthetic dataset $\hat{D}$ is defined as:
$$ \text{Error}(D, \hat{D}) = \frac{1}{k \cdot |D|}\sum_{i=1}^k c_i \norm{M_{r_i}(D) - M_{r_i}(\hat{D})}_1 $$
\end{definition}

\revision{
  We measure error using a normalized $L_1$ distance between the true workload query answers and the synthetic workload query answers.  This $L_1$ error metric is a common choice  \cite{zhang2020privsyn,cai2021data,zhang2017privbayes,mckenna2019graphical}; although, alternatives have been considered in prior work including $L_{\infty}$ error \cite{aydore2021differentially,liu2021iterative,vietri2020new,liu2021leveraging} and $L_2$ (squared) error \cite{chen2015differentially,mckenna2018optimizing}.  The $L_1$ metric is appealing because it captures the overall error better than the $L_{\infty}$ metric, and is easily interpretable. We also provide supplemental evaluations with $L_{\infty}$ and $L_2$ error in \cref{sec:metrics} of the full paper.
}

\subsection{Differential Privacy}
Differential privacy protects individuals by bounding the impact any one individual can have on the output of an algorithm. This is formalized using the notion of neighboring datasets.  Two datasets $D, D' \in \mathcal{D}$ are neighbors (denoted $ D \sim D'$) if $D'$ can be obtained from $D$ by adding or removing a single record.


\begin{definition}[Differential Privacy] \label{def:dp}
A randomized mechanism $\alg : \mathcal{D} \rightarrow \mathcal{R} $ satisfies $(\epsilon, \delta)$-differential privacy (DP) if for any neighboring datasets $D \sim D' \in \mathcal{D}$, and any subset of possible outputs $S \subseteq \mathcal{R}$,
$$ \Pr[\alg(D) \in S] \leq \exp(\epsilon) \Pr[\alg(D') \in S] + \delta.$$
\end{definition}

A key quantity needed to reason about the privacy of common randomized mechanisms is the \emph{sensitivity}, defined below.

\begin{definition}[Sensitivity]
Let $f : \mathcal{D} \rightarrow \mathbb{R}^p $ be a vector-valued function of the input data.  The $L_2$ sensitivity of $f$ is \\ $\Delta(f) = \max_{D \sim D'} \norm{f(D) - f(D')}_2$.
\end{definition}

It is easy to verify that the $L_2$ sensitivity of any marginal query $M_r$ is $1$, regardless of the attributes in $r$. This is because one individual can only contribute a count of one to a single cell of the output vector.  Below we introduce the two building block mechanisms used in this work.

\begin{definition}[Gaussian Mechanism]
Let $f : \mathcal{D} \rightarrow \mathbb{R}^p $ be a vector-valued function of the input data. The Gaussian Mechanism adds i.i.d. Gaussian noise with scale $\sigma \Delta(f)$ to each entry of $f(D)$.  That is,
$$ \alg(D) = f(D) + \sigma \Delta(f) \mathcal{N}(0, \mathbb{I}),$$
where $\mathbb{I}$ is a $p \times p$ identity matrix.
\end{definition}
\begin{definition}[Exponential Mechanism]
Let $q_r : \mathcal{D} \rightarrow \mathbb{R}$ be quality score function defined for all $r \in \mathcal{R}$ and let $\epsilon \geq 0$ be a real number.  Then the exponential mechanism outputs a candidate $r \in \mathcal{R}$ according to the following distribution:
\vspace{-0.5em}
$$ \Pr[\mathcal{M}(D) = r] \propto \exp{\Big( \frac{\epsilon}{2 \Delta} \cdot q_r(D) \Big)}, $$
where $ \Delta = \max_{r \in \mathcal{R}} \Delta(q_r)$.
\end{definition}

Our algorithm is defined using zCDP, an alternate version of differential privacy definition which offers beneficial composition properties.  We convert to $(\epsilon, \delta)$ guarantees when necessary.
\begin{definition}[zero-Concentrated Differential Privacy (zCDP)] \label{df:zcdp}
A randomized mechanism $\alg$ is $\rho$-zCDP if for any two neighboring datasets $D$ and $D'$, and all $ \alpha \in (1, \infty)$, we have:
$$ D_{\alpha}(\alg(D) \mid \mid \alg(D')) \leq \rho \cdot \alpha, $$

where $D_{\alpha}$ is the Rényi divergence of order $\alpha$.
\end{definition}

\begin{proposition}[zCDP of the Gaussian Mechanism \cite{bun2016concentrated}] \label{prop:rdpgauss}
The Gaussian Mechanism satisfies $ \frac{1}{2 \sigma^2}$-zCDP.
\end{proposition}

\begin{proposition}[zCDP of the Exponential Mechanism \cite{cesar2021bounding}] \label{prop:expprivacy}
The Exponential Mechanism satisfies $\frac{\epsilon^2}{8}$-zCDP.
\end{proposition}
We rely on the following propositions to reason about multiple adaptive invocations of zCDP mechanisms, and the translation from zCDP to $(\epsilon, \delta)$-DP.  The proposition below covers 2-fold adaptive composition of zCDP mechanisms, and it can be inductively applied to obtain analogous k-fold adaptive composition guarantees.

\begin{proposition}[Adaptive Composition of zCDP Mechanisms \cite{bun2016concentrated}] \label{prop:composition}
Let $\alg_1 : \mathcal{D} \rightarrow \mathcal{R}_1$ be $\rho_1$-zCDP and $\alg_2 : \mathcal{D} \times \mathcal{R}_1 \rightarrow \mathcal{R}_2$ be $\rho_2$-zCDP.  Then the mechanism $\alg = \alg_2(D, \alg_1(D))$ is $(\rho_1 + \rho_2)$-zCDP.
\end{proposition}

\begin{proposition}[zCDP to DP \cite{canonne2020discrete}] \label{prop:rdpdp}
If a mechanism $\alg$ satisfies $\rho$-zCDP, it also satisfies
 $(\epsilon, \delta)$-differential privacy for all $\epsilon \geq 0$ and

$$ \delta = \min_{\alpha > 1} \frac{\exp{\big((\alpha-1)(\alpha \rho - \epsilon)\big)}}{\alpha-1} \Big(1 - \frac{1}{\alpha}\Big)^{\alpha}. $$
\end{proposition}

\subsection{Private-PGM}

An important component of our approach is a tool called \pgm \cite{mckenna2019graphical,mckenna2021winning,dporgblog}.  For the purposes of this paper, we will treat \pgm as a black box that exposes an interface for solving subproblems important to our mechanism.  We briefly summarize \pgm and three core utilities it provides.  \pgm consumes as input a collection of noisy marginals of the sensitive data, in the format of a list of tuples $(\tilde{\mu}_i, \sigma_i, r_i)$ for $i = 1, \dots, k$, where $\tilde{\mu}_i = M_{r_i}(D) + \mathcal{N}(0, \sigma_i^2 \mathbb{I})$.\footnote{\pgm is more general than this, but this is the most common setting.}

\paragraph*{\textbf{Distribution Estimation}}

At the heart of \pgm is an optimization problem to find a distribution $\hat{p}$ that ``best explains'' the noisy observations $\tilde{\mu}_i$:
$$ \hat{p} \in \argmin_{p \in \mathcal{S}} \sum_{i=1}^k \frac{1}{\sigma_i} \norm{M_{r_i}(p) - \tilde{\mu}_i}_2^2 $$
Here $\mathcal{S} = \set{p \mid p(x) \geq 0 \text{ and } \sum_{x \in \Omega} p(x) = n}$ is the set of (scaled) probability distributions over the domain $\Omega$.\footnote{When using unbounded DP, $n$ is sensitive and therefore we must estimate it.}   When $\tilde{\mu_i}$ are corrupted with i.i.d. Gaussian noise, this is exactly a maximum likelihood estimation problem \cite{mckenna2019graphical,mckenna2021winning,dporgblog}.  In general, convex optimization over the scaled probability simplex is intractable for the high-dimensional domains we are interested in.  \pgm overcomes this curse of dimensionality by exploiting the fact that the objective only depends on $p$ through its marginals.  The key observation is that one of the minimizers of this problem is a graphical model $\hat{p}_{\theta}$.  The parameters $\theta$ provide a compact representation of the distribution $p$ that we can optimize efficiently.

\paragraph*{\textbf{Junction Tree Size}}

The time and space complexity of \pgm depends on the measured marginal queries in a nuanced way, the main factor being the size of the junction tree implied by the measured marginal queries \cite{mckenna2021winning,mckenna2021relaxed}.  While understanding the junction tree construction is not necessary for this paper, it is important to note that \pgm exposes a callable function $\jtsize(r_1, \dots, r_k)$ that can be invoked to check how large a junction tree is.
\jtsize is measured in megabytes, and the runtime of distribution estimation is roughly proportional to this quantity.
If arbitrary marginals are measured, \jtsize can grow out of control, no longer fitting in memory, and leading to unacceptable runtime.

\paragraph*{\textbf{Synthetic Data Generation}}

Given an estimated model $\hat{p}$, \\\pgm implements a routine for generating synthetic tabular data that approximately matches the given distribution.  It achieves this with a randomized rounding procedure, which is a lower variance alternative to sampling from $\hat{p}$ \cite{mckenna2021winning}.


\section{Prior Work on Synthetic Data} \label{sec:prior}

In this section we survey the state of the field, describing basic elements of a good synthetic data mechanism, along with novelties of more sophisticated mechanisms.  We focus our attention on \emph{marginal-based approaches} to differentially private synthetic data in this section, as these have generally seen the most success in practical applications.  These mechanisms include \privbayes \cite{zhang2017privbayes}, \privbayespgm \cite{mckenna2019graphical}, \mwempgm \cite{mckenna2019graphical}, \mst \cite{mckenna2021winning}, \privsyn \cite{zhang2020privsyn}, \rap \cite{aydore2021differentially}, \gem \cite{liu2021iterative}, and \privmrf \cite{cai2021data}.  
We will begin with a formal problem statement:

\begin{problem}[Workload Error Minimization]
Given a workload $W$, our goal is to design an $(\epsilon, \delta)$-DP synthetic data mechanism $\mathcal{M} : \mathcal{D} \rightarrow \mathcal{D}$ such that the expected error defined in \cref{def:error} is minimized.
\end{problem}

\subsection{The Select-Measure-Generate Paradigm}

\begin{algorithm}[tb]
    \caption{\label{alg:mwem} \mwempgm}
\begin{algorithmic}
    \STATE {\bfseries Input:} Dataset $D$, workload $W$, privacy parameter $\rho$
    \STATE {\bfseries Output:} Synthetic Dataset $\hat{D}$
    \STATE {\bfseries Hyper-Parameters:} rounds $T=d$, budget split $\alpha=0.9$
    \STATE Initialize $\hat{p}_0 = \text{Uniform}[\mathcal{X}] $
    \STATE $\epsilon = \sqrt{8 (1-\alpha) \rho / T}$
    \STATE $\sigma = \sqrt{T / 2 \alpha \rho}$
    \FOR{$t = 1, \dots, T$}
        \STATE \textbf{select} $r_t \in W$ using exponential mechanism with $\epsilon$ budget:
$$ q_r(D) = \norm{ M_r(D) - M_r(\hat{p}_{t-1}) }_1 - n_r $$
        \STATE \textbf{measure} marginal on $C$:
$$\tilde{\mu}_t = M_{r_t}(D) + \mathcal{N}(0, \sigma^2 \mathbb{I})$$
        \STATE \textbf{estimate} data distribution using \pgm:
$$ \hat{p_t} = \argmin_{p \in S} \sum_{i=1}^t \norm{ M_{r_i}(p) - y_i }_2^2 $$
    \ENDFOR
    \STATE \textbf{generate} synthetic data $\hat{D}$ using \pgm:
    \STATE {\bfseries return} $\hat{D}$
\end{algorithmic}
\end{algorithm}

We begin by providing a broad overview of the basic approach employed by many differentially private mechanisms for synthetic data.  These mechanisms all fit naturally into the \emph{select-measure-generate} framework.  This framework represents a class of mechanisms which can naturally be broken up into 3 steps: (1) \emph{select} a set of queries, (2) \emph{measure} those queries using a noise-addition mechanism, and (3) \emph{generate} synthetic data that explains the noisy measurements well.  We consider iterative mechanisms that alternate between the select and measure step to be in this class as well.  Mechanisms within this class differ in their methodology for selecting queries, the noise mechanism used, and the approach to generating synthetic data from the noisy measurements.  


\mwempgm, shown in \cref{alg:mwem}, is one mechanism from this class that serves as a concrete example as well as the starting point for our improved mechanism, \aim.  As the name implies, \mwempgm is a scalable instantiation of the well-known \mwem algorithm \cite{hardt2010simple} for linear query answering, where the multiplicative weights (MW) step is replaced by a call to \pgm.  It is a greedy, iterative mechanism for workload-aware synthetic data generation, and there are several variants.  One variant is shown in \cref{alg:mwem}.  The mechanism begins by initializing an estimate of the joint distribution to be uniform over the data domain.  Then, it runs for $T$ rounds, and in each round it does three things: (1) selects (via the exponential mechanism) a marginal query that is poorly approximated under the current estimate, (2) measures the selected marginal using the Gaussian mechanism, and (3) estimates a new data distribution (using \pgm) that explains the noisy measurements well.  After $T$ rounds, the estimated distribution is used to generate synthetic tabular data.  \revision{
\mwempgm represents one mechansim from this broad class, but many others are very closely related to it.  In fact, \rap and \gem can both be seen as scalable instantiations of \mwem, that use different algorithms to estimate the data distribution instead of \pgm.  \privmrf is also closely related to \mwempgm (and uses \pgm), with some minor differences in design decisions in other parts of the algorithm.  Algorithms like \privbayes, \mst, and \privsyn are also conceptually similar to \mwempgm, as they attempt to select marginal queries that are poorly approximated under a simple model.   While all of these algorithms are conceptually similar, each one makes different design decisions that may have important performance implications in practice.
}
 In the subsequent subsections, we will characterize existing mechanisms in terms of how they approach these different aspects of the problem, \revision{and discuss some of the design decisions made by these mechansism.}

\subsection{Basic Elements of a Good Mechanism} \label{sec:basic}

In this section we outline some basic criteria reasonable mechanisms should satisfy to get good performance.  These recommendations primarily apply to the \emph{measure} step.

\paragraph*{\textbf{Measure Entire Marginals}}

Marginals are an appealing statistic to measure because every individual contributes a count of one to exactly one cell of the marginal.  As a result, we can measure every cell of $M_r(D)$ at the same privacy cost of measuring a single cell.  With a few exceptions \cite{aydore2021differentially,liu2021iterative,vietri2020new}, existing mechanisms utilize this property of marginals or can be extended to use it.  The alternative of measuring a single counting query at a time sacrifices utility unnecessarily.

\paragraph*{\textbf{Use Gaussian Noise.}}

Back of the envelope calculations reveal that if the number of measurements is greater than roughly $ \log{(1/\delta)} \\+ \epsilon$, which is often the case, then the standard deviation of the required Gaussian noise is lower than that of the Laplace noise.  Many newer mechanisms recognize this and use Gaussian noise, while older mechanisms were developed with Laplace noise, but can easily be adapted to use Gaussian noise instead.

\paragraph*{\textbf{Use Unbounded DP}}

For fixed $(\epsilon, \delta)$, the required noise magnitude is lower by a factor of $\sqrt{2}$ when using unbounded DP (add / remove one record) over bounded DP (modify one record).  
This is because the $L_2$ sensitivity of a marginal query $M_r$ is $1$ under unbounded DP, and $\sqrt{2}$ under bounded DP.  
We remark that these two different definitions of DP are qualitatively different, and because of that, the privacy parameters have different interpretations.  The $\sqrt{2}$ difference could be recovered in bounded DP by increasing the privacy budget appropriately. In some cases, the privacy model is imposed externally, in which case it is better if the mechanism naturally supports both bounded and unbounded DP.  When either privacy definition is acceptable, as in recent NIST competitions \cite{ridgeway2021challenge}, unbounded DP should be preferred.

\revision{
\paragraph*{\textbf{Devote more Budget to the Measure Step}}
For mechanisms that select marginal queries based on the data, the privacy budget must be split between the select step and the measure step.   A simple 50/50 split is usually suboptimal, and it is often better to allocate the majority of the privacy budget for the measure step.  Indeed, prior work has reported 10/90 splits to work well empirically in a variety of settings \cite{zhang2020privsyn,cai2021data}.  Intuitively, this uneven split makes sense because the statistics needed to select marginal queries are often coarser grained aggregations than the marginal queries themselves, and as a result are more robust to noise.

\paragraph*{\textbf{Summary}}
The implementation of \mwempgm in \cref{alg:mwem} gets these basic elements right.  This particular implementation of \mwempgm is new --- the original measured a single counting query per round, used Laplace noise, bounded DP, and an even select/measure budget split \cite{hardt2010simple,mckenna2019graphical}.  While the modifications made are simple, as we will show in \cref{sec:ablation}, they have a substantial influence on the performance of the mechanism in practice.
}

\subsection{Distinguishing Elements of Existing Work} \label{sec:distinguishing}

\begin{table}
\caption{\label{table:mechanisms} Taxonomy of select-measure-generate mechanisms.}
\resizebox{0.48\textwidth}{!}{
\begin{tabular}{lc||cccc}
Name & Year & Workload & Data & Budget & Efficiency \\
& & Aware & Aware & Aware & Aware \\\hline
\independent & - & & & & \check \\
\gaussian & - & \check & & &  \\
\privbayes \cite{zhang2017privbayes} & 2014 &  & \check & \check & \check \\
\hdmm \cite{mckenna2019graphical} & 2019 & \check & & & \\
\privbayespgm \cite{mckenna2019graphical} & 2019 &  & \check & \check & \check \\
\mwempgm \cite{mckenna2019graphical} & 2019 & \check & \check &  &  \\
\privsyn \cite{zhang2020privsyn} & 2020 &  & \check & \check & \check \\
\mst \cite{mckenna2021winning} & 2021 &  & \check &  & \check \\
\rap \cite{aydore2021differentially} & 2021 & \check & \check &  & \check \\
\gem \cite{liu2021iterative} & 2021 & \check & \check &  & \check \\
\privmrf \cite{cai2021data} & 2021 &  & \check & \check & \check \\
\rowcolor{lightgray} \aim [\small{This Work}] & 2022 & \check & \check & \check & \check \\
\end{tabular}}
\end{table}

Beyond the basics, different mechanisms exhibit different novelties, and understanding the design considerations underlying the existing work can be enlightening.  We provide a simple taxonomy of this space in \cref{table:mechanisms} in terms of four criteria: workload-, data-, budget-, and efficiency-awareness.  These characteristics primarily pertain to the \emph{select} step of each mechanism.

\paragraph*{\textbf{Workload-awareness}}

Different mechanisms select from a different set of candidate marginal queries.  \privbayes and \privmrf, for example, select from a particular subset of $k$-way marginals, determined from the data.  
Other mechanisms, like \mst and \privsyn, restrict the set of candidates to 2-way marginal queries.  
On the other end of the spectrum, the candidates considered by \mwempgm, \rap, and \gem, are exactly the marginal queries in the workload.  This is appealing, since these mechanisms will not waste the privacy budget to measure marginals that are not relevant to the workload. 



\paragraph*{\textbf{Data-awareness}}
Many mechanisms select marginal queries from a set of candidates based on the data, and are thus data-aware.  For example, \mwempgm selects marginal queries using the exponential mechanism with a quality score function that depends on the data.  \independent, \gaussian, and \hdmm are the exceptions, as they always select the same marginal queries no matter what the underlying data distribution is.  

\paragraph*{\textbf{Budget-awareness}}

Another aspect of different mechanisms is how well do they adapt to the privacy budget available.  Some mechanisms, like \privbayes, \privsyn, and \privmrf recognize that we can afford to measure more (or larger) marginals when the privacy budget is sufficiently large.  When the privacy budget is limited, these mechanisms recognize that fewer (and smaller) marginals should be measured instead.  In contrast, the number and size of the marginals selected by mechanisms like \mst, \mwempgm, \rap, and \gem does not depend on the privacy budget available.\footnote{The number of rounds to run \mwempgm, \rap, and \gem is a hyper-parameter, and the best setting of this hyper-parameter depends on the privacy budget available.}

\paragraph*{\textbf{Efficiency-awareness}}

Mechanisms that build on top of \pgm must take care when selecting measurements to ensure \jtsize remains sufficiently small to ensure computational tractability.   Among these, \privbayespgm, \mst, and \privmrf all have built-in heuristics in the selection criteria to ensure the selected marginal queries give rise to a tractable model.  \gaussian, \hdmm and \mwempgm have no such safeguards, and they can sometimes select marginal queries that lead to intractable models.  In the extreme case, when the workload is all 2-way marginals, \gaussian selects all 2-way marginals, the model required for \pgm explodes to the size of the entire domain, which is often intractable. 

Mechanisms that utilize different techniques for post-processing noisy marginals into synthetic data, like \privsyn, \rap, and \gem, do not have this limitation, and are free to select from a wider collection of marginals.  While these methods do not suffer from this particular limitation of \pgm, they have other pros and cons which were surveyed in a recent article~\cite{dporgblog}.  

\paragraph*{\textbf{Summary}}

With the exception of our new mechanism \aim, no mechanism listed in \cref{table:mechanisms} is aware of all four factors we discussed.  Mechanisms that do not have four checkmarks in \cref{table:mechanisms} are not necessarily bad, but there are clear ways in which they can be improved.  Conversely, mechanisms that have more checkmarks than other mechanisms are not necessarily better.  For example, \independent only has one checkmark, but as we show in \cref{sec:experiments}, it sometimes outperforms mechanisms with three checkmarks.

\subsection{Other Design Considerations} \label{sec:otherconsiderations}

Beyond these four characteristics summarized in the previous section, different methods make different design decisions that are relevant to mechanism performance, but do not correspond to the four criteria discussed in the previous section.  In this section, we summarize some of those additional design considerations.

\paragraph*{\textbf{Selection method}}

Some mechanisms select marginals to measure in a \emph{batch}, while other mechanisms select them \emph{iteratively}.  Generally speaking, iterative methods like \mwempgm, \rap, \gem, and \privmrf are preferable to batch methods, because the selected marginals will capture important information about the distribution that was not effectively captured by the previously measured marginals.  On the other hand, \privbayes, \mst, and \privsyn select all the marginals before measuring any of them.  It is not difficult to construct examples where a batch method like \privsyn has suboptimal behavior.  For example, suppose the data contains three perfectly correlated attributes.  We can expect iterative methods to capture the distribution after measuring any two 2-way marginals.  On the other hand, a batch method like \privsyn will determine that all three 2-way marginals need to be measured.

\paragraph*{\textbf{Budget split}}

Every mechanism in this discussion, except for \privsyn, splits the privacy budget equally among selected marginals. This is a simple and natural thing to do, but it does not account for the fact that larger marginals have smaller counts that are less robust to noise, requiring a larger fraction of the privacy budget to answer accurately.  \privsyn provides a simple formula for dividing privacy budget among marginals of different sizes, but this approach is inherently tied to their batch selection methodology. 
It is much less clear how to divide the privacy budget within a mechanism that uses an iterative selection procedure.

\paragraph*{\textbf{Hyperparameters}}

All mechanisms have some hyperparameters than can be tuned to affect the behavior of the mechanism.  Mechanisms like \privbayes, \mst, \privsyn, and \privmrf have reasonable default values for these hyperparameters, and these mechanisms can be expected to work well out of the box.  
On the other hand, \mwempgm, \rap, and \gem have to tune the number of rounds to run, and it is not obvious how to select this a priori.  While the open source implementations may include a default value, the experiments conducted in the respective papers did not use these default values, in favor of non-privately optimizing over this hyper-parameter for each dataset and privacy level considered \cite{aydore2021differentially,liu2021iterative}.


\begin{algorithm}[tb]
    \caption{\label{alg:aim} \aim: An Adaptive and Iterative Mechanism}
\begin{algorithmic}[1]
    \STATE {\bfseries Input:} Dataset $D$, workload $W$, privacy parameter $\rho$
    \STATE {\bfseries Output:} Synthetic Dataset $\hat{D}$
    \STATE {\bfseries Hyper-Parameters:} \textsf{MAX-SIZE}=80MB, $T=16d$, $\alpha=0.9$
    \STATE $\sigma_0 = \sqrt{T / (2 \: \alpha \: \rho)}$
    \STATE $\rho_{used} = 0$
    \STATE $t = 0$
    \STATE Initialize $\hat{p}_t$ using \cref{alg:init} \label{line:init}
    \STATE $w_r = \sum_{s \in W} c_s \mid r \cap s \mid$ 
    \STATE $\sigma_{t+1} \leftarrow \sigma_0$ \:\:\: $\epsilon_{t+1} \leftarrow \sqrt{8 (1 - \alpha) \rho / T}$
    \WHILE{$\rho_{used} < \rho$} \label{line:loop}
        \STATE $t = t+1$
        \STATE $\rho_{used} \leftarrow \rho_{used} + \frac{1}{8} \epsilon_t^2 + \frac{1}{2\sigma_t^2}$ \label{line:budget}
        \STATE $C_t =\set{ r_t \in W_+ \mid \textsf{JT-SIZE}(r_1, \dots, r_t)) \leq \frac{\rho_{used}}{\rho} \cdot \textsf{MAX-SIZE} }$ \label{line:candidates}
        \STATE \textbf{select} $r_t \in C_t $ using the exponential mechanism with:
$$ q_r(D) = w_r \Big( \norm{ M_r(D) - M_r(\hat{p}_{t-1}) }_1 - \sqrt{2/\pi} \cdot \sigma_t \cdot n_r \Big) $$ \label{line:quality}
        \STATE \textbf{measure} marginal on $r_t$:
$$\tilde{y}_t = M_{r_t}(D) + \mathcal{N}(0, \sigma_t^2 \mathbb{I})$$
        \STATE \textbf{estimate} data distribution using \pgm:
$$ \hat{p}_t = \argmin_{p \in S} \sum_{i=1}^t \frac{1}{\sigma_i} \norm{ M_{r_i}(p) - \tilde{y}_i }_2^2 $$
        \STATE anneal $\epsilon_{t+1}$ and $\sigma_{t+1}$ using \cref{alg:anneal} \label{line:anneal}
    \ENDWHILE
    \STATE \textbf{generate} synthetic data $\hat{D}$ from $\hat{p}_t$ using \pgm
    \STATE {\bfseries return} $\hat{D}$
\end{algorithmic}
\end{algorithm}

\begin{algorithm}[tb]
    \caption{\label{alg:init} Initialize $p_t$ (subroutine of \cref{alg:aim})}
\begin{algorithmic}[1]
\FOR{$r \in \set{r \in W_+ \mid |r| = 1}$}
    \STATE $t = t+1$ \:\:\: $\sigma_t \leftarrow \sigma_0 $ \:\:\: $r_t \leftarrow r$
    \STATE $\tilde{y}_t = M_r(D) + \mathcal{N}(0, \sigma_t^2 \mathbb{I})$
    \STATE $\rho_{used} \leftarrow \rho_{used} + \frac{1}{2 \sigma_t^2}$
\ENDFOR
\STATE $ \hat{p_t} = \argmin_{p \in S} \sum_{i=1}^t \frac{1}{\sigma_i} \norm{ M_{r_i}(p) - \tilde{y}_i }_2^2 $
\end{algorithmic}
\end{algorithm}

\begin{algorithm}[tb]
    \caption{\label{alg:anneal} Budget annealing (subroutine of \cref{alg:aim})}
\begin{algorithmic}[1]
    \IF{$\norm{M_{r_t}(\hat{p}_t) - M_{r_t}(\hat{p}_{t-1})}_1 \leq \sqrt{2/\pi} \cdot \sigma_t \cdot n_{r_t}$}
        \STATE $\epsilon_{t+1} \leftarrow 2 \cdot \epsilon_t$ 
        \STATE $\sigma_{t+1} \leftarrow \sigma_t / 2$
    \ELSE
        \STATE $\epsilon_{t+1} \leftarrow \epsilon_t$
        \STATE $\sigma_{t+1} \leftarrow \sigma_t$
    \ENDIF

    \IF{$(\rho - \rho_{used}) \leq 2 \big( \frac{1}{2\sigma_{t+1}^2} + \frac{1}{8} \epsilon_{t+1}^2 \big)$} \label{line:endcond}
        \STATE $\epsilon_{t+1} = \sqrt{8 \cdot (1-\alpha) \cdot (\rho - \rho_{used})}$
        \STATE $\sigma_{t+1} = \sqrt{1 / (2 \cdot \alpha \cdot (\rho - \rho_{used}))}$
    \ENDIF
\end{algorithmic}
\end{algorithm}

\section{AIM: An Adaptive and Iterative Mechanism for Synthetic Data} \label{sec:aim}



While \mwempgm is a simple and intuitive algorithm, it leaves significant room for improvement, \revision{even after getting the basic elements right}.  Our new mechanism, \aim, is presented in \cref{alg:aim}.  In this section, we describe the differences between \mwempgm and \aim, the justifications for the relevant design decisions, as well as prove the privacy of \aim. 

\paragraph*{\textbf{Intelligent Initialization.}}

In Line \ref{line:init} of \aim, we spend a small fraction of the privacy budget to measure $1$-way marginals in the set of candidates.  Estimating $\hat{p}$ from these noisy marginals gives rise to an \emph{independent} model where all $1$-way marginals are preserved well, and higher-order marginals can be estimated under an independence assumption.  
\revision{Intuitively, this feature of \aim is justified by the fact that \mwempgm tends to select marginal queries covering disjoint attribute subsets in the first few rounds in an attempt to correctly preserve the 1-way marginal distributions.  By measuring all 1-way marginals immediately instead, we are saving the privacy budget that would otherwise be spent to select these marginal queries.}


\paragraph*{\textbf{New Candidates.}}

In Line \ref{line:candidates} of \aim, we make two notable modifications to the candidate set that serve different purposes.  Specifically, the set of candidates is a carefully chosen subset of the marginal queries in the \emph{downward closure} of the workload.  The downward closure of the workload is the set of marginal queries whose attribute sets are subsets of some marginal query in the workload, i.e., $W_+ = \set{r \mid r \subseteq s, s \in W}$.  

Using the downward closure is based on the observation that marginals with many attributes have low counts, and answering them directly with a noise addition mechanism may not provide an acceptable signal to noise ratio.  In these situations, it may be better to answer lower-dimensional marginals, as these tend to exhibit a better signal to noise ratio, while still being useful to estimate the higher-dimensional marginals in the workload.

We filter candidates from this set that do not meet a specific model capacity requirement.  Specifically, the set will only consist of candidates that, if selected, will lead to a \jtsize below a prespecified limit (the default is 80 \textsf{MB}).  This ensures that \aim will never select candidates that lead to an intractable model, and hence allows the mechanism to execute consistently with a predictable memory footprint and runtime.  

\paragraph*{\textbf{Better Selection Criteria.}}

In Line \ref{line:quality} of \aim, we make two modifications to the quality score function for marginal query selection to better reflect the utility we expect from measuring the selected marginal.  In particular, our new quality score function is
\begin{equation} \label{eq:quality}
q_r(D) = w_r \big( \norm{M_r(D) - M_r(p_{t-1})}_1 - \sqrt{2/\pi} \cdot \sigma_t \cdot n_r \big),
\end{equation}
which differs from \mwempgm's quality score function $q_r(D) =\norm{M_r(D) - M_r(p_{t-1})} - n_r$ in two ways.  

First, the expression inside parentheses can be interpreted as the \emph{expected improvement} in $L_1$ error we can expect by measuring that marginal.  It consists of two terms: the $L_1$ error under the current model minus the expected $L_1$ error if it is measured at the current noise level (\cref{thm:halfnorm} in the full paper \cite{mckenna2022aim}).   
Compared to the quality score function in \mwempgm, this quality score function penalizes larger marginals to a much more significant degree, since $\sigma_t \gg 1$ in most cases.  Moreover, this modification makes the selection criteria ``budget-adaptive'', since it recognizes that we can afford to measure larger marginals when $\sigma_t$ is smaller, and we should prefer smaller marginals when $\sigma_t$ is larger.

Second, we give different marginal queries different weights to capture how relevant they are to the workload.  
In particular, we weight the quality score function for a marginal query $r$ using the formula $ w_r = \sum_{s \in W} c_s \mid r \cap s \mid $, as this captures the degree to which the marginal queries in the workload overlap with $r$.  In general, this weighting scheme places more weight on marginals involving more attributes.  
Note that now the sensitivity of $q_r$ is $w_r$ rather than $1$.  \revision{ Thus, to apply the exponential mechanism to select a candidate, we use $\Delta_t = \max_{r \in C_t} w_r$.  
A nice property of using $w_r$ as a multiplicative weight is a certain invariance to how the workload is represented: in particular, the behavior of \aim is identical in the  two cases where (1) two copies of a marginal query are included in the workload, (2) the marginal query appears once with a weight of two. This is not true of \mwempgm, which generally has different behavior based on how the workload is represented. }

This quality score function exhibits an interesting trade-off: the penalty term $\sqrt{2/\pi} \sigma_t n_r$ discourages marginals with more cells, while the weight $w_r$ favors marginals with more attributes.  However, if the inner expression is negative, then the larger weight will make it more negative, and much less likely to be selected.  

\paragraph*{\textbf{Adaptive Rounds and Budget Split.}}

In Lines \ref{line:budget} and \ref{line:anneal} of \aim, we introduce logic to modify the per-round privacy budget as execution progresses, and as a result, eliminate the need to provide the number of rounds up front.  This makes \aim hyper-parameter free, relieving practitioners from that often overlooked burden. 

Specifically, we use a simple annealing procedure (\cref{alg:anneal}) that gradually increases the budget per round when an insufficient amount of information is learned at the current per-round budget.  The annealing condition is activated if the difference between  $M_{r_t}(\hat{p}_t)$ and $M_{r_t}(\hat{p}_{t-1})$ is small, which indicates that not much information was learned in the previous round.  If it is satisfied, then $\epsilon_t$ for is doubled, while $\sigma_t$ is cut in half.  

This check can pass for two reasons: (1) there were no good candidates (all scores are low in \cref{eq:quality}) in which case increasing $\sigma_t$ will make more candidates good, and (2) there were good candidates, but they were not selected because there was too much noise in the select step, which can be remedied by increasing $\epsilon_t$.
The precise annealing threshold used is $\sqrt{2/\pi} \cdot \sigma_t \cdot n_{r_t}$, which is the expected error of the noisy marginal, and an approximation for the expected error of $\hat{p}_t$ on marginal $r$.  When the available privacy budget is small, this condition will be activated more frequently, and as a result, \aim will run for fewer rounds.  Conversely, when the available privacy budget is large, \aim will run for many rounds before this condition activates.

As $\sigma_t$ decreases throughout execution, quality scores generally increase, and it has the effect of ``unlocking'' new candidates that previously had negative quality scores.   We initialize $\sigma_t$ and $\epsilon_t$ conservatively, assuming the mechanism will be run for $T=16d$ rounds.  This is an upper bound on the number of rounds that \aim will run, but in practice the number of rounds will be much less.  


\paragraph*{\textbf{Privacy Analysis.}}

The privacy analysis of \aim utilizes the notion of a \emph{privacy filter} \cite{rogers2016privacy,cesar2021bounding,feldman2021individual}, and the algorithm runs until the realized privacy budget spent matches the total privacy budget available, $\rho$.  To ensure that the budget is not over-spent, there is a special condition (Line \ref{line:endcond} in \cref{alg:anneal}) that checks if the remaining budget is insufficient for two rounds at the current $\epsilon_t$ and $\sigma_t$ parameters.  If this condition is satisfied, $\epsilon_t$ and $\sigma_t$ are set to use up all of the remaining budget in one final round of execution.  

\begin{theorem}
For any $T \geq d$, $\alpha \in (0,1)$, $\rho \geq 0$, \aim satisfies $\rho$-zCDP.
\end{theorem}

\begin{proof}
There are three steps in \aim that depend on the sensitive data: initialization, selection, and measurement.  The initialization step satisfies $\rho_0$-zCDP for
$\rho_0 = | \set{r \in W_+ \mid |r|=1} | / 2 \sigma_0^2
\leq d / 2 \sigma_0^2
= 2 \alpha d \rho / 2 T
\leq \rho
$.
For this step, all we need is that the privacy budget is not over-spent.  The remainder of \aim runs until the budget is consumed.  Each step of \aim involves one invocation of the exponential mechanism, and one invocation of the Gaussian mechanism.  By \cref{prop:expprivacy,prop:rdpgauss,prop:composition}, round $t$ of \aim is $\rho_t$-zCDP for $ \rho_t = \frac{1}{8} \epsilon_t^2/8 + 1/2 \sigma_t^2 $.
Note that at round $t$, $\rho_{used} = \sum_{i=0}^t \rho_i $, and by Theorem 3.1 of \cite{feldman2021individual}, it suffices to show that $\rho_{used}$ never exceeds $\rho$.  There are two cases to consider: the condition in Line \ref{line:endcond} of \cref{alg:anneal} is either true or false.  If it is true, then we know after round $t$ that $ \rho - \rho_{used} \geq 2 \rho_{t+1} $, i.e., the remaining budget is enough to run round $t+1$ without over-spending the budget.  If it is false, then we modify $\epsilon_{t+1}$ and $ \rho_{t+1}$ to exactly use up the remaining budget.  Specifically,
$\rho_{t+1} = 8 (1-\alpha) (\rho - \rho_{used}) / 8 + 2 \alpha (\rho - \rho_{used}) / 2 = \rho - \rho_{used}$.
As a result, when the condition is true, $\rho_{used}$ at time $t+1$ is exactly $\rho$, and after that iteration, the main loop of \aim terminates.  The remainder of the mechanism does not access the data.
\end{proof}


\section{Uncertainty Quantification} \label{sec:uncertainty}

We now propose a solution to the uncertainty quantification problem for \aim.  Our method uses information from \emph{both} the noisy marginals, measured with Gaussian noise, and the marginal queries selected by the exponential mechanism.  The method does not require additional privacy budget, as it quantifies uncertainty only by analyzing the private outputs of \aim. We give guarantees for marginals in the (downward closure of the) workload, which is exactly the set of marginals the analyst cares about. Providing guarantees for marginals outside this set is an area for future work. 

We break our analysis up into two cases: the ``easy'' case, where we have access to unbiased answers for a particular marginal, and the ``hard'' case, where we do not.  In both cases, we identify an \emph{estimator} for a marginal whose error we can bound with high probability.  Then, we connect the error of this estimator to the error of the synthetic data by invoking the triangle inequality.  
Proofs of all statements in this section appear in the full paper \cite{mckenna2022aim}.

\paragraph*{\textbf{The Easy Case: Supported Marginal Queries}}

A marginal query r is ``supported'' whenever $r \subseteq r_t$ for some $t$.  In this case, we can readily obtain an unbiased estimate of $M_r(D)$ from $y_t$, and analytically derive the variance of that estimate.  If there are multiple $t$ satisfying the condition above, we have multiple estimates we can use to reduce the variance.  We can combine these independent estimates to obtain a \emph{weighted average estimator}:

\begin{restatable}[Weighted Average Estimator]{theorem}{estimator} \label{thm:estimator1}
Let $r_1, \dots, r_t$ and $y_1, \dots, y_t$ be as defined in \cref{alg:aim}, and let $ R = \set{ r_1, \dots, r_t } $.  For any $r \in R_+$, there is an (unbiased) estimator $\bar{y}_r = f_r(y_1, \dots, y_t)$ such that:
$$ \bar{y}_r \sim \mathcal{N}(M_r(D), \bar{\sigma}_r^2 \mathbb{I}) \:\: \text{ where } \:\: \bar{\sigma}_r^2 = \Big[\sum_{\substack{i=1 \\ r \subseteq r_i}}^t \frac{n_r}{n_{r_i} \sigma_i^2}\Big]^{-1}, $$
\end{restatable}
While this is not the only (or best) estimator to use \cite{ding2011differentially}, 
the simplicity allows us to easily bound its error, as we show in \cref{thm:supported1}. 
%
%
\begin{restatable}[Confidence Bound]{theorem}{supported} \label{thm:supported1}
Let $\bar{y}_r$ be the estimator from \cref{thm:estimator1}.  Then,
for any $\lambda \geq 0$, with probability at least $1 - \exp{(-\lambda^2)}$:
$$\norm{M_r(D) - \bar{y}_r}_1 \leq \sqrt{2 \log{2}} \bar{\sigma}_r n_r + \lambda \bar{\sigma}_r \sqrt{2 n_r}$$
\end{restatable}
Note that \cref{thm:supported1} gives a guarantee on the error of $\bar{y}_r$, but we are ultimately interested in the error of $\hat{D}$.  Fortunately, it easy easy to relate the two by using the triangle inequality: 
\begin{corollary} \label{cor:supported1}
Let $\hat{D}$ be any synthetic dataset, and let $\bar{y}_r$ be the estimator from \cref{thm:estimator1}.  Then with probability at least $1 - \exp{(-\lambda^2)}$:
$$\norm{M_r(D) - M_r(\hat{D})}_1 \leq \norm{M_r(\hat{D}) - \bar{y}_r}_1 + \sqrt{2 \log{2}} \bar{\sigma}_r n_r + \lambda \bar{\sigma}_r \sqrt{2 n_r}$$
\end{corollary}


The LHS is what we are interested in bounding, and we can readily compute the RHS from the output of \aim.  The RHS is a random quantity that, with the stated probability, upper bounds the error.
When we plug in the realized values we get a concrete numerical bound that can be interpreted as a (one-sided) confidence interval.  
In general, we expect $M_r(\hat{D})$ to be close to $\bar{y}_r$, so the error bound for $\hat{D}$ will not be that much larger than that of $\bar{y}_r$.\footnote{\revision{From prior experience, we might expect the error of $\hat{D}$ to be \emph{lower} than the error of $\bar{y}_r$ \cite{nikolov2013geometry,mckenna2019graphical}, so we are paying for this difference by increasing the error bound when we might hope to save instead.  Unfortunately, this intuition does not lend itself to a clear analysis that provides better guarantees.}}

\paragraph*{\textbf{The Hard Case: Unsupported Marginal Queries}}

We now shift our attention to the hard case, providing guarantees about the error of different marginals even for unsupported marginal queries (those not selected during execution of \aim).  This problem is significantly more challenging.  Our key insight is that marginal queries \emph{not selected} have relatively low error compared to the marginal queries that were selected.  We can easily bound the error of selected queries and relate that to non-selected queries by utilizing the guarantees of the exponential mechanism.  In \cref{thm:unsupported1} below, we provide expressions that capture the uncertainty of these marginals with respect to $\hat{p}_{t-1}$, the iterates of \aim.

\begin{restatable}[Confidence Bound]{theorem}{unsupported} \label{thm:unsupported1}
Let $\sigma_t, \epsilon_t, r_t, \tilde{y_t}, C_t, \hat{p}_t$ be as defined in \cref{alg:aim}, and let $\Delta_t = \max_{r \in C_t} w_r$.  For all $r \in C_t$,
 with probability at least $1 - e^{-\lambda_1^2/2} - e^{-\lambda_2}$:
\begin{align*}
\norm{M_r(D) - M_r(\hat{p}_{t-1})}_1 \leq w_r^{-1} \big(B_r + \lambda_1 \sigma_t \sqrt{n_{r_t}} + \lambda_2 \frac{2 \Delta_t}{\epsilon_t}\big)
\end{align*}
where $B_r$ is equal to:
$$w_{r_t} \underbrace{\norm{M_{r_t}(\hat{p}_{t-1})-y_t}_1}_{\text{estimated error on } r_t} + \underbrace{\sqrt{2/\pi} \sigma_t \big(w_r n_r - w_{r_t} n_{r_t} \big)}_{\substack{\text{relationship to} \\ \text{non-selected candidates}}} + \underbrace{\frac{2\Delta_t}{\epsilon_t} \log{(|C_t|)}}_{\substack{\text{uncertainty from} \\ \text{exponential mech.}}}  $$
\end{restatable}

We can readily compute $B_r$ from the output of \aim, and use it to provide a bound on error in the form of a one-sided confidence interval that captures the true error with high probability.  While these error bounds are expressed with respect to $\hat{p}_{t-1}$, they can readily be extended to give a guarantee with respect to $\hat{D}$.  

\begin{corollary} \label{cor:unsupported1}
Let $\hat{D}$ be any synthetic dataset, and let $B_r$ be as defined in \cref{thm:unsupported1}. Then with probability at least $1 - e^{-\lambda_1^2/2} - e^{-\lambda_2}$:
\begin{align*}
&\norm{M_r(D) - M_r(\hat{D})}_1 \\ &\leq \norm{M_r(\hat{D}) - M_r(\hat{p}_{t-1})}_1 + w_r^{-1} \big(B_r + \lambda_1 \sigma_t \sqrt{n_{r_t}} + \lambda_2 \frac{2 \Delta_t}{\epsilon_t}\big)
\end{align*}
\end{corollary}

Again, the LHS is what we are interested in bounding, and we can compute the RHS from the output of \aim.  We expect $\hat{p}_{t-1}$ to be reasonably close to $\hat{D}$, especially when $t$ is larger, so this bound will often be comparable to the original bound on $\hat{p}_{t-1}$.  

\paragraph*{\textbf{Putting it Together}}

We've provided guarantees for both supported and unsupported marginals.  The guarantees for unsupported marginals also apply for supported marginals, although we generally expect them to be looser.  In addition, there is one guarantee \emph{for each round of \aim}.  It is tempting to use the bound that provides the smallest estimate, although unfortunately doing this invalidates the bound.  To ensure a valid bound, we must pick only one round, and that cannot be decided based on the value of the bound.  A natural choice is to use only the last round, for three reasons: (1) $\sigma_t$ is smallest and $\epsilon_t$ is largest in that round, (2) the error of $\hat{p}_{t}$ generally goes down with $t$, and (3) the distance between $\hat{p}_t$ and $\hat{D}$ should be the smallest in the last round.  However, there may be some marginal queries which were not in the candidate set for that round.  To bound the error on these marginals, we use the last round where that marginal query was in the candidate set.  


\section{Experiments} \label{sec:experiments}

In this section we empirically evaluate \aim, comparing it to a collection of state-of-the-art mechanisms and baseline mechanisms for a variety of workloads, datasets, and privacy levels.

\subsection{Experimental Setup}

\paragraph*{\textbf{Datasets}}

Our evaluation includes datasets with varying size and dimensionality.  We describe our exact pre-processing scheme in the full paper \cite{mckenna2022aim}, and summarize the pre-processed datasets and their characteristics in the table below.

\begin{table}[H]
\caption{ \label{table:datasets} Summary of datasets used in the experiments.}
\begin{tabular}{l|cccc}
\multirow{2}{*}{Dataset} & \multirow{2}{*}{Records} & \multirow{2}{*}{Dimensions} & Min/Max & Total \\
& & & Domains & Domain Size \\\hline
\adult \cite{kohavi1996scaling} & 48842 & 15 & 2--42  & $4 \times 10^{16}$  \\
\salary \cite{hay2016principled} & 135727 & 9 & 3--501 & $1 \times 10^{13}$ \\
\msnbc \cite{cadez2000visualization} & 989818 & 16 & 18 & $1 \times 10^{20}$ \\
\fire \cite{ridgeway2021challenge} & 305119 & 15 & 2--46 & $4 \times 10^{15}$ \\
\nltcs \cite{nltcs} & 21574 & 16 & 2 & $7 \times 10^4$ \\
\titanic \cite{titanic} & 1304 & 9 & 2--91 & $9 \times 10^{7}$ \\
\end{tabular}
\end{table}

\paragraph*{\textbf{Workloads}}

We consider 3 workloads for each dataset, \general, \target, and \weighted.  Each workload contains a collection of 3-way marginal queries. The \general workload contains queries for \emph{all} 3-way marginals.  The \target workload contains queries for all 3-way marginals involving some specified \emph{target} attribute.  For the \adult and \titanic datasets, these are the \textsc{income>50K} attribute and the \textsc{Survived} attribute, as those correspond to the attributes we are trying to predict for those datasets.  For the other datasets, the target attribute is chosen uniformly at random.  The \weighted workload contains a collection of 3-way marginal queries \emph{biased} towards certain attributes and attribute combinations.  In particular, each attribute is assigned a weight sampled from a squared exponential distribution.  256 triples of attributes are sampled with probability proportional to the product of their weights.  This results in workloads where certain attributes appear far more frequently than others, and is intended to capture the situation where analysts focus on a small number of interesting attributes.  \revision{In \cref{sec:2way}, we provide results on a fourth workload, \workload{all-2way} as well.} All randomness in the construction of the workload was done with a fixed random seed, to ensure that the workloads remain the same across executions of different mechanisms and parameter settings.

\paragraph*{\textbf{Mechanisms}}

We compare against both workload-agnostic and workload-aware mechanisms in this section.  The workload-agnostic mechanisms we consider are \privbayespgm, \mst, \privmrf.  The workload-aware mechanisms we consider are \mwempgm, \rap, \gem, and \aim.  We set the hyper-parameters of every mechanism to default values available in their open source implementations.
\revision{While these default hyper-parameters may be suboptimal, we conducted sensitivity experiments in \cref{sec:sensitivity} to evaluate the impact of hyper-parameters on the performance of competing mechanisms, and found that the improvement in utility from optimizing hyper-parameters is outweighed by the cost to privacy needed to run an appropriate DP hyper-parameter selection mechanism.}
We also consider baseline mechanisms: \independent and \gaussian.  The former measures all 1-way marginals using the Gaussian mechanism, and generates synthetic data using an independence assumption.  The latter answers all queries in the workload using the Gaussian mechanism (using the optimal privacy budget allocation described in \cite{zhang2020privsyn}).  Note that this mechanism \emph{does not} generate synthetic data, only query answers.  

\paragraph*{\textbf{Privacy Budgets}}

We consider a wide range of privacy parameters, varying $\epsilon \in [0.01, 100.0]$ and setting $\delta = 10^{-9}$.  The most practical regime is $\epsilon \in [0.1, 10.0]$, but mechanism behavior at the extremes can be enlightening so we include them as well.

\paragraph*{\textbf{Evaluation.}}

For each dataset, workload, and $\epsilon$, we run each mechanism for 5 trials, and measure the workload error from \cref{def:error}.  We report the average workload error across the five trials, along with error bars corresponding to the minimum and maximum workload error observed across the five trials.  \revision{In \cref{sec:metrics}, we evaluate error using $L_2$ and $L_{\inf}$ error metrics as well.}

\paragraph*{\textbf{Runtime Environment.}}

We ran most experiments on a single core of a compute cluster with a 4 GB memory limit and a 24 hour time limit.\footnote{These experiments usually completed in well under the time limit.}  These resources were not sufficient to run \privmrf or \rap, so we utilized different machines to run those mechanisms.  \privmrf requires a GPU to run, so we used one node a different compute cluster, which has a Nvidia GeForce RTX 2080 Ti GPU.  \rap required significant memory resources, so we ran those experiments on a machine with 16 cores and 64 GB of RAM.

\begin{figure*}[t!]
\centering
\includegraphics[width=\textwidth]{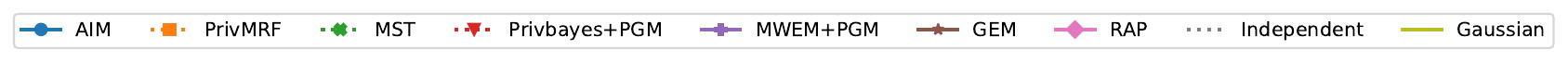}
\includegraphics[width=0.366\textwidth]{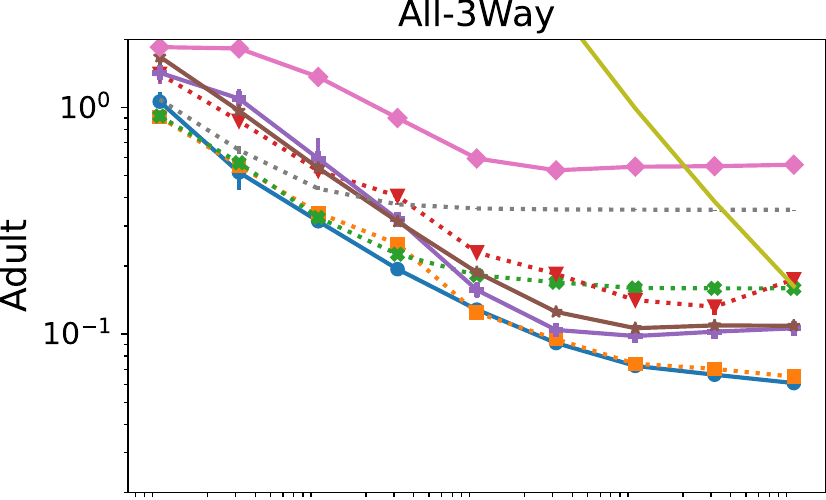}
\includegraphics[width=0.312\textwidth]{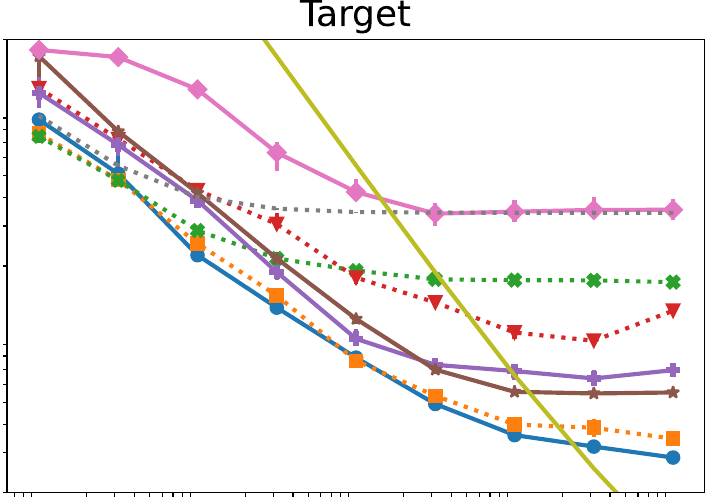}
\includegraphics[width=0.312\textwidth]{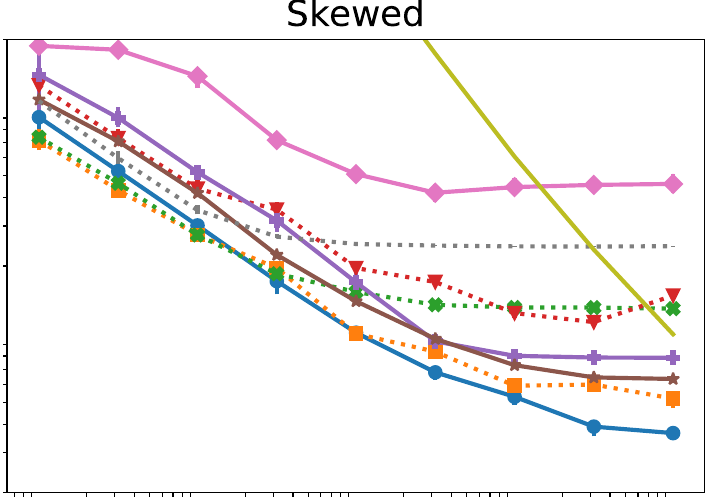}
\includegraphics[width=0.366\textwidth]{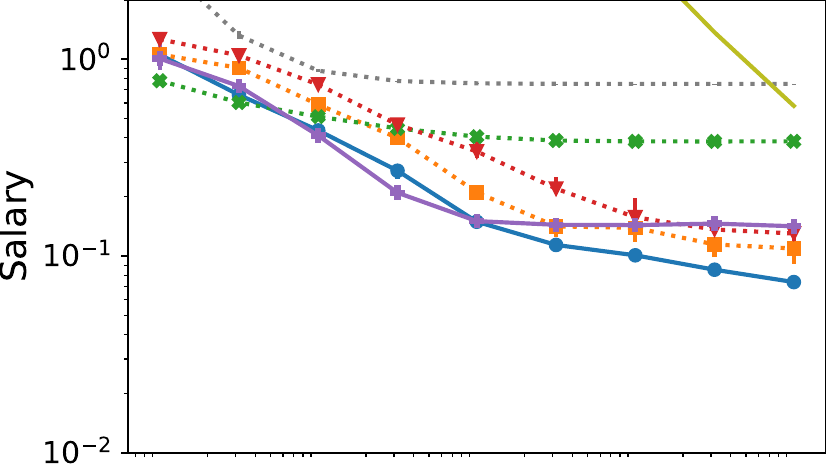}
\includegraphics[width=0.312\textwidth]{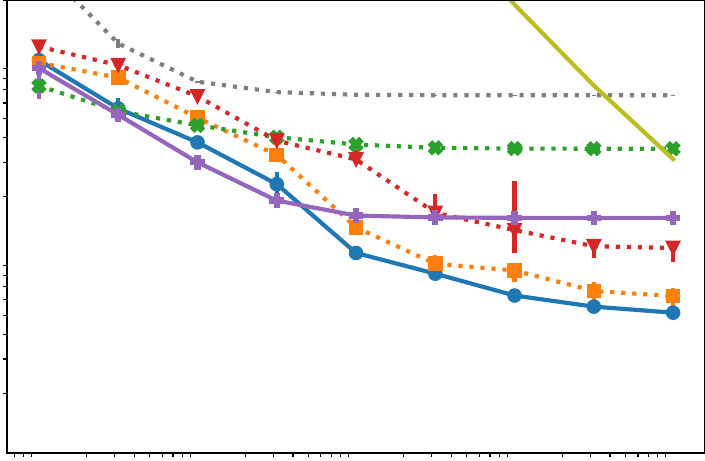}
\includegraphics[width=0.312\textwidth]{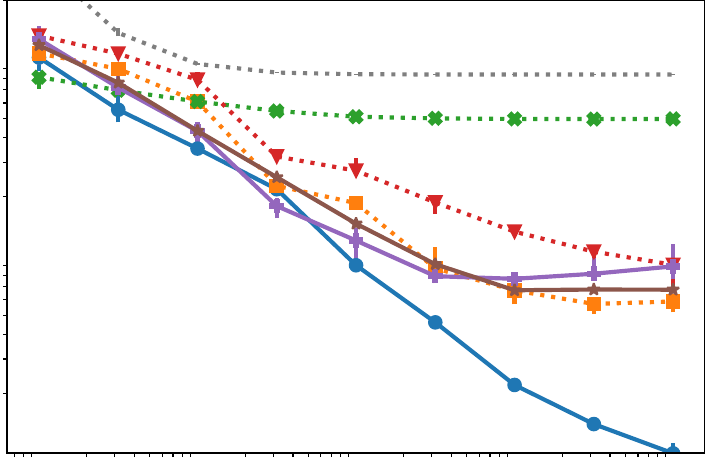}
\includegraphics[width=0.366\textwidth]{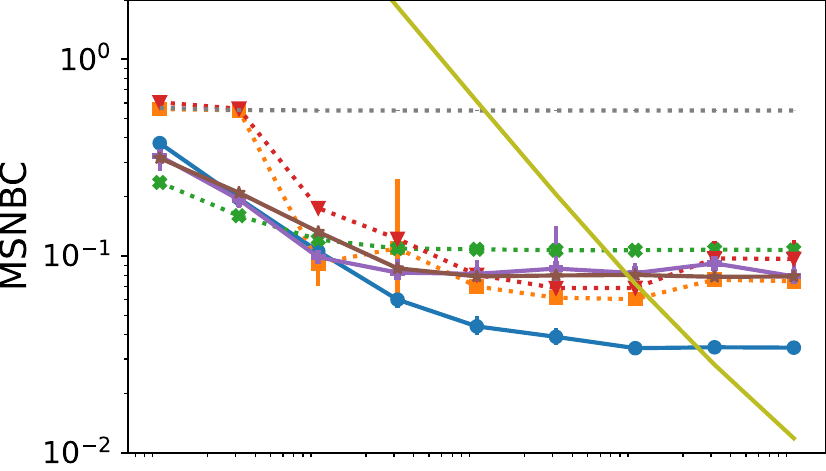}
\includegraphics[width=0.312\textwidth]{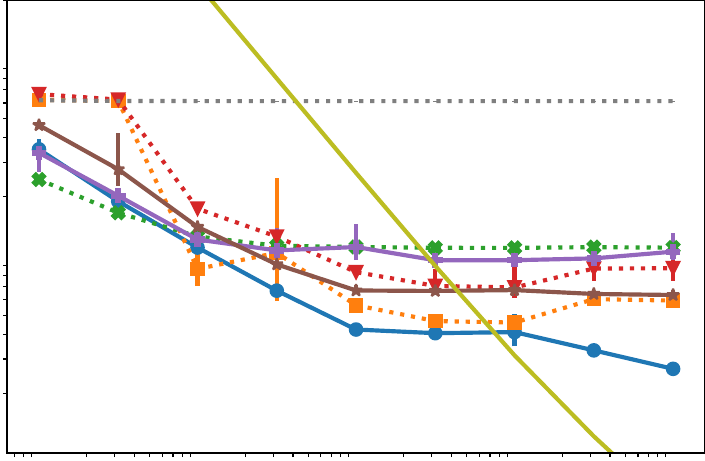}
\includegraphics[width=0.312\textwidth]{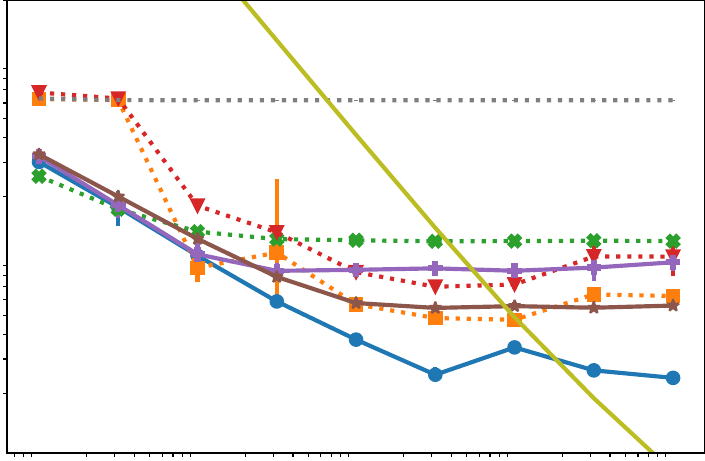}
\includegraphics[width=0.366\textwidth]{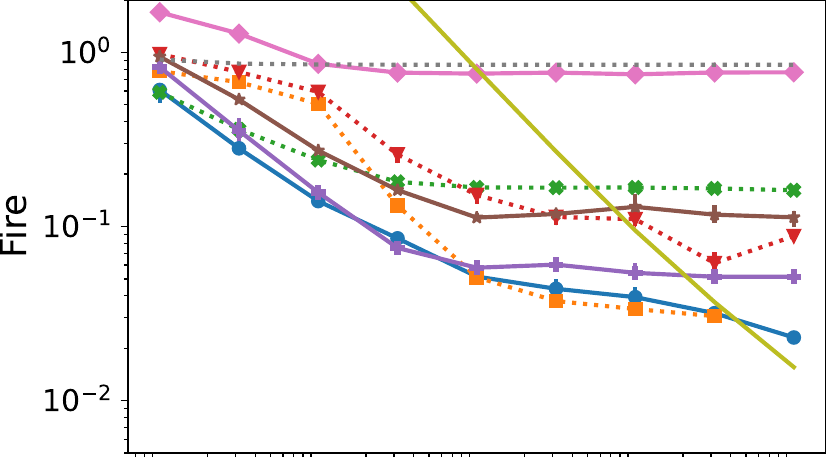}
\includegraphics[width=0.312\textwidth]{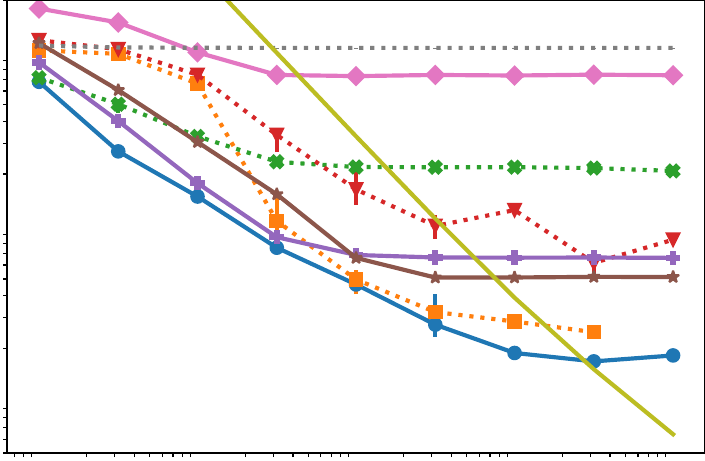}
\includegraphics[width=0.312\textwidth]{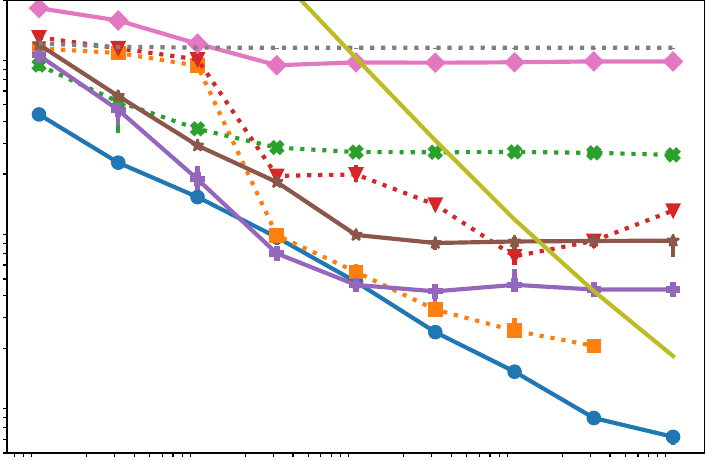}
\includegraphics[width=0.366\textwidth]{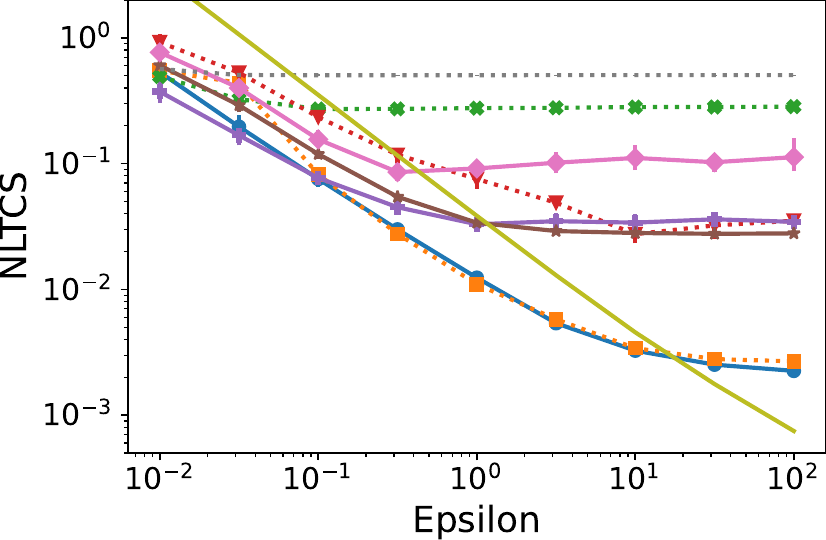}
\includegraphics[width=0.312\textwidth]{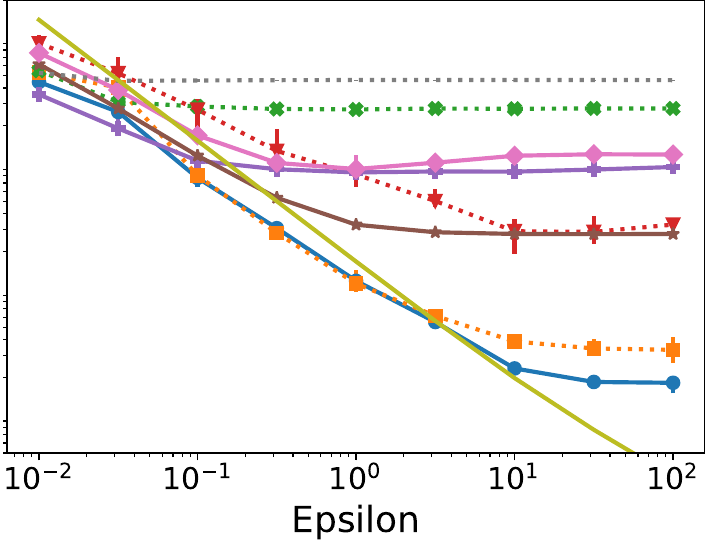}
\includegraphics[width=0.312\textwidth]{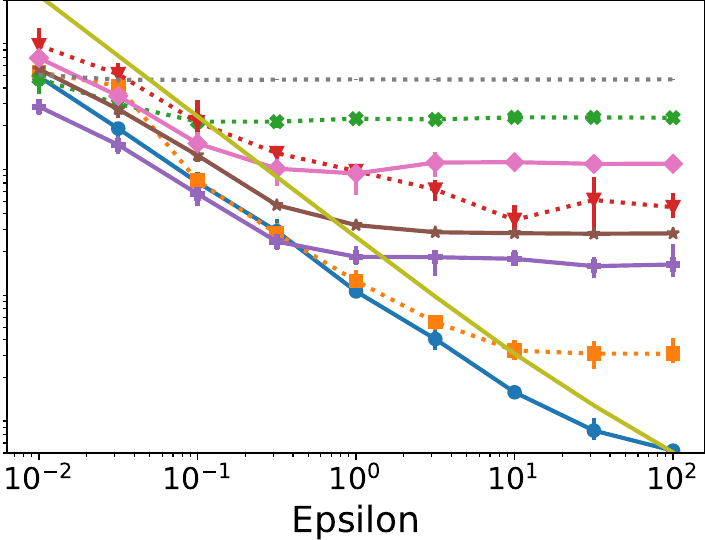}
\caption{\label{fig:main} Workload error (y-axis) vs Epsilon (x-axis) of competing mechanisms on the \general (left), \target (center), and \weighted (right) workloads for $\delta=10^{-9}$.}
\end{figure*}

\subsection{Experimental Results}

Experimental results are shown in \cref{fig:main}.  \revision{Results for the \titanic dataset are omitted due to space.}  Workload-aware mechanisms are shown by solid lines, while workload-agnostic mechanisms are shown with dotted lines.  \revision{Some points are missing from the plots, indicating a mechanism failed to complete in under the 24 hour time limit for that experimental setting.}  From these plots, we make the following observations:

\paragraph*{\textbf{\general Workload}}


\begin{enumerate}
\item \aim  consistently achieves competitive workload error, across all datasets and privacy regimes considered.  On average, across all six datasets and nine privacy parameters, \aim improved over \privmrf by a factor of $1.3\times$, \mst by a factor of $2.6\times$, \mwempgm by a factor of $1.5\times$, \privbayespgm by a factor $2.2\times$, \rap by a factor $5.6\times$, and \gem by a factor $2.0\times$.  In the most extreme cases, \aim improved over \privmrf by a factor $3.6\times$, \mst by a factor $118\times$, \mwempgm by a factor $16\times$, \privbayespgm by a factor $14.7\times$, \rap by a factor $47.1\times$, and \gem by a factor $11.7\times$.   
\item Prior to \aim, \privmrf was consistently the best performing mechanism, even outperforming all workload-aware mechanisms.  The \general workload is one we expect workload agnostic mechanisms like \privmrf to perform well on, so it is interesting, but not surprising that it outperforms workload-aware mechanisms in this setting.
\item Prior to \aim, the best \emph{workload-aware} mechanism varied for different datasets and privacy levels: \mwempgm was best in 72\% of settings, \gem was best in 28\% of settings \footnote{We compare against a variant of \gem that selects an entire marginal query in each round.  In results not shown, we also evaluated the variant of that measures a single counting query, and found that this variant performs significantly worse.}
, and \rap was best in 0\% of settings.  Including \aim, we observe that it is best in 76\% of settings, followed by \mwempgm in 18\% of settings and \gem in 5\% of settings.  Additionally, in the most interesting regime for practical deployment $(\epsilon \geq 1.0)$, \aim is best in 100\% of settings.
\end{enumerate}

\paragraph*{\textbf{\target Workload}}


\begin{enumerate}
\item All three high-level findings from the previous section are supported by these figures as well.
\item Somewhat surprisingly, \privmrf outperforms all workload-aware mechanisms prior to \aim on this workload.
This is an impressive accomplishment for \privmrf, and clearly highlights the suboptimality of existing workload-aware mechanisms like \mwempgm, \gem, and \rap.  Even though \privmrf is not workload-aware, it is clear from their paper that every detail of the mechanism was carefully thought out to make the mechanism work well in practice, which explains it's impressive performance.  While \aim did outperform \privmrf again, the relative performance did not increase by a meaningful margin --- offering a $1.4\times$ improvement on average and a $4.6\times$ improvement in the best case.
\end{enumerate}

\paragraph*{\textbf{\weighted Workload}}


\begin{enumerate}
\item All four high-level findings from the previous sections are generally supported by these figures as well, with the following interesting exception:
\item \privmrf did not score well on \salary, and while it was still generally the second best mechanism on the other datasets (again out-performing the workload-aware mechanisms in many cases), the improvement offered by \aim over \privmrf is much larger for this workload, averaging a $2\times$ improvement with up to a $5.7\times$ improvement in the best case.  We suspect for this setting, workload-awareness is essential to achieve strong performance.
\end{enumerate}

\revision{
\subsection{Ablations} \label{sec:ablation}
In this section, we systematically evaluate the components of \aim, by making modifications to the base mechanism and measuring their impact on workload error.
Specifically, the elements we study are enumerated in \cref{table:ablation}, and are labeled by B1, B2, and B3 for the basic elements of a good mechanism described in \cref{sec:basic}, A1, A2, and A3 for the new elements of \aim described in \cref{sec:aim}, and O1 for an additional relevent element.  For each element of \aim listed below, we run \aim with and without that element across the entire set of experimental configurations we considered in this work, i.e., 9 privacy budgets $ \times $ 6 datasets $\times$ 3 workloads $ \times $ 5 trials.  For each of the $162$ (privacy budget, dataset, workload) triples, we have $5$ measurements which we use to compute two things: (1) the ratio of average workload errors with and without the specified element, and (2) a p-value from a one sided t-test.  The former quantity provides a measure of \emph{practical significance}, while the latter quantity provides a measure of \emph{statistical significance}.

\cref{fig:ablation1} shows the distribution of error ratios for each element across experimental settings, visualized as a box-and-whisker plot;
ratios above $1$ indicate the element of \aim reduced error.  Aggregating the error ratios via geometric mean reveals the three basic elements improved error by a factor $1.18$ on average for Gaussian noise, $1.13$ for unbounded DP, and $1.08$ for a 10/90 budget split.  The new elements of \aim improved error by a factor of $1.03$ for initialization, $1.37$ for the new selection critera, and $1.48$ for adaptive rounds and budget split.  Finally, using \pgm in the generate step, rather than an alternative known as relaxed projection \cite{aydore2021differentially}, improved error by a factor of $2.36$ on average.  Among these elements, the improvement offered by using adaptive rounds + budget split, as well as \pgm, showed a clear depeendence on $\epsilon$, with improvements growing with increasing $\epsilon$.  The other elements showed no clear dependence on $\epsilon$.

Aggregating the $162$ p-values via Stouffer's Z-score method \cite{zaykin2011optimally}, we see that the combined p-value for every element tested ranges from $10^{-22}$ for A1 (initialization) all the way to $10^{-166}$ for A3 (adaptive rounds + budget split).  Thus, it is clear that all elements have a positive effect on the performance of \aim in a statistical sense.

In addition to the algorithmic elements of \aim we evaluate in this section, we conducted experiments varying the model capacity parameter of \aim in \cref{sec:capacity}.  Unsurprisingly, the utility of \aim increases with larger model capacities, at the cost of increased runtime.


\begin{figure*}[t!]
\vspace*{-5mm}
\centering
\subcaptionbox{\revision{Practical Significance \label{fig:ablation1}}}{\includegraphics[width=0.31\textwidth]{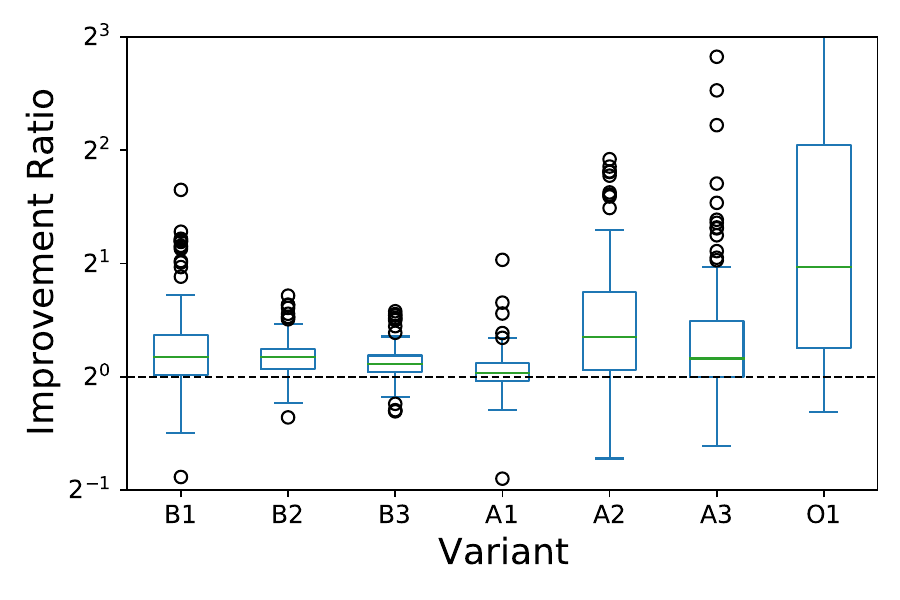}}
\subcaptionbox{\revision{Element codes and descriptions \label{table:ablation}}}{
\resizebox{0.36\textwidth}{!}{
\begin{tabular}{c|c|c}
\textbf{Variant} & \textbf{Element of \aim} & \textbf{Alternative} \\\hline
B1 & Gaussian Noise & Laplace Noise \\\hline
B2 & Unbounded DP & Bounded DP \\\hline
B3 & 10/90 budget split & 50/50 budget split \\\hline\hline
A1 & Independent initialization & Uniform initialization \\\hline
A2 & New selection criteria + & \mwempgm selection + \\
& candidate set & criteria + candidate set \\\hline
A3 & Adaptive rounds + & $d$ rounds + fixed \\
& budget split & budget per round \\\hline\hline
O1 & \pgm & Relaxed Projection \\\hline
\multicolumn{3}{c}{} \\
\multicolumn{3}{c}{} \\
\end{tabular}}}
\subcaptionbox{True Error vs. Error Bound \label{fig:uncertainty}}{\includegraphics[width=0.29\textwidth,height=0.22\textwidth]{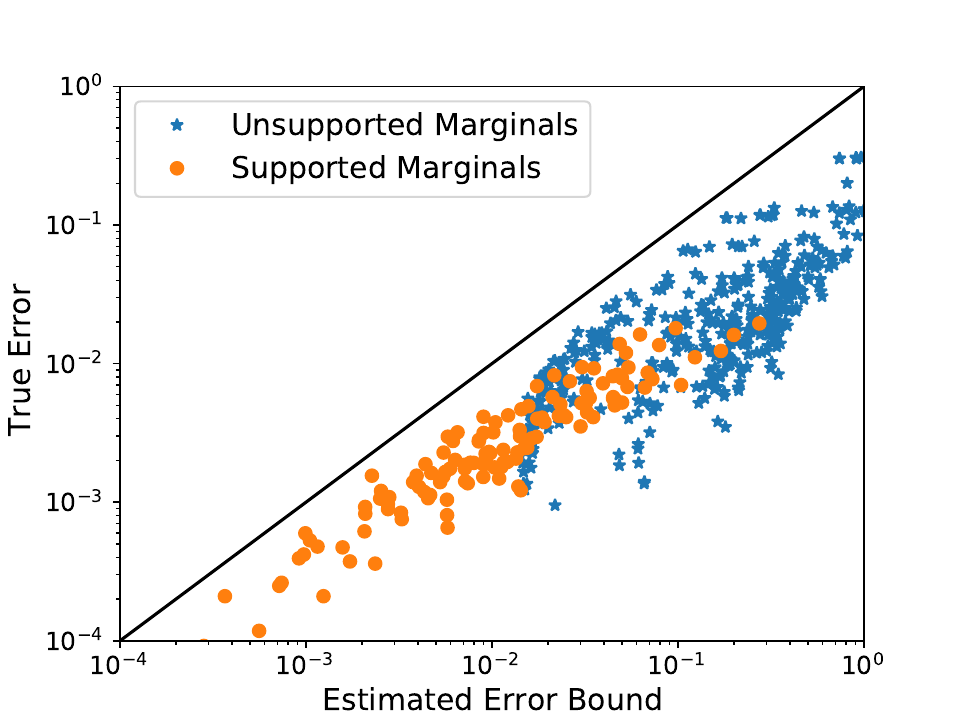}}
\vspace*{-3mm}
\caption{\label{fig:ablation} \revision{(a) Box plot of the ratio of errors with and without using an element of \aim across all experimental settings.  (b) Table describing elements of \aim removed and the alternatives used in this ablation study.}  (c) Accuracy of the uncertainty quantification estimates.}
\vspace*{-3mm}
\end{figure*}

}

\subsection{Uncertainty Quantification}

In this section, we demonstrate that our expressions for uncertainty quantification correctly bound the error, and evaluate how tight the bound is.  For this experiment, we ran \aim on the \fire dataset with the \general workload at $\epsilon=10$.  In \cref{fig:misc} (c), we plot the true error of \aim on each marginal in the workload against the error bound predicted by our expressions.  We set $\lambda=1.7$ in \cref{cor:supported1}, and $\lambda_1=2.7$, $\lambda_2=3.7$ in \cref{cor:unsupported1}, which provides 95\% confidence bounds.
Our main findings are listed below:

\begin{enumerate}
\item For all marginals in the (downward closure of the) workload, the error bound is always greater than true error.  This confirms the validity of the bound, and suggests they are safe to use in practice.  Note that even if some errors were above the bounds, that would not be inconsistent with our guarantee, as at a 95\% confidence level, the bound could fail to hold 5\% of the time.  The fact that it doesn't suggests there is some looseness in the bound.
\item The true errors and the error bounds vary considerably, ranging from $10^{-4}$ all the way up to and beyond $1$.  In general, the supported marginals have both lower errors, and lower error bounds than the unsupported marginals, which is not surprising.   The error bounds are also \emph{tighter} for the supported marginals.  The median ratio between error bound and observed error is $4.4$ for supported marginals and $8.3$ for unsupported marginals.  Intuitively, this makes sense because we know selected marginals should have higher error than non-selected marginals, but the error of the non-selected marginal can be far below that of the selected marginal (and hence the bound), which explains the larger gap between the actual error and our predicted bound.
\end{enumerate}



\section{Discussion and Limitations} \label{sec:open}

\revision{In this paper, we studied the problem of differentially private synthetic data generation, surveying the field and identifying strengths and weaknesses of prior work.  While much of the prior work is conceptually similar, details and specific design decisions differ from mechanism to mechanism, and these small differences can lead to large performance differences in practice.  In practical deployments of differential privacy, these details matter to obtain the best privacy-utility trade-off.  In this work, we propose \aim, a new mechanism where every detail is carefully thought out to maximize utility in practice.  These details allowed \aim to consistently and significantly outperform competitors in our empirical evaluation.  In addition, our uncertainty quantification guarantees enables analysts to understand which queries the synthetic data preserves well, and which it does not, which is important to know when performing downstream analyses on synthetic data.
While our work significantly improves over prior work, the problem of differentially private synthetic data remains far from solved, and there are a number of promising avenues for future work in this space.  We enumerate some of the limitations of \aim below, and identify potential future research directions.
}

\fancyPara{Handling More General Workloads}
In this work, we focused on weighted marginal query workloads.  Designing mechanisms that work for the more general class of linear queries (perhaps defined over the low-dimensional marginals) remains an important open problem.  While the prior work, \mwempgm, \rap, and \gem can handle workloads of this form, they achieve this by selecting a single counting query in each round, rather than a full marginal query, and thus there is likely significant room for improvement.  Beyond linear query workloads, other workloads of interest include more abstract objectives like machine learning efficacy and other non-linear query workloads.  These metrics have been used to evaluate the quality of workload-agnostic synthetic data mechanisms, but have not been provided as input to the mechanisms themselves.  

\fancyPara{Handling Mixed Data Types}
In this work, we assumed the input data was discrete, and each attribute had a finite domain with a reasonably small number of possible values.   Data with numerical attributes must be appropriately discretized before running \aim.  The quality of the discretization could have a significant impact on the quality of the generated synthetic data.  Designing mechanisms that appropriately handle mixed (categorical and numerical) data type is an important problem.  

\fancyPara{Utilizing Public Data}
A promising avenue for future research is to design synthetic data mechanisms that incorporate public data in a principled way.  There are many places in which public data can be naturally incorporated into \aim, and exploring these ideas is a promising way to boost the utility of \aim in real world settings where public data is available.  Early work on this problem includes \cite{liu2021leveraging, mckenna2021winning, liu2021iterative}, but \revision{it certainly warrants additional research.}

\revision{
\fancyPara{Uncertainty Quantification Guarantees}
In this paper, we initiated the study of formal and well-calibrated guarantees about the error of the synthetic data on different marginal queries.  These error estimates can be used to determine to what degree the synthetic data should be trusted. However, our guarantees only pertain to the $L_1$ error of each marginal, and we provide no guarantees on the error in each individual cell of the marginals. These finer-grained guarantees could be useful in some applications, and is an interesting technical challenge for future research.

\fancyPara{Small Workloads}
In our experimental evaluation, as in much of the current literature, we focus on workloads with a large number of marginal queries where privacy and scalability constraints prevent measuring them all.  For smaller workloads, simpler techniques like \gaussianpgm may achieve better performance than \aim, since it does not have to devote budget to the select step.

\fancyPara{High-cardinality attributes}
The scalability of \aim, and more generally any method that uses \pgm depends on the domain of attributes in the dataset.  The datasets we considered in \cref{sec:experiments} were preprocessed to have reasonable domain sizes (most attributes had a domain size $\leq 50$).  Datasets with high-cardinality attributes often have sparse marginals that may deserve special treatment not covered in this paper.
}


\begin{acks}
This work was supported by the National Science Foundation under grants IIS-1749854 and CNS-1954814, and by Oracle Labs, part of Oracle America, through a gift to the University of Massachusetts Amherst in support of academic research.
\end{acks}

\bibliographystyle{ACM-Reference-Format}
\bibliography{refs}


\begin{thebibliography}{57}


\ifx \showCODEN    \undefined \def \showCODEN     #1{\unskip}     \fi
\ifx \showDOI      \undefined \def \showDOI       #1{#1}\fi
\ifx \showISBNx    \undefined \def \showISBNx     #1{\unskip}     \fi
\ifx \showISBNxiii \undefined \def \showISBNxiii  #1{\unskip}     \fi
\ifx \showISSN     \undefined \def \showISSN      #1{\unskip}     \fi
\ifx \showLCCN     \undefined \def \showLCCN      #1{\unskip}     \fi
\ifx \shownote     \undefined \def \shownote      #1{#1}          \fi
\ifx \showarticletitle \undefined \def \showarticletitle #1{#1}   \fi
\ifx \showURL      \undefined \def \showURL       {\relax}        \fi
\providecommand\bibfield[2]{#2}
\providecommand\bibinfo[2]{#2}
\providecommand\natexlab[1]{#1}
\providecommand\showeprint[2][]{arXiv:#2}

\bibitem[\protect\citeauthoryear{Abay, Zhou, Kantarcioglu, Thuraisingham, and
  Sweeney}{Abay et~al\mbox{.}}{2018}]%
        {abay2018privacy}
\bibfield{author}{\bibinfo{person}{Nazmiye~Ceren Abay}, \bibinfo{person}{Yan
  Zhou}, \bibinfo{person}{Murat Kantarcioglu}, \bibinfo{person}{Bhavani~M.
  Thuraisingham}, {and} \bibinfo{person}{Latanya Sweeney}.}
  \bibinfo{year}{2018}\natexlab{}.
\newblock \showarticletitle{Privacy Preserving Synthetic Data Release Using
  Deep Learning}. In \bibinfo{booktitle}{\emph{Machine Learning and Knowledge
  Discovery in Databases - European Conference, {ECML} {PKDD} 2018, Dublin,
  Ireland, September 10-14, 2018, Proceedings, Part {I}}}
  \emph{(\bibinfo{series}{Lecture Notes in Computer Science})},
  \bibfield{editor}{\bibinfo{person}{Michele Berlingerio},
  \bibinfo{person}{Francesco Bonchi}, \bibinfo{person}{Thomas G{\"{a}}rtner},
  \bibinfo{person}{Neil Hurley}, {and} \bibinfo{person}{Georgiana Ifrim}}
  (Eds.), Vol.~\bibinfo{volume}{11051}. \bibinfo{publisher}{Springer},
  \bibinfo{pages}{510--526}.
\newblock
\urldef\tempurl%
\url{https://doi.org/10.1007/978-3-030-10925-7\_31}
\showDOI{\tempurl}


\bibitem[\protect\citeauthoryear{Asghar, Ding, Rakotoarivelo, Mrabet, and
  K{\^{a}}afar}{Asghar et~al\mbox{.}}{2019}]%
        {asghar2019differentially}
\bibfield{author}{\bibinfo{person}{Hassan~Jameel Asghar}, \bibinfo{person}{Ming
  Ding}, \bibinfo{person}{Thierry Rakotoarivelo}, \bibinfo{person}{Sirine
  Mrabet}, {and} \bibinfo{person}{Mohamed~Ali K{\^{a}}afar}.}
  \bibinfo{year}{2019}\natexlab{}.
\newblock \showarticletitle{Differentially Private Release of High-Dimensional
  Datasets using the Gaussian Copula}.
\newblock \bibinfo{journal}{\emph{CoRR}}  \bibinfo{volume}{abs/1902.01499}
  (\bibinfo{year}{2019}).
\newblock
\showeprint[arXiv]{1902.01499}
\urldef\tempurl%
\url{http://arxiv.org/abs/1902.01499}
\showURL{%
\tempurl}


\bibitem[\protect\citeauthoryear{Aydore, Brown, Kearns, Kenthapadi, Melis,
  Roth, and Siva}{Aydore et~al\mbox{.}}{2021}]%
        {aydore2021differentially}
\bibfield{author}{\bibinfo{person}{Sergul Aydore}, \bibinfo{person}{William
  Brown}, \bibinfo{person}{Michael Kearns}, \bibinfo{person}{Krishnaram
  Kenthapadi}, \bibinfo{person}{Luca Melis}, \bibinfo{person}{Aaron Roth},
  {and} \bibinfo{person}{Ankit~A Siva}.} \bibinfo{year}{2021}\natexlab{}.
\newblock \showarticletitle{Differentially Private Query Release Through
  Adaptive Projection}. In \bibinfo{booktitle}{\emph{Proceedings of the 38th
  International Conference on Machine Learning}}
  \emph{(\bibinfo{series}{Proceedings of Machine Learning Research})},
  \bibfield{editor}{\bibinfo{person}{Marina Meila} {and} \bibinfo{person}{Tong
  Zhang}} (Eds.), Vol.~\bibinfo{volume}{139}. \bibinfo{publisher}{PMLR},
  \bibinfo{pages}{457--467}.
\newblock
\urldef\tempurl%
\url{https://proceedings.mlr.press/v139/aydore21a.html}
\showURL{%
\tempurl}


\bibitem[\protect\citeauthoryear{Bindschaedler, Shokri, and
  Gunter}{Bindschaedler et~al\mbox{.}}{2017}]%
        {bindschaedler2017plausible}
\bibfield{author}{\bibinfo{person}{Vincent Bindschaedler},
  \bibinfo{person}{Reza Shokri}, {and} \bibinfo{person}{Carl~A. Gunter}.}
  \bibinfo{year}{2017}\natexlab{}.
\newblock \showarticletitle{Plausible Deniability for Privacy-Preserving Data
  Synthesis}.
\newblock \bibinfo{journal}{\emph{Proceedings of the VLDB Endowment}}
  \bibinfo{volume}{10}, \bibinfo{number}{5} (\bibinfo{year}{2017}),
  \bibinfo{pages}{481--492}.
\newblock
\urldef\tempurl%
\url{https://doi.org/10.14778/3055540.3055542}
\showDOI{\tempurl}


\bibitem[\protect\citeauthoryear{Bun and Steinke}{Bun and Steinke}{2016}]%
        {bun2016concentrated}
\bibfield{author}{\bibinfo{person}{Mark Bun} {and} \bibinfo{person}{Thomas
  Steinke}.} \bibinfo{year}{2016}\natexlab{}.
\newblock \showarticletitle{Concentrated Differential Privacy: Simplifications,
  Extensions, and Lower Bounds}. In \bibinfo{booktitle}{\emph{Theory of
  Cryptography Conference}}. \bibinfo{publisher}{Springer},
  \bibinfo{pages}{635--658}.
\newblock
\urldef\tempurl%
\url{https://doi.org/10.1007/978-3-662-53641-4_24}
\showDOI{\tempurl}


\bibitem[\protect\citeauthoryear{Cadez, Heckerman, Meek, Smyth, and
  White}{Cadez et~al\mbox{.}}{2000}]%
        {cadez2000visualization}
\bibfield{author}{\bibinfo{person}{Igor Cadez}, \bibinfo{person}{David
  Heckerman}, \bibinfo{person}{Christopher Meek}, \bibinfo{person}{Padhraic
  Smyth}, {and} \bibinfo{person}{Steven White}.}
  \bibinfo{year}{2000}\natexlab{}.
\newblock \showarticletitle{Visualization of navigation patterns on a web site
  using model-based clustering}. In \bibinfo{booktitle}{\emph{Proceedings of
  the sixth ACM SIGKDD international conference on Knowledge discovery and data
  mining}}. \bibinfo{pages}{280--284}.
\newblock


\bibitem[\protect\citeauthoryear{Cai, Lei, Wei, and Xiao}{Cai
  et~al\mbox{.}}{2021}]%
        {cai2021data}
\bibfield{author}{\bibinfo{person}{Kuntai Cai}, \bibinfo{person}{Xiaoyu Lei},
  \bibinfo{person}{Jianxin Wei}, {and} \bibinfo{person}{Xiaokui Xiao}.}
  \bibinfo{year}{2021}\natexlab{}.
\newblock \showarticletitle{Data synthesis via differentially private markov
  random fields}.
\newblock \bibinfo{journal}{\emph{Proceedings of the VLDB Endowment}}
  \bibinfo{volume}{14}, \bibinfo{number}{11} (\bibinfo{year}{2021}),
  \bibinfo{pages}{2190--2202}.
\newblock


\bibitem[\protect\citeauthoryear{Canonne, Kamath, and Steinke}{Canonne
  et~al\mbox{.}}{2020}]%
        {canonne2020discrete}
\bibfield{author}{\bibinfo{person}{Clément~L. Canonne},
  \bibinfo{person}{Gautam Kamath}, {and} \bibinfo{person}{Thomas Steinke}.}
  \bibinfo{year}{2020}\natexlab{}.
\newblock \showarticletitle{The Discrete Gaussian for Differential Privacy}. In
  \bibinfo{booktitle}{\emph{NeurIPS}}.
\newblock
\urldef\tempurl%
\url{https://proceedings.neurips.cc/paper/2020/hash/b53b3a3d6ab90ce0268229151c9bde11-Abstract.html}
\showURL{%
\tempurl}


\bibitem[\protect\citeauthoryear{Cesar and Rogers}{Cesar and Rogers}{2021}]%
        {cesar2021bounding}
\bibfield{author}{\bibinfo{person}{Mark Cesar} {and} \bibinfo{person}{Ryan
  Rogers}.} \bibinfo{year}{2021}\natexlab{}.
\newblock \showarticletitle{Bounding, Concentrating, and Truncating: Unifying
  Privacy Loss Composition for Data Analytics}. In
  \bibinfo{booktitle}{\emph{Proceedings of the 32nd International Conference on
  Algorithmic Learning Theory}} \emph{(\bibinfo{series}{Proceedings of Machine
  Learning Research})}, \bibfield{editor}{\bibinfo{person}{Vitaly Feldman},
  \bibinfo{person}{Katrina Ligett}, {and} \bibinfo{person}{Sivan Sabato}}
  (Eds.), Vol.~\bibinfo{volume}{132}. \bibinfo{publisher}{PMLR},
  \bibinfo{pages}{421--457}.
\newblock
\urldef\tempurl%
\url{https://proceedings.mlr.press/v132/cesar21a.html}
\showURL{%
\tempurl}


\bibitem[\protect\citeauthoryear{Charest}{Charest}{2011}]%
        {charest2011can}
\bibfield{author}{\bibinfo{person}{Anne{-}Sophie Charest}.}
  \bibinfo{year}{2011}\natexlab{}.
\newblock \showarticletitle{How Can We Analyze Differentially-Private Synthetic
  Datasets?}
\newblock \bibinfo{journal}{\emph{Journal of Privacy and Confidentiality}}
  \bibinfo{volume}{2}, \bibinfo{number}{2} (\bibinfo{year}{2011}).
\newblock
\urldef\tempurl%
\url{https://doi.org/10.29012/jpc.v2i2.589}
\showDOI{\tempurl}


\bibitem[\protect\citeauthoryear{Chen, Xiao, Zhang, and Xu}{Chen
  et~al\mbox{.}}{2015}]%
        {chen2015differentially}
\bibfield{author}{\bibinfo{person}{Rui Chen}, \bibinfo{person}{Qian Xiao},
  \bibinfo{person}{Yu Zhang}, {and} \bibinfo{person}{Jianliang Xu}.}
  \bibinfo{year}{2015}\natexlab{}.
\newblock \showarticletitle{Differentially private high-dimensional data
  publication via sampling-based inference}. In
  \bibinfo{booktitle}{\emph{Proceedings of the 21th ACM SIGKDD International
  Conference on Knowledge Discovery and Data Mining}}. ACM,
  \bibinfo{pages}{129--138}.
\newblock
\urldef\tempurl%
\url{https://doi.org/10.1145/2783258.2783379}
\showDOI{\tempurl}


\bibitem[\protect\citeauthoryear{Ding, Winslett, Han, and Li}{Ding
  et~al\mbox{.}}{2011}]%
        {ding2011differentially}
\bibfield{author}{\bibinfo{person}{Bolin Ding}, \bibinfo{person}{Marianne
  Winslett}, \bibinfo{person}{Jiawei Han}, {and} \bibinfo{person}{Zhenhui Li}.}
  \bibinfo{year}{2011}\natexlab{}.
\newblock \showarticletitle{Differentially private data cubes: optimizing noise
  sources and consistency}. In \bibinfo{booktitle}{\emph{Proceedings of the
  {ACM} {SIGMOD} International Conference on Management of Data, {SIGMOD} 2011,
  Athens, Greece, June 12-16, 2011}},
  \bibfield{editor}{\bibinfo{person}{Timos~K. Sellis},
  \bibinfo{person}{Ren{\'{e}}e~J. Miller}, \bibinfo{person}{Anastasios
  Kementsietsidis}, {and} \bibinfo{person}{Yannis Velegrakis}} (Eds.).
  \bibinfo{publisher}{{ACM}}, \bibinfo{pages}{217--228}.
\newblock
\urldef\tempurl%
\url{https://doi.org/10.1145/1989323.1989347}
\showDOI{\tempurl}


\bibitem[\protect\citeauthoryear{Dinur and Nissim}{Dinur and Nissim}{2003}]%
        {dinur2003revealing}
\bibfield{author}{\bibinfo{person}{Irit Dinur} {and} \bibinfo{person}{Kobbi
  Nissim}.} \bibinfo{year}{2003}\natexlab{}.
\newblock \showarticletitle{Revealing information while preserving privacy}. In
  \bibinfo{booktitle}{\emph{Proceedings of the Twenty-Second {ACM}
  {SIGACT-SIGMOD-SIGART} Symposium on Principles of Database Systems, June
  9-12, 2003, San Diego, CA, {USA}}}, \bibfield{editor}{\bibinfo{person}{Frank
  Neven}, \bibinfo{person}{Catriel Beeri}, {and} \bibinfo{person}{Tova Milo}}
  (Eds.). \bibinfo{publisher}{{ACM}}, \bibinfo{pages}{202--210}.
\newblock
\urldef\tempurl%
\url{https://doi.org/10.1145/773153.773173}
\showDOI{\tempurl}


\bibitem[\protect\citeauthoryear{Dwork, Nissim, and Smith}{Dwork
  et~al\mbox{.}}{2006}]%
        {dwork2006calibrating}
\bibfield{author}{\bibinfo{person}{Cynthia Dwork}, \bibinfo{person}{Frank
  McSherry~Kobbi Nissim}, {and} \bibinfo{person}{Adam Smith}.}
  \bibinfo{year}{2006}\natexlab{}.
\newblock \showarticletitle{Calibrating Noise to Sensitivity in Private Data
  Analysis}. In \bibinfo{booktitle}{\emph{TCC}}. \bibinfo{pages}{265--284}.
\newblock
\urldef\tempurl%
\url{https://doi.org/10.29012/jpc.v7i3.405}
\showDOI{\tempurl}


\bibitem[\protect\citeauthoryear{Feldman and Zrnic}{Feldman and Zrnic}{2021}]%
        {feldman2021individual}
\bibfield{author}{\bibinfo{person}{Vitaly Feldman} {and}
  \bibinfo{person}{Tijana Zrnic}.} \bibinfo{year}{2021}\natexlab{}.
\newblock \showarticletitle{Individual privacy accounting via a renyi filter}.
\newblock \bibinfo{journal}{\emph{Advances in Neural Information Processing
  Systems}}  \bibinfo{volume}{34} (\bibinfo{year}{2021}),
  \bibinfo{pages}{28080--28091}.
\newblock


\bibitem[\protect\citeauthoryear{Frame}{Frame}{1945}]%
        {frame1945mean}
\bibfield{author}{\bibinfo{person}{James~S Frame}.}
  \bibinfo{year}{1945}\natexlab{}.
\newblock \showarticletitle{Mean deviation of the binomial distribution}.
\newblock \bibinfo{journal}{\emph{The American Mathematical Monthly}}
  \bibinfo{volume}{52}, \bibinfo{number}{7} (\bibinfo{year}{1945}),
  \bibinfo{pages}{377--379}.
\newblock


\bibitem[\protect\citeauthoryear{Frank E. Harrell~Jr.}{Frank E.
  Harrell~Jr.}{[n.d.]}]%
        {titanic}
\bibfield{author}{\bibinfo{person}{Thomas~Cason Frank E. Harrell~Jr.}}
  \bibinfo{year}{[n.d.]}\natexlab{}.
\newblock \bibinfo{title}{Encyclopedia Titanica}.
\newblock
\newblock


\bibitem[\protect\citeauthoryear{Ge, Mohapatra, He, and Ilyas}{Ge
  et~al\mbox{.}}{2021}]%
        {ge2020kamino}
\bibfield{author}{\bibinfo{person}{Chang Ge}, \bibinfo{person}{Shubhankar
  Mohapatra}, \bibinfo{person}{Xi He}, {and} \bibinfo{person}{Ihab~F. Ilyas}.}
  \bibinfo{year}{2021}\natexlab{}.
\newblock \showarticletitle{Kamino: Constraint-Aware Differentially Private
  Data Synthesis}.
\newblock \bibinfo{journal}{\emph{Proceedings of the VLDB Endowment}}
  \bibinfo{volume}{14}, \bibinfo{number}{10} (\bibinfo{year}{2021}),
  \bibinfo{pages}{1886--1899}.
\newblock
\urldef\tempurl%
\url{http://www.vldb.org/pvldb/vol14/p1886-ge.pdf}
\showURL{%
\tempurl}


\bibitem[\protect\citeauthoryear{Goodfellow, Pouget{-}Abadie, Mirza, Xu,
  Warde{-}Farley, Ozair, Courville, and Bengio}{Goodfellow
  et~al\mbox{.}}{2014}]%
        {goodfellow2014generative}
\bibfield{author}{\bibinfo{person}{Ian~J. Goodfellow}, \bibinfo{person}{Jean
  Pouget{-}Abadie}, \bibinfo{person}{Mehdi Mirza}, \bibinfo{person}{Bing Xu},
  \bibinfo{person}{David Warde{-}Farley}, \bibinfo{person}{Sherjil Ozair},
  \bibinfo{person}{Aaron~C. Courville}, {and} \bibinfo{person}{Yoshua Bengio}.}
  \bibinfo{year}{2014}\natexlab{}.
\newblock \showarticletitle{Generative Adversarial Nets}. In
  \bibinfo{booktitle}{\emph{Advances in Neural Information Processing Systems
  27: Annual Conference on Neural Information Processing Systems 2014, December
  8-13 2014, Montreal, Quebec, Canada}},
  \bibfield{editor}{\bibinfo{person}{Zoubin Ghahramani}, \bibinfo{person}{Max
  Welling}, \bibinfo{person}{Corinna Cortes}, \bibinfo{person}{Neil~D.
  Lawrence}, {and} \bibinfo{person}{Kilian~Q. Weinberger}} (Eds.).
  \bibinfo{pages}{2672--2680}.
\newblock
\urldef\tempurl%
\url{https://proceedings.neurips.cc/paper/2014/hash/5ca3e9b122f61f8f06494c97b1afccf3-Abstract.html}
\showURL{%
\tempurl}


\bibitem[\protect\citeauthoryear{Greene}{Greene}{2003}]%
        {greene2003econometric}
\bibfield{author}{\bibinfo{person}{William~H Greene}.}
  \bibinfo{year}{2003}\natexlab{}.
\newblock \bibinfo{booktitle}{\emph{Econometric analysis}}.
\newblock \bibinfo{publisher}{Pearson Education India}.
\newblock


\bibitem[\protect\citeauthoryear{Hardt, Ligett, and McSherry}{Hardt
  et~al\mbox{.}}{2012}]%
        {hardt2010simple}
\bibfield{author}{\bibinfo{person}{Moritz Hardt}, \bibinfo{person}{Katrina
  Ligett}, {and} \bibinfo{person}{Frank McSherry}.}
  \bibinfo{year}{2012}\natexlab{}.
\newblock \showarticletitle{A Simple and Practical Algorithm for Differentially
  Private Data Release}. In \bibinfo{booktitle}{\emph{Advances in Neural
  Information Processing Systems 25: 26th Annual Conference on Neural
  Information Processing Systems 2012. Proceedings of a meeting held December
  3-6, 2012, Lake Tahoe, Nevada, United States}},
  \bibfield{editor}{\bibinfo{person}{Peter~L. Bartlett},
  \bibinfo{person}{Fernando C.~N. Pereira}, \bibinfo{person}{Christopher J.~C.
  Burges}, \bibinfo{person}{L{\'{e}}on Bottou}, {and}
  \bibinfo{person}{Kilian~Q. Weinberger}} (Eds.). \bibinfo{pages}{2348--2356}.
\newblock
\urldef\tempurl%
\url{https://proceedings.neurips.cc/paper/2012/hash/208e43f0e45c4c78cafadb83d2888cb6-Abstract.html}
\showURL{%
\tempurl}


\bibitem[\protect\citeauthoryear{Hartung, Knapp, Sinha, and Sinha}{Hartung
  et~al\mbox{.}}{2008}]%
        {hartung2008statistical}
\bibfield{author}{\bibinfo{person}{Joachim Hartung}, \bibinfo{person}{Guido
  Knapp}, \bibinfo{person}{Bimal~K Sinha}, {and} \bibinfo{person}{Bimal~K
  Sinha}.} \bibinfo{year}{2008}\natexlab{}.
\newblock \bibinfo{booktitle}{\emph{Statistical meta-analysis with
  applications}}. Vol.~\bibinfo{volume}{6}.
\newblock \bibinfo{publisher}{Wiley Online Library}.
\newblock


\bibitem[\protect\citeauthoryear{Hay, Machanavajjhala, Miklau, Chen, and
  Zhang}{Hay et~al\mbox{.}}{2016}]%
        {hay2016principled}
\bibfield{author}{\bibinfo{person}{Michael Hay}, \bibinfo{person}{Ashwin
  Machanavajjhala}, \bibinfo{person}{Gerome Miklau}, \bibinfo{person}{Yan
  Chen}, {and} \bibinfo{person}{Dan Zhang}.} \bibinfo{year}{2016}\natexlab{}.
\newblock \showarticletitle{Principled evaluation of differentially private
  algorithms using dpbench}. In \bibinfo{booktitle}{\emph{Proceedings of the
  2016 International Conference on Management of Data}}.
  \bibinfo{pages}{139--154}.
\newblock


\bibitem[\protect\citeauthoryear{Huang, McKenna, Bissias, Miklau, Hay, and
  Machanavajjhala}{Huang et~al\mbox{.}}{2019}]%
        {huang2019psyndb}
\bibfield{author}{\bibinfo{person}{Zhiqi Huang}, \bibinfo{person}{Ryan
  McKenna}, \bibinfo{person}{George Bissias}, \bibinfo{person}{Gerome Miklau},
  \bibinfo{person}{Michael Hay}, {and} \bibinfo{person}{Ashwin
  Machanavajjhala}.} \bibinfo{year}{2019}\natexlab{}.
\newblock \showarticletitle{PSynDB: accurate and accessible private data
  generation}.
\newblock \bibinfo{journal}{\emph{VLDB Demo}} (\bibinfo{year}{2019}).
\newblock
\urldef\tempurl%
\url{https://people.cs.umass.edu/~miklau/assets/pubs/dp/huang19psyndata-demo.pdf}
\showURL{%
\tempurl}


\bibitem[\protect\citeauthoryear{Johnson, Kemp, and Kotz}{Johnson
  et~al\mbox{.}}{2005}]%
        {johnson2005univariate}
\bibfield{author}{\bibinfo{person}{Norman~L Johnson},
  \bibinfo{person}{Adrienne~W Kemp}, {and} \bibinfo{person}{Samuel Kotz}.}
  \bibinfo{year}{2005}\natexlab{}.
\newblock \bibinfo{booktitle}{\emph{Univariate discrete distributions}}.
  Vol.~\bibinfo{volume}{444}.
\newblock \bibinfo{publisher}{John Wiley \& Sons}.
\newblock


\bibitem[\protect\citeauthoryear{Jordon, Yoon, and van~der Schaar}{Jordon
  et~al\mbox{.}}{2019}]%
        {jordon2018pate}
\bibfield{author}{\bibinfo{person}{James Jordon}, \bibinfo{person}{Jinsung
  Yoon}, {and} \bibinfo{person}{Mihaela van~der Schaar}.}
  \bibinfo{year}{2019}\natexlab{}.
\newblock \showarticletitle{{PATE-GAN:} Generating Synthetic Data with
  Differential Privacy Guarantees}. In \bibinfo{booktitle}{\emph{7th
  International Conference on Learning Representations, {ICLR} 2019, New
  Orleans, LA, USA, May 6-9, 2019}}. \bibinfo{publisher}{OpenReview.net}.
\newblock
\urldef\tempurl%
\url{https://openreview.net/forum?id=S1zk9iRqF7}
\showURL{%
\tempurl}


\bibitem[\protect\citeauthoryear{King}{King}{[n.d.]}]%
        {kingnoisy}
\bibfield{author}{\bibinfo{person}{Gary King}.}
  \bibinfo{year}{[n.d.]}\natexlab{}.
\newblock \showarticletitle{Noisy Data from the Noisy Census}.
\newblock  (\bibinfo{year}{[n.\,d.]}).
\newblock


\bibitem[\protect\citeauthoryear{Kohavi et~al\mbox{.}}{Kohavi
  et~al\mbox{.}}{1996}]%
        {kohavi1996scaling}
\bibfield{author}{\bibinfo{person}{Ron Kohavi} {et~al\mbox{.}}}
  \bibinfo{year}{1996}\natexlab{}.
\newblock \showarticletitle{Scaling up the accuracy of naive-bayes classifiers:
  A decision-tree hybrid.}. In \bibinfo{booktitle}{\emph{Kdd}},
  Vol.~\bibinfo{volume}{96}. \bibinfo{pages}{202--207}.
\newblock


\bibitem[\protect\citeauthoryear{Li, Xiong, and Jiang}{Li
  et~al\mbox{.}}{2014}]%
        {li2014differentially}
\bibfield{author}{\bibinfo{person}{Haoran Li}, \bibinfo{person}{Li Xiong},
  {and} \bibinfo{person}{Xiaoqian Jiang}.} \bibinfo{year}{2014}\natexlab{}.
\newblock \showarticletitle{Differentially Private Synthesization of
  Multi-Dimensional Data using Copula Functions}. In
  \bibinfo{booktitle}{\emph{Proceedings of the 17th International Conference on
  Extending Database Technology, {EDBT} 2014, Athens, Greece, March 24-28,
  2014}}, \bibfield{editor}{\bibinfo{person}{Sihem Amer{-}Yahia},
  \bibinfo{person}{Vassilis Christophides}, \bibinfo{person}{Anastasios
  Kementsietsidis}, \bibinfo{person}{Minos~N. Garofalakis},
  \bibinfo{person}{Stratos Idreos}, {and} \bibinfo{person}{Vincent Leroy}}
  (Eds.). \bibinfo{publisher}{OpenProceedings.org}, \bibinfo{pages}{475--486}.
\newblock
\urldef\tempurl%
\url{https://doi.org/10.5441/002/edbt.2014.43}
\showDOI{\tempurl}


\bibitem[\protect\citeauthoryear{Liu}{Liu}{2016}]%
        {liu2016model}
\bibfield{author}{\bibinfo{person}{Fang Liu}.} \bibinfo{year}{2016}\natexlab{}.
\newblock \showarticletitle{Model-based differentially private data synthesis}.
\newblock \bibinfo{journal}{\emph{arXiv preprint arXiv:1606.08052}}
  (\bibinfo{year}{2016}).
\newblock
\urldef\tempurl%
\url{https://arxiv.org/abs/1606.08052}
\showURL{%
\tempurl}


\bibitem[\protect\citeauthoryear{Liu and Talwar}{Liu and Talwar}{2019}]%
        {liu2019private}
\bibfield{author}{\bibinfo{person}{Jingcheng Liu} {and} \bibinfo{person}{Kunal
  Talwar}.} \bibinfo{year}{2019}\natexlab{}.
\newblock \showarticletitle{Private selection from private candidates}. In
  \bibinfo{booktitle}{\emph{Proceedings of the 51st Annual ACM SIGACT Symposium
  on Theory of Computing}}. \bibinfo{pages}{298--309}.
\newblock


\bibitem[\protect\citeauthoryear{Liu, Vietri, Steinke, Ullman, and Wu}{Liu
  et~al\mbox{.}}{2021b}]%
        {liu2021leveraging}
\bibfield{author}{\bibinfo{person}{Terrance Liu}, \bibinfo{person}{Giuseppe
  Vietri}, \bibinfo{person}{Thomas Steinke}, \bibinfo{person}{Jonathan~R.
  Ullman}, {and} \bibinfo{person}{Zhiwei~Steven Wu}.}
  \bibinfo{year}{2021}\natexlab{b}.
\newblock \showarticletitle{Leveraging Public Data for Practical Private Query
  Release}. In \bibinfo{booktitle}{\emph{ICML}}. \bibinfo{pages}{6968--6977}.
\newblock
\urldef\tempurl%
\url{http://proceedings.mlr.press/v139/liu21w.html}
\showURL{%
\tempurl}


\bibitem[\protect\citeauthoryear{Liu, Vietri, and Wu}{Liu
  et~al\mbox{.}}{2021a}]%
        {liu2021iterative}
\bibfield{author}{\bibinfo{person}{Terrance Liu}, \bibinfo{person}{Giuseppe
  Vietri}, {and} \bibinfo{person}{Steven Wu}.}
  \bibinfo{year}{2021}\natexlab{a}.
\newblock \showarticletitle{Iterative Methods for Private Synthetic Data:
  Unifying Framework and New Methods}. In \bibinfo{booktitle}{\emph{Advances in
  Neural Information Processing Systems}},
  \bibfield{editor}{\bibinfo{person}{A.~Beygelzimer},
  \bibinfo{person}{Y.~Dauphin}, \bibinfo{person}{P.~Liang}, {and}
  \bibinfo{person}{J.~Wortman Vaughan}} (Eds.).
\newblock


\bibitem[\protect\citeauthoryear{Manton}{Manton}{2010}]%
        {nltcs}
\bibfield{author}{\bibinfo{person}{Kenneth~G. Manton}.}
  \bibinfo{year}{2010}\natexlab{}.
\newblock \bibinfo{title}{National Long-Term Care Survey: 1982, 1984, 1989,
  1994, 1999, and 2004}.
\newblock
\newblock


\bibitem[\protect\citeauthoryear{McKenna and Liu}{McKenna and Liu}{2022}]%
        {dporgblog}
\bibfield{author}{\bibinfo{person}{Ryan McKenna} {and}
  \bibinfo{person}{Terrance Liu}.} \bibinfo{year}{2022}\natexlab{}.
\newblock \bibinfo{title}{A simple recipe for private synthetic data
  generation}.
\newblock
\newblock
\urldef\tempurl%
\url{DifferentialPrivacy.org}
\showURL{%
\tempurl}


\bibitem[\protect\citeauthoryear{McKenna, Miklau, Hay, and
  Machanavajjhala}{McKenna et~al\mbox{.}}{2018}]%
        {mckenna2018optimizing}
\bibfield{author}{\bibinfo{person}{Ryan McKenna}, \bibinfo{person}{Gerome
  Miklau}, \bibinfo{person}{Michael Hay}, {and} \bibinfo{person}{Ashwin
  Machanavajjhala}.} \bibinfo{year}{2018}\natexlab{}.
\newblock \showarticletitle{Optimizing error of high-dimensional statistical
  queries under differential privacy}.
\newblock \bibinfo{journal}{\emph{Proceedings of the VLDB Endowment}}
  \bibinfo{volume}{11}, \bibinfo{number}{10} (\bibinfo{year}{2018}),
  \bibinfo{pages}{1206--1219}.
\newblock
\urldef\tempurl%
\url{https://doi.org/10.14778/3231751.3231769}
\showDOI{\tempurl}


\bibitem[\protect\citeauthoryear{McKenna, Miklau, and Sheldon}{McKenna
  et~al\mbox{.}}{2021a}]%
        {mckenna2021winning}
\bibfield{author}{\bibinfo{person}{Ryan McKenna}, \bibinfo{person}{Gerome
  Miklau}, {and} \bibinfo{person}{Daniel Sheldon}.}
  \bibinfo{year}{2021}\natexlab{a}.
\newblock \showarticletitle{Winning the NIST Contest: A scalable and general
  approach to differentially private synthetic data}.
\newblock \bibinfo{journal}{\emph{Journal of Privacy and Confidentiality}}
  \bibinfo{volume}{11}, \bibinfo{number}{3} (\bibinfo{year}{2021}).
\newblock


\bibitem[\protect\citeauthoryear{McKenna, Mullins, Sheldon, and Miklau}{McKenna
  et~al\mbox{.}}{2022}]%
        {mckenna2022aim}
\bibfield{author}{\bibinfo{person}{Ryan McKenna}, \bibinfo{person}{Brett
  Mullins}, \bibinfo{person}{Daniel Sheldon}, {and} \bibinfo{person}{Gerome
  Miklau}.} \bibinfo{year}{2022}\natexlab{}.
\newblock \showarticletitle{AIM: An Adaptive and Iterative Mechanism for
  Differentially Private Synthetic Data}.
\newblock  (\bibinfo{year}{2022}).
\newblock


\bibitem[\protect\citeauthoryear{McKenna, Pradhan, Sheldon, and Miklau}{McKenna
  et~al\mbox{.}}{2021b}]%
        {mckenna2021relaxed}
\bibfield{author}{\bibinfo{person}{Ryan McKenna}, \bibinfo{person}{Siddhant
  Pradhan}, \bibinfo{person}{Daniel~R Sheldon}, {and} \bibinfo{person}{Gerome
  Miklau}.} \bibinfo{year}{2021}\natexlab{b}.
\newblock \showarticletitle{Relaxed Marginal Consistency for Differentially
  Private Query Answering}.
\newblock \bibinfo{journal}{\emph{Advances in Neural Information Processing
  Systems}}  \bibinfo{volume}{34} (\bibinfo{year}{2021}).
\newblock


\bibitem[\protect\citeauthoryear{McKenna, Sheldon, and Miklau}{McKenna
  et~al\mbox{.}}{2019}]%
        {mckenna2019graphical}
\bibfield{author}{\bibinfo{person}{Ryan McKenna}, \bibinfo{person}{Daniel
  Sheldon}, {and} \bibinfo{person}{Gerome Miklau}.}
  \bibinfo{year}{2019}\natexlab{}.
\newblock \showarticletitle{Graphical-model based estimation and inference for
  differential privacy}. In \bibinfo{booktitle}{\emph{International Conference
  on Machine Learning}}. \bibinfo{pages}{4435--4444}.
\newblock
\urldef\tempurl%
\url{http://proceedings.mlr.press/v97/mckenna19a.html}
\showURL{%
\tempurl}


\bibitem[\protect\citeauthoryear{Nikolov, Talwar, and Zhang}{Nikolov
  et~al\mbox{.}}{2013}]%
        {nikolov2013geometry}
\bibfield{author}{\bibinfo{person}{Aleksandar Nikolov}, \bibinfo{person}{Kunal
  Talwar}, {and} \bibinfo{person}{Li Zhang}.} \bibinfo{year}{2013}\natexlab{}.
\newblock \showarticletitle{The geometry of differential privacy: the sparse
  and approximate cases}. In \bibinfo{booktitle}{\emph{Proceedings of the
  forty-fifth annual ACM symposium on Theory of computing}}.
  \bibinfo{pages}{351--360}.
\newblock


\bibitem[\protect\citeauthoryear{Papernot and Steinke}{Papernot and
  Steinke}{2021}]%
        {papernot2021hyperparameter}
\bibfield{author}{\bibinfo{person}{Nicolas Papernot} {and}
  \bibinfo{person}{Thomas Steinke}.} \bibinfo{year}{2021}\natexlab{}.
\newblock \showarticletitle{Hyperparameter Tuning with Renyi Differential
  Privacy}.
\newblock \bibinfo{journal}{\emph{arXiv preprint arXiv:2110.03620}}
  (\bibinfo{year}{2021}).
\newblock


\bibitem[\protect\citeauthoryear{Ridgeway, Theofanos, Manley, Task,
  et~al\mbox{.}}{Ridgeway et~al\mbox{.}}{2021}]%
        {ridgeway2021challenge}
\bibfield{author}{\bibinfo{person}{Diane Ridgeway}, \bibinfo{person}{Mary
  Theofanos}, \bibinfo{person}{Terese Manley}, \bibinfo{person}{Christine
  Task}, {et~al\mbox{.}}} \bibinfo{year}{2021}\natexlab{}.
\newblock \showarticletitle{Challenge Design and Lessons Learned from the 2018
  Differential Privacy Challenges}.
\newblock  (\bibinfo{year}{2021}).
\newblock


\bibitem[\protect\citeauthoryear{Rogers, Roth, Ullman, and Vadhan}{Rogers
  et~al\mbox{.}}{2016}]%
        {rogers2016privacy}
\bibfield{author}{\bibinfo{person}{Ryan~M Rogers}, \bibinfo{person}{Aaron
  Roth}, \bibinfo{person}{Jonathan Ullman}, {and} \bibinfo{person}{Salil
  Vadhan}.} \bibinfo{year}{2016}\natexlab{}.
\newblock \showarticletitle{Privacy odometers and filters: Pay-as-you-go
  composition}.
\newblock \bibinfo{journal}{\emph{Advances in Neural Information Processing
  Systems}}  \bibinfo{volume}{29} (\bibinfo{year}{2016}),
  \bibinfo{pages}{1921--1929}.
\newblock


\bibitem[\protect\citeauthoryear{Tantipongpipat, Waites, Boob, Siva, and
  Cummings}{Tantipongpipat et~al\mbox{.}}{2019}]%
        {tantipongpipat2019differentially}
\bibfield{author}{\bibinfo{person}{Uthaipon Tantipongpipat},
  \bibinfo{person}{Chris Waites}, \bibinfo{person}{Digvijay Boob},
  \bibinfo{person}{Amaresh~Ankit Siva}, {and} \bibinfo{person}{Rachel
  Cummings}.} \bibinfo{year}{2019}\natexlab{}.
\newblock \showarticletitle{Differentially Private Mixed-Type Data Generation
  For Unsupervised Learning}.
\newblock \bibinfo{journal}{\emph{CoRR}}  \bibinfo{volume}{abs/1912.03250}
  (\bibinfo{year}{2019}).
\newblock
\showeprint[arXiv]{1912.03250}
\urldef\tempurl%
\url{http://arxiv.org/abs/1912.03250}
\showURL{%
\tempurl}


\bibitem[\protect\citeauthoryear{Tao, McKenna, Hay, Machanavajjhala, and
  Miklau}{Tao et~al\mbox{.}}{2021}]%
        {tao2021benchmarking}
\bibfield{author}{\bibinfo{person}{Yuchao Tao}, \bibinfo{person}{Ryan McKenna},
  \bibinfo{person}{Michael Hay}, \bibinfo{person}{Ashwin Machanavajjhala},
  {and} \bibinfo{person}{Gerome Miklau}.} \bibinfo{year}{2021}\natexlab{}.
\newblock \showarticletitle{Benchmarking Differentially Private Synthetic Data
  Generation Algorithms}.
\newblock \bibinfo{journal}{\emph{Third AAAI Privacy-Preserving Artificial
  Intelligence (PPAI-22) workshop}} (\bibinfo{year}{2021}).
\newblock


\bibitem[\protect\citeauthoryear{Torfi, Fox, and Reddy}{Torfi
  et~al\mbox{.}}{2022}]%
        {torfi2020differentially}
\bibfield{author}{\bibinfo{person}{Amirsina Torfi}, \bibinfo{person}{Edward~A
  Fox}, {and} \bibinfo{person}{Chandan~K Reddy}.}
  \bibinfo{year}{2022}\natexlab{}.
\newblock \showarticletitle{Differentially private synthetic medical data
  generation using convolutional gans}.
\newblock \bibinfo{journal}{\emph{Information Sciences}}  \bibinfo{volume}{586}
  (\bibinfo{year}{2022}), \bibinfo{pages}{485--500}.
\newblock


\bibitem[\protect\citeauthoryear{Torkzadehmahani, Kairouz, and
  Paten}{Torkzadehmahani et~al\mbox{.}}{2019}]%
        {torkzadehmahani2019dp}
\bibfield{author}{\bibinfo{person}{Reihaneh Torkzadehmahani},
  \bibinfo{person}{Peter Kairouz}, {and} \bibinfo{person}{Benedict Paten}.}
  \bibinfo{year}{2019}\natexlab{}.
\newblock \showarticletitle{{DP-CGAN:} Differentially Private Synthetic Data
  and Label Generation}. In \bibinfo{booktitle}{\emph{{IEEE} Conference on
  Computer Vision and Pattern Recognition Workshops, {CVPR} Workshops 2019,
  Long Beach, CA, USA, June 16-20, 2019}}. \bibinfo{publisher}{Computer Vision
  Foundation / {IEEE}}, \bibinfo{pages}{98--104}.
\newblock
\urldef\tempurl%
\url{https://doi.org/10.1109/CVPRW.2019.00018}
\showDOI{\tempurl}


\bibitem[\protect\citeauthoryear{Tsagris, Beneki, and Hassani}{Tsagris
  et~al\mbox{.}}{2014}]%
        {tsagris2014folded}
\bibfield{author}{\bibinfo{person}{Michail Tsagris}, \bibinfo{person}{Christina
  Beneki}, {and} \bibinfo{person}{Hossein Hassani}.}
  \bibinfo{year}{2014}\natexlab{}.
\newblock \showarticletitle{On the folded normal distribution}.
\newblock \bibinfo{journal}{\emph{Mathematics}} \bibinfo{volume}{2},
  \bibinfo{number}{1} (\bibinfo{year}{2014}), \bibinfo{pages}{12--28}.
\newblock


\bibitem[\protect\citeauthoryear{Vietri, Tian, Bun, Steinke, and Wu}{Vietri
  et~al\mbox{.}}{2020}]%
        {vietri2020new}
\bibfield{author}{\bibinfo{person}{Giuseppe Vietri}, \bibinfo{person}{Grace
  Tian}, \bibinfo{person}{Mark Bun}, \bibinfo{person}{Thomas Steinke}, {and}
  \bibinfo{person}{Zhiwei~Steven Wu}.} \bibinfo{year}{2020}\natexlab{}.
\newblock \showarticletitle{New Oracle-Efficient Algorithms for Private
  Synthetic Data Release}. In \bibinfo{booktitle}{\emph{Proceedings of the 37th
  International Conference on Machine Learning, {ICML} 2020, 13-18 July 2020,
  Virtual Event}} \emph{(\bibinfo{series}{Proceedings of Machine Learning
  Research})}, Vol.~\bibinfo{volume}{119}. \bibinfo{publisher}{{PMLR}},
  \bibinfo{pages}{9765--9774}.
\newblock
\urldef\tempurl%
\url{http://proceedings.mlr.press/v119/vietri20b.html}
\showURL{%
\tempurl}


\bibitem[\protect\citeauthoryear{Xie, Lin, Wang, Wang, and Zhou}{Xie
  et~al\mbox{.}}{2018}]%
        {xie2018differentially}
\bibfield{author}{\bibinfo{person}{Liyang Xie}, \bibinfo{person}{Kaixiang Lin},
  \bibinfo{person}{Shu Wang}, \bibinfo{person}{Fei Wang}, {and}
  \bibinfo{person}{Jiayu Zhou}.} \bibinfo{year}{2018}\natexlab{}.
\newblock \showarticletitle{Differentially Private Generative Adversarial
  Network}.
\newblock \bibinfo{journal}{\emph{CoRR}}  \bibinfo{volume}{abs/1802.06739}
  (\bibinfo{year}{2018}).
\newblock
\showeprint[arXiv]{1802.06739}
\urldef\tempurl%
\url{http://arxiv.org/abs/1802.06739}
\showURL{%
\tempurl}


\bibitem[\protect\citeauthoryear{Xu, Ren, Zhang, Qin, and Ren}{Xu
  et~al\mbox{.}}{2017}]%
        {xu2017dppro}
\bibfield{author}{\bibinfo{person}{Chugui Xu}, \bibinfo{person}{Ju Ren},
  \bibinfo{person}{Yaoxue Zhang}, \bibinfo{person}{Zhan Qin}, {and}
  \bibinfo{person}{Kui Ren}.} \bibinfo{year}{2017}\natexlab{}.
\newblock \showarticletitle{DPPro: Differentially Private High-Dimensional Data
  Release via Random Projection}.
\newblock \bibinfo{journal}{\emph{IEEE Transactions on Information Forensics
  and Security}} \bibinfo{volume}{12}, \bibinfo{number}{12}
  (\bibinfo{year}{2017}), \bibinfo{pages}{3081--3093}.
\newblock
\urldef\tempurl%
\url{https://doi.org/10.1109/TIFS.2017.2737966}
\showDOI{\tempurl}


\bibitem[\protect\citeauthoryear{Zaykin}{Zaykin}{2011}]%
        {zaykin2011optimally}
\bibfield{author}{\bibinfo{person}{Dmitri~V Zaykin}.}
  \bibinfo{year}{2011}\natexlab{}.
\newblock \showarticletitle{Optimally weighted Z-test is a powerful method for
  combining probabilities in meta-analysis}.
\newblock \bibinfo{journal}{\emph{Journal of evolutionary biology}}
  \bibinfo{volume}{24}, \bibinfo{number}{8} (\bibinfo{year}{2011}),
  \bibinfo{pages}{1836--1841}.
\newblock


\bibitem[\protect\citeauthoryear{Zhang, Cormode, Procopiuc, Srivastava, and
  Xiao}{Zhang et~al\mbox{.}}{2017}]%
        {zhang2017privbayes}
\bibfield{author}{\bibinfo{person}{Jun Zhang}, \bibinfo{person}{Graham
  Cormode}, \bibinfo{person}{Cecilia~M. Procopiuc}, \bibinfo{person}{Divesh
  Srivastava}, {and} \bibinfo{person}{Xiaokui Xiao}.}
  \bibinfo{year}{2017}\natexlab{}.
\newblock \showarticletitle{PrivBayes: Private Data Release via Bayesian
  Networks}.
\newblock \bibinfo{journal}{\emph{ACM Transactions on Database Systems (TODS)}}
  \bibinfo{volume}{42}, \bibinfo{number}{4} (\bibinfo{year}{2017}),
  \bibinfo{pages}{25:1--25:41}.
\newblock
\urldef\tempurl%
\url{https://doi.org/10.1145/3134428}
\showDOI{\tempurl}


\bibitem[\protect\citeauthoryear{Zhang, Zhao, Wei, and Chen}{Zhang
  et~al\mbox{.}}{2019}]%
        {zhang2019differentially}
\bibfield{author}{\bibinfo{person}{Wei Zhang}, \bibinfo{person}{Jingwen Zhao},
  \bibinfo{person}{Fengqiong Wei}, {and} \bibinfo{person}{Yunfang Chen}.}
  \bibinfo{year}{2019}\natexlab{}.
\newblock \showarticletitle{Differentially Private High-Dimensional Data
  Publication via Markov Network}.
\newblock \bibinfo{journal}{\emph{{EAI} Endorsed Trans. Security Safety}}
  \bibinfo{volume}{6}, \bibinfo{number}{19} (\bibinfo{year}{2019}),
  \bibinfo{pages}{e4}.
\newblock
\urldef\tempurl%
\url{https://doi.org/10.4108/eai.29-7-2019.159626}
\showDOI{\tempurl}


\bibitem[\protect\citeauthoryear{Zhang, Ji, and Wang}{Zhang
  et~al\mbox{.}}{2018}]%
        {zhang2018differentially}
\bibfield{author}{\bibinfo{person}{Xinyang Zhang}, \bibinfo{person}{Shouling
  Ji}, {and} \bibinfo{person}{Ting Wang}.} \bibinfo{year}{2018}\natexlab{}.
\newblock \showarticletitle{Differentially private releasing via deep
  generative model (technical report)}.
\newblock \bibinfo{journal}{\emph{arXiv preprint arXiv:1801.01594}}
  (\bibinfo{year}{2018}).
\newblock
\urldef\tempurl%
\url{https://arxiv.org/abs/1801.01594}
\showURL{%
\tempurl}


\bibitem[\protect\citeauthoryear{Zhang, Wang, Li, Honorio, Backes, He, Chen,
  and Zhang}{Zhang et~al\mbox{.}}{2021}]%
        {zhang2020privsyn}
\bibfield{author}{\bibinfo{person}{Zhikun Zhang}, \bibinfo{person}{Tianhao
  Wang}, \bibinfo{person}{Ninghui Li}, \bibinfo{person}{Jean Honorio},
  \bibinfo{person}{Michael Backes}, \bibinfo{person}{Shibo He},
  \bibinfo{person}{Jiming Chen}, {and} \bibinfo{person}{Yang Zhang}.}
  \bibinfo{year}{2021}\natexlab{}.
\newblock \showarticletitle{PrivSyn: Differentially Private Data Synthesis}. In
  \bibinfo{booktitle}{\emph{30th {USENIX} Security Symposium ({USENIX} Security
  21)}}. \bibinfo{publisher}{{USENIX} Association}, \bibinfo{pages}{929--946}.
\newblock
\showISBNx{978-1-939133-24-3}
\urldef\tempurl%
\url{https://www.usenix.org/conference/usenixsecurity21/presentation/zhang-zhikun}
\showURL{%
\tempurl}


\end{thebibliography}
\balance


\clearpage
\appendix
\newpage
\section{Data Preprocessing} \label{sec:preprocessing}

We apply consistent preprocessing to all datasets in our empirical evaluation.  There are three steps to our preprocessing procedure, described below:
\paragraph*{\textbf{Attribute selection}}
For each dataset, we identify a set of attributes to keep.  For the \adult, \salary, \nltcs, and \titanic datasets, we keep all attributes from the original data source.  For the \fire dataset, we drop the 15 attributes relating to incident times, since after discretization, they contain redundant information.  The \msnbc dataset is a streaming dataset, where each row has a different number of entries.  We keep only the first 16 entries for each row.  

\paragraph*{\textbf{Domain identification}}
Usually we expect the domain to be supplied separately from the data file.  For example, the IPUMS website contains comprehensive documentation about U.S. Census data products.  However, for the datasets we used, no such domain file was available.  Thus, we ``cheat'' and look at the active domain to automatically derive a domain file from the dataset.  For each attribute, we identify if it is categorical or numerical.  For each categorical attribute, we list the set of observed values (including null) for that attribute, which we treat as the set of possible values for that attribute.  For each numerical attribute, we record the minimum and maximum observed value for that attribute.  

\paragraph*{\textbf{Discretization}}
We discretize each numerical attribute into $32$ equal-width bins, using the min/max values from the domain file.  This turns each numerical attribute into a categorical attribute, satisfying our assumption.  

\section{Uncertainty Quantification Proofs} \label{sec:proofs}

\subsection{The Easy Case: Supported Marginals}

\estimator*

\begin{proof}
For each $ r_i \supseteq r $, we observe $\tilde{y}_i \sim M_{r_i}(D) + \mathcal{N}(0, \sigma_i^2 \mathbb{I}) $.  We can use this noisy marginal to obtain an unbiased estimate $M_{r}(D)$ by marginalizing out attributes in the set $ r_i \setminus r$.  This requires summing up $n_{r_i} / n_r$ cells, so the variance in each cell becomes $n_{r_i} \sigma_i^2 / n_r$.  Moreover, the noise is still normally distributed, since the sum of independent normal random variables is normal.  We thus have such an estimate for each $i$ satisfying $r_i \supseteq r$, and we can combine these independent estimates using \emph{inverse variance weighting} \cite{hartung2008statistical}, resulting in an unbiased estimator with the stated variance.  For the same reason as before, the noise is still normally distributed.
\end{proof}

\supported*

\begin{proof}
Noting that $M_r(D) - \bar{y} \sim \mathcal{N}(0, \sigma^2 \mathbb{I})$, the statement is a direct consequence of \cref{thm:halfnorm}, below.
\end{proof}

\begin{theorem} \label{thm:halfnorm}
Let $x \sim N(0, \sigma^2)^n$, then:
$$ \mathbb{E}[\norm{x}_1] = \sqrt{2 / \pi} n \sigma $$
and
$$ \Pr[\norm{x}_1 \geq \sqrt{2 \log{2}} \sigma n + c \sigma \sqrt{2n}] \leq \exp{(-c^2)} $$
\end{theorem}

\begin{proof}

First observe that $ | x_i | $ is a sample from a \emph{half-normal} distribution.  Thus, $ \mathbb{E}[x_i] = \sqrt{2/\pi} \sigma $.  From the linearity of expectation, we obtain $ \mathbb{E}[\norm{x}_1] = \sqrt{2/\pi} \sigma n $, as desired.
For the second statement, we begin by deriving the moment generating function of the random variable $ | x_i | $.  By definition, we have:
\begin{align*}
\mathbb{E}[\exp{(t \cdot | x_i |)}] &= \int_{-\infty}^{\infty} \phi(z) \exp{(t \cdot | z |)} dz \\
&= 2 \int_0^{\infty} \phi(z) \exp{(t \cdot z)} dz \\
&= 2 \int_0^{\infty} \frac{1}{\sigma \sqrt{2 \pi}} \exp{\Big(-\frac{z^2}{2 \sigma^2} \Big)} \exp{(t \cdot z)} dz \\
&= \frac{1}{\sigma} \sqrt{\frac{2}{\pi}} \int_0^{\infty} \exp{\Big(-\frac{z^2}{2 \sigma^2} + t \cdot z \Big)} dz \\
&= \exp{ \Big(\frac{\sigma^2 t^2}{2} \Big)} \Big(\Phi\Big(\frac{t \sigma}{\sqrt{2}}\Big) + 1 \Big) \\
\end{align*}
Moreover, since $\norm{x}_1 = \sum_{i=1}^n | x_i | $ is a sum of i.i.d random variables, the moment generating function of $\norm{x}_1$ is:
$$ \mathbb{E}[\exp{(t \cdot \norm{x}_1)}] = \exp{ \Big(\frac{\sigma^2 t^2}{2} \Big)}^n \Big(\Phi\Big(\frac{t \sigma}{\sqrt{2}}\Big) + 1 \Big)^n $$
From the Chernoff bound, we have
\begin{align*}
\Pr[\norm{x}_1 \geq a] &\leq \min_{t \geq 0} \frac{\mathbb{E}[\exp{(t \cdot \norm{x}_1 )}]}{\exp{(t a)}} \\
&= \min_{t \geq 0} \exp{ \Big(\frac{n \sigma^2 t^2}{2} - t a \Big)} \Big(\Phi\Big(\frac{t \sigma}{\sqrt{2}}\Big) + 1 \Big)^n \\
&\leq \min_{t \geq 0} 2^n \exp{ \Big(\frac{n \sigma^2 t^2}{2} - t a \Big)} \\
&\leq 2^n \exp{\Big( \frac{n \sigma^2 (a / n \sigma^2)^2}{2} - (a / n \sigma^2) a}\Big) \\
&= 2^n \exp{\Big( \frac{a^2}{2 n \sigma^2} - \frac{a^2}{n \sigma^2}\Big)} \\
&= 2^n \exp{\Big(- \frac{a^2}{2 n \sigma^2} \Big)} \\
&= \exp{\Big(- \frac{a^2}{2 n \sigma^2} + n \log{2} \Big)}
\end{align*}
With some further manipulation of the bound, we obtain:
\begin{align*}
&\Pr[\norm{x}_1 \geq d \sigma \sqrt{2 n} ] \leq \exp{\Big(-d^2 + n \log{2} \Big)} \tag{$a=d \sigma \sqrt{2n}$} \\
&\Pr[\norm{x}_1 \geq (c + \sqrt{n \log{2}}) \sigma \sqrt{2 n}] \leq \exp{(-c^2)} \tag{$d = c + \sqrt{n \log{2}}$}\\
&\Pr[\norm{x}_1 \geq \sqrt{2 \log{2}} \sigma n + c \sigma \sqrt{2 n}] \leq \exp{(-c^2)} \\
\end{align*}
\end{proof}

\subsection{The Hard Case: Unsupported Marginals}

\unsupported*
\begin{proof}
%
%



By the guarantees of the exponential mechanism, we know that, with probability at most $ e^{-\lambda_2} $, for all $r \in C_t$ we have:
$$ q_{r_t} \leq q_r - \frac{2 \Delta_t}{\epsilon_t} (\log{(| C_t |)} + \lambda_2) $$
Now define $E_r = \norm{M_r(D) - M_r(p_{t-1})}_1$.  Plugging in $q_r = w_r (E_r - \sqrt{2/\pi} \sigma_t n_r)$ and rearranging gives:

$$ E_r \geq \frac{w_{r_t} (E_{r_t} - \sqrt{2/\pi} \sigma_t n_{r_t}) + \frac{2\Delta_t}{\epsilon_t} (\log{(|C_t|) + \lambda_2)}}{w_r} + \sqrt{2/\pi} \sigma_t n_r $$

From \cref{thm:hdline2}, with probability at most $e^{-\lambda_1^2/2}$, we have:
$$ \norm{M_{r_t}(p_{t-1}) - y_t}_1 + \lambda_1 \sigma_t \sqrt{n_{r_t}} \leq E_{r_t} $$
Combining these two facts via the union bound, along with some algebraic manipulation, yields the stated result.  
\end{proof}


\begin{theorem} \label{thm:hdline2}
Let $a, b \in \mathbb{R}^k$ and let $c = b + z$ where $z \sim \mathcal{N}(0, \sigma^2)^n$.
$$ \Pr[\norm{a-c}_1 \leq \norm{a-b}_1 - \lambda \sigma \sqrt{n}] \leq \exp{\Big(-\frac{1}{2} \lambda^2\Big)} $$
\end{theorem}

\begin{proof}

First note that $|a_i - c_i|=|a_i - b_i - z_i|$, which is distributed according to a folded normal distribution with mean $|a_i - b_i|$.  It is well known \cite{tsagris2014folded} that the moment generating function for this random variable is $M_i(t)$, where:
\begin{align*}
M_i(t) &= \exp{\Big(\frac{1}{2} \sigma^2 t^2 + | a_i - b_i | t \Big)} \Phi(|a_i - b_i|/\sigma + \sigma t) \\ 
&+ \exp{ \Big(\frac{1}{2} \sigma^2 t^2 - | a_i - b_i | t \Big)} \Phi(-|a_i - b_i|/\sigma + \sigma t). \\
\end{align*}
Moreover, the moment generating function of $\norm{a-c}_1$ is $M(t) = \prod_i M_i(t)$.  We will begin by focusing our attention on bounding $M_i(-t)$.  For simplicity, let $ \mu = |a_i - b_i|$.  We have:
\begin{align*}
M_i(-t) =& \exp{\Big( \frac{\sigma^2 t^2}{2} - \mu t\Big)} \Phi(\mu/\sigma - \sigma t) \\
&+ \exp{\Big(\frac{\sigma^2 t^2}{2} + \mu t \Big)} \Phi(-\mu/\sigma - \sigma t) \\
=& \exp{\Big( \frac{\sigma^2 t^2}{2} - \mu t\Big)} (1 - \Phi(-\mu/\sigma + \sigma t)) \\
&+ \exp{\Big(\frac{\sigma^2 t^2}{2} + \mu t \Big)} \Phi(-\mu/\sigma - \sigma t) \\
=& \exp{\Big( \frac{\sigma^2 t^2}{2} - \mu t\Big)} \\
&- \exp{\Big(\frac{\sigma^2 t^2}{2} - \mu t \Big)} \Phi(-\mu/\sigma + \sigma t) \\
&+ \exp{\Big(\frac{\sigma^2 t^2}{2} + \mu t \Big)} \Phi(-\mu/\sigma - \sigma t) \\
\leq& \exp{\Big( \frac{\sigma^2 t^2}{2} - \mu t\Big)} \tag{\cref{conjecture} below; $a=\sigma t, b=\mu/\sigma$}\\
\end{align*}
We are now ready to plug this result into the Chernoff bound, which states:
\begin{align*}
\Pr[\norm{a-c}_1 \leq r] &\leq \min_{t \geq 0} \exp{(t \cdot r)} M(-t) \\
&\leq \min_{t \geq 0} \exp{(t \cdot r)} \prod_i \exp{\Big(\frac{\sigma^2 t^2}{2} - | a_i - b_i | t \Big)} \\
&= \min_{t \geq 0} \exp{(t \cdot r + \frac{n \sigma^2 t^2}{2} - \norm{a-b}_1 t)} \\
\end{align*}
Setting $r = \norm{a-b}_1 - \lambda \sigma \sqrt{n}$ gives the desired result
\begin{align*}
\Pr&[\norm{a-c}_1 \leq \norm{a-b}_1 - \lambda \sigma \sqrt{n}] \\
&\leq \min_{t \geq 0} \exp{(t \cdot (\norm{a-b}_1 - \lambda \sigma \sqrt{n}) + \frac{n \sigma^2 t^2}{2} - \norm{a-b}_1 t)} \\
&= \min_{t \geq 0} \exp{\Big(-t \lambda \sigma \sqrt{n} + \frac{n \sigma^2 t^2}{2}\Big)} \\
&\leq \exp{(-\lambda^2 / 2)} \tag{set $t  = \lambda / \sigma \sqrt{n}$}\\
\end{align*}
\end{proof}

\begin{lemma} \label{conjecture}
Let $a,b \geq 0$, and let $\Phi$ denote the CDF of the standard normal distribution.  Then,
$$ \exp{\Big(\frac{1}{2} a^2 + ab \Big)} \Phi(-a-b) \leq \exp{\Big(\frac{1}{2}a^2 - ab \Big)} \Phi(a-b) $$
\end{lemma}

\begin{proof}
First observe that:
\begin{align*}
\exp{\Big(\frac{1}{2} a^2 + ab \Big)} \Phi(-a-b)  &= \exp{\Big(-\frac{1}{2} b^2 \Big)} \frac{\Phi(-a-b)}{\phi(-a-b)} \\
\exp{\Big(\frac{1}{2} a^2 - ab \Big)} \Phi(a-b)  &= \exp{\Big(-\frac{1}{2} b^2 \Big)} \frac{\Phi(a-b)}{\phi(a-b)} \\
\end{align*}
Since $a,b \geq 0$, we know that $-a-b \leq a-b$.  We will now argue that the function $\frac{\Phi(\alpha)}{\phi(\alpha)}$ is monotonically increasing in $\alpha$, which suffices to prove the desired claim.  To prove this, we will observe that this is this quantity is known as the \emph{Mills ratio} \cite{greene2003econometric} for the normal distribution.  We know that the Mills ratio is connected to a particular expectation; specifically, if $X \sim \mathcal{N}(0,1)$, then
$$ \mathbb{E}[X \mid X < \alpha] = -\frac{\phi(\alpha)}{\Phi(\alpha)} $$
Using this interpretation, it is clear that the LHS (and hence the RHS) is monotonically increasing in $\alpha$.  Since $-\frac{\phi(\alpha)}{\Phi(\alpha)}$ is monotonically increasing, so is $\frac{\Phi(\alpha)}{\phi(\alpha)}$.  
\end{proof}

\section{Interpretable Error Rate and Subsampling Mechanism}

\begin{figure*}[t!]
\centering
\subcaptionbox{\vspace{1em} General}{\includegraphics[width=0.335\textwidth]{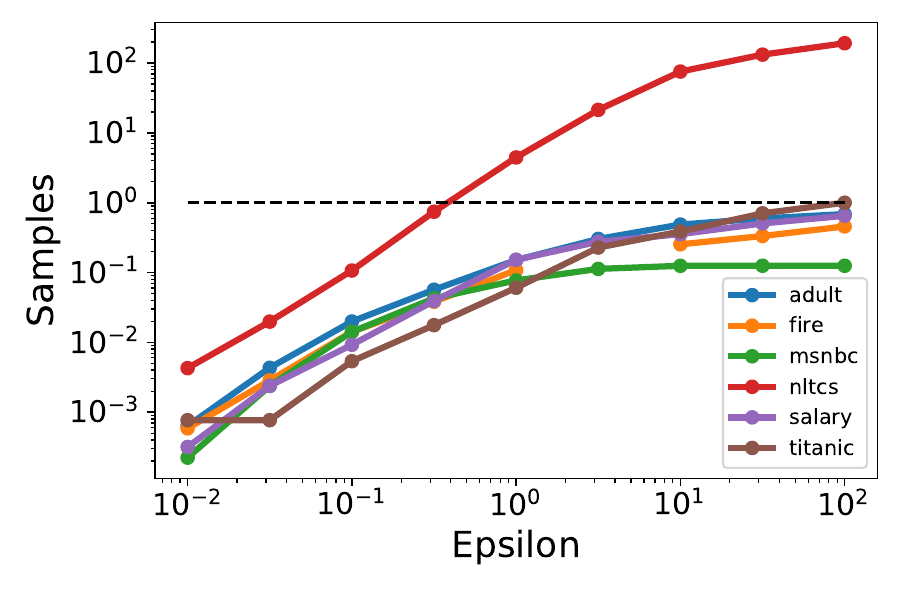}}
\subcaptionbox{Target}{\includegraphics[width=0.318\textwidth]{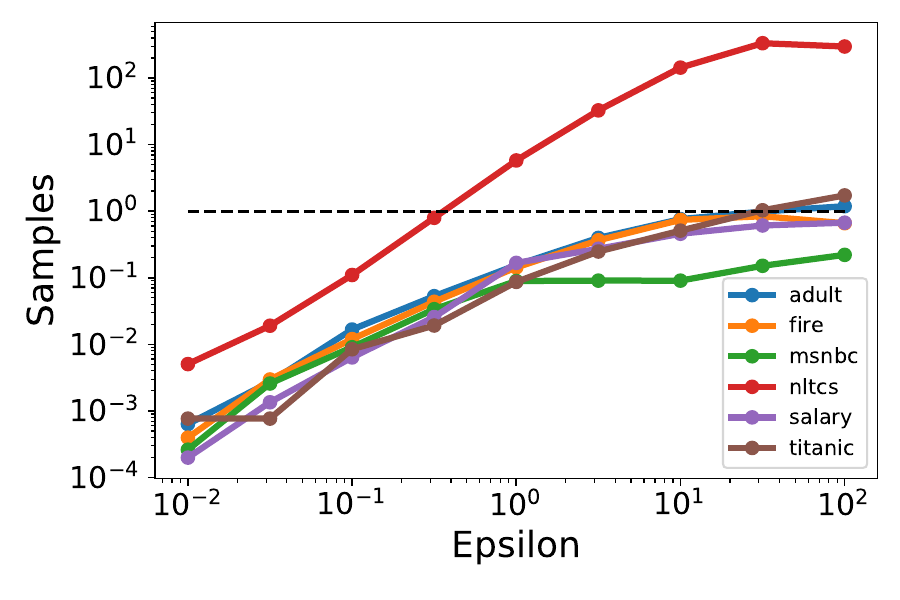}}
\subcaptionbox{Weighted}{\includegraphics[width=0.318\textwidth]{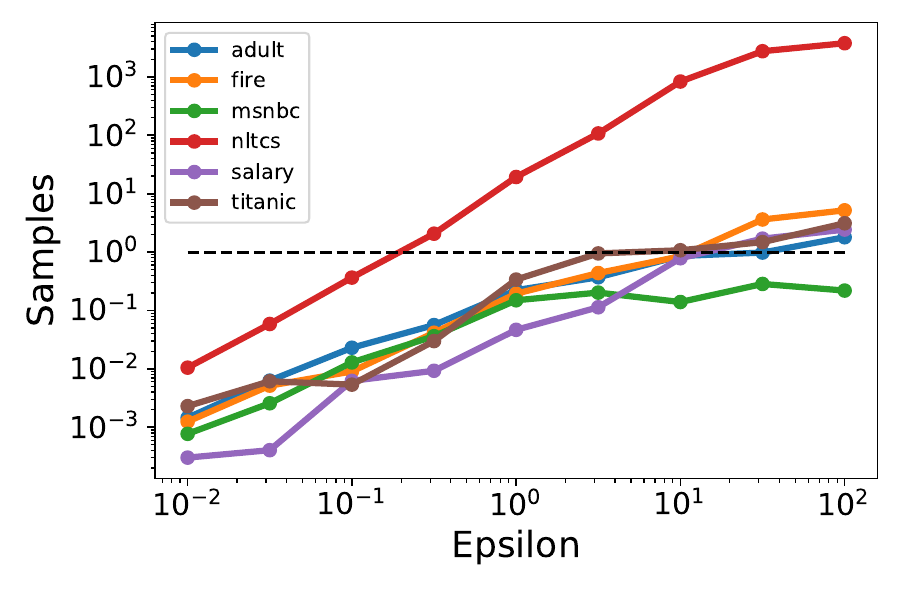}}
\caption{\label{fig:samples} Performance of \aim as measured by the number of samples needed to match the achieved workload error.}
\end{figure*}

In \cref{sec:experiments}, we saw that \aim offers the best error relative to existing synthetic data mechanisms, although it is not obvious whether a given $L_1$ error should be considered ``good''.  This is necessary for setting the privacy parameters to strike the right privacy/utility tradeoff.  We can bring more clarity to this problem by comparing \aim to a (non-private) baseline that simply resamples $K$ records from the dataset.  Then, if \aim achieves the same error as resampling $K = \frac{N}{2}$ records, this provides a clear interpretation: that the price of privacy is losing about half the data.  Due to the simplicity of this baseline, we can compute the expected workload error in closed form, without actually running the mechanism.  We provide details of these calculations in the next section.  

\cref{fig:samples} plots the performance of \aim on each dataset, epsilon, and workload considered, measured using the \emph{fraction} of samples needed for the subsampling mechanism to match the performance of \aim.  These plots reveal that at $\epsilon=10$, the median subsampling fraction is about $0.37$ for the \textsc{general} workload, $0.62$ for the \textsc{target} workload, and $0.85$ for the \textsc{weighted} workload.  At $\epsilon=1$, these numbers are $0.13$, $0.15$, and $0.21$, respectively. The results are comparable across five out of six datasets, with \textsc{nltcs} being a clear outlier.  For that dataset, a subsampling fraction of 1.0 was reached by $\epsilon = 0.31$ for all workload.  This could be an indication of overfitting to the data; a possible reason for this behavior is that the domain size of the \textsc{nltcs} data is small compared to the number of records.  \textsc{mnsbc} is also an outlier to a lesser extent, with worse performance than the other datasets for larger $\epsilon$.  A possible reason for this behavior is that \textsc{msnbc} has the most data points, so subsampling with the same fraction of points has much lower error.  \aim may not be able to match that low error due to the computational constraints imposed on the model size, combined with the fact that this dataset has a large domain.

\subsection{Mathematical Details of Subsampling}
We begin by analyzing the expected workload error of the (non-private) mechanism that randomly samples $K$ items with replacement from $D$.  Then, we will connect that to the error of \aim, and determine the value of $K$ where the error rates match.  \cref{thm:sampling} gives a closed form expression for the expected $L_1$ error on a single marginal as a function of the number of sampled records.  

\begin{theorem} \label{thm:sampling}
Let $\hat{D}$ be the dataset obtained by sampling $K$ items with replacement from $D$.  Further, let $\vec{\mu} = \frac{1}{N} M_r(D)$ and $\vec{s} = \ceil{K \vec{\mu}}$.
\begin{align*}
\mathbb{E} &\Big[\norm{ \frac{1}{N} M_r(D) - \frac{1}{K} M_r(\hat{D})} \Big] =\\
& \frac{2}{K} \sum_{x \in \dom_r} s(x) \binom{K}{s(x)} \mu(x)^{s(x)} (1-\mu(x))^{K-s(x)+1} \\
\end{align*}
\end{theorem}

\begin{proof}
The theorem statement follows directly from \cref{lem:sampling} and \cref{lem:multinomial}.
\end{proof}

\begin{lemma}[Mean Deviation \cite{frame1945mean,johnson2005univariate}] \label{lem:binomial}
Let $k \sim Binomial(n, p)$, then:
$$ \mathbb{E}\Big[ \Big| p - \frac{k}{n} \Big| \Big] =  \frac{2}{n} s \binom{n}{s} p^s (1-p)^{n-s+1}, $$
where $s = \ceil{n \cdot p}$.
\end{lemma}

\begin{proof}
This statement appears and is proved in \cite{frame1945mean,johnson2005univariate}.
\end{proof}

\begin{lemma}[$L_1$ Deviation] \label{lem:multinomial}
Let $\vec{k} \sim Multinomial(n, \vec{p})$, then:
$$ \mathbb{E}[ \norm{\vec{p} - \vec{k}/n}_1 ] = \frac{2}{n} \sum_x s(x) \binom{n}{s(x)} p(x)^{s(x)} (1-p(x))^{n-s(x)+1}, $$
where $s(x) = \ceil{n \cdot p(x)}$.
\end{lemma}

\begin{proof}
The statement follows immediately from \cref{lem:binomial} and the fact that $k(x) \sim Binomial(n, p(x))$.  
\end{proof}

\begin{lemma} \label{lem:sampling}
Let $\hat{D}$ be the dataset obtained by sampling $K$ items with replacement from $D$.  Then,

$$ M_r(\hat{D}) \sim Multinomial\Big(K, \frac{1}{N} M_r(D)\Big) $$
\end{lemma}

\begin{proof}
The statement follows from the definition of the multinomial distribution.
\end{proof}

\section{Structural Zeros}

In this section, we describe a simple and principled method to specify and enforce \emph{structural zeros} in the mechanism.  These capture attribute combinations that cannot occur in the real data.  Without specifying this, synthetic data mechanisms will usually generate records that violate these constraints that hold in the real data as the process of adding noise can introduce spurious records, especially in high privacy regimes.   These spurious records can be confusing for downstream analysis of the synthetic data and can lead the analyst to distrust the quality of the data.  By imposing known structural zero constraints, we can avoid this problem while also improving the quality of the synthetic data on the workload of interest.

Structural zeros, if they exist, can usually be enumerated by a domain expert.  We can very naturally incorporate these into our mechanism with only one minor change to the underlying \pgm library.  These structural zeros can be specified as input as a list of pairs $(r, \mathcal{Z}_r)$ where $ \mathcal{Z}_r \subseteq \Omega_r$. The first entry of the pair specifies the set of attributes relevant to the structural zeros, while the second entry enumerates the attribute combinations whose counts should all be zero.  The method we propose can be used within any mechanism that builds on top of \pgm and is hence more broadly useful outside the context of \aim.

To understand the technical ideas in this section, please refer to the background on \pgm \cite{mckenna2019graphical}.  Usually \pgm is initialized by setting $\theta_r(x_r) = 0$ for all $r$ in the model and all $x_r \in \Omega_r$.  This corresponds to a model where $\mu_r(x_r)$ is uniform across all $x_r$.  Our basic observation is that by initializing \pgm by setting $ \theta_r(x_r) = -\infty$ for each $x_r \in Z_r$ the cell of the associated marginal will be $\mu_r(x_r) = 0$, as desired.  Moreover, each update within the \pgm estimation procedure will try to update $\theta_r(x_r)$ by a finite amount, leaving it unchanged.  Thus, $\mu_r(x_r)$ will remain $0$ during the entire estimation procedure.  We conjecture that the estimation procedure solves the following modified convex optimization problem:

$$ \hat{\mu} = \min_{\substack{\mu \in \mathcal{M} \\ \mu_r(Z_r) = 0}} L(\mu)$$

This approach is appealing because other simple approaches that discard invalid tuples can inadvertently bias the distribution, which is undesirable.

Note that for each clique in the set of structural zeros, we must include that clique in our model, which increases the size of that model.  Thus, we should treat it as we would treat a clique selected by \aim.  That is, when calculating \textsc{JT-SIZE} in line 12 of \aim, we need to include both the cliques selected in earlier iterations as well as the cliques included in the structural zeros.

\subsection{Experiments}

In this section, we empirically evaluate this structural zeros enhancement, showing that it can reduce workload error in some cases.  For this experiment, we consider the \textsc{general} workload on the \textsc{fire} dataset and compare the performance of \aim with and without imposing structural zero constraints.  This dataset contains several related attributes, like ``Zipcode of Incident`` and ``City''.  While these attributes are not perfectly correlated, significant numbers of attribute combinations are impossible.  We identified a total of nine attribute pairs which contain some structural zeros and a total of 2696 structural zero constraints within these nine marginals.

The results of this experiment are shown in \cref{table:zeros}.  On average, imposing structural zeros improves the performance of the mechanism, although the improvement is not universal across all values of epsilon we tested.   Nevertheless, it is still useful to impose these constraints for data quality purposes.

\begin{table}[H]
\caption{\label{table:zeros} Error of \aim on the \textsc{fire} dataset, with and \\ without imposing structural zero constraints.}
\begin{tabular}{c|ccc}
\toprule
$\epsilon$ &    AIM & AIM+Structural Zeros & Ratio \\
\midrule
0.010   &  0.613 &     0.542 &        1.130 \\
0.031   &  0.303 &     0.263 &        1.151 \\
0.100   &  0.141 &     0.153 &        0.924 \\
0.316   &  0.087 &     0.077 &        1.124 \\
1.000   &  0.052 &     0.053 &        0.979 \\
3.162   &  0.044 &     0.045 &        0.964 \\
10.00  &  0.038 &     0.032 &        1.170 \\
31.62  &  0.029 &     0.026 &        1.149 \\
100.0 &  0.025 &     0.025 &        1.004 \\
\bottomrule
\end{tabular}
\end{table}

\section{Runtime Experiments}

Our primary focus in the main body of the paper was mechanism utility, as measured by the workload error.  In this section we discuss the runtime of \aim, which is an important consideration when deploying it in practice.  Note that we do not compare against runtime of other mechanisms here, because different mechanisms were executed in different runtime environments.  \cref{fig:runtime} below shows the runtime of \aim as a function of the privacy parameter.  As evident from the figure, runtime increases drastically with the privacy parameter.  This is not surprising because \aim is budget-aware: it knows to select larger marginals and run for more rounds when the budget is higher, which in turn leads to longer runtime.  For large $\epsilon$, the constraint on \jtsize is essential to allow the mechanism to terminate at all.  Without it, \aim may try to select marginal queries that exceed memory resources and result in much longer runtime.  For small $\epsilon$, this constraint is not active, and could be removed without affecting the behavior of \aim.  

Recall that these experiments were conducted on one core of a compute cluster with 4 GB of memory and a CPU speed of 2.4 GHz.  These machines were used due to the large number of experiments we needed to conduct, but in real-world scenarios we only need to run one execution of \aim, for a single dataset, workload, privacy parameter, and trial.  For this, we can use machines with much better specs, which would improve the runtime significantly.  

\begin{figure}[H]
\includegraphics[width=\linewidth]{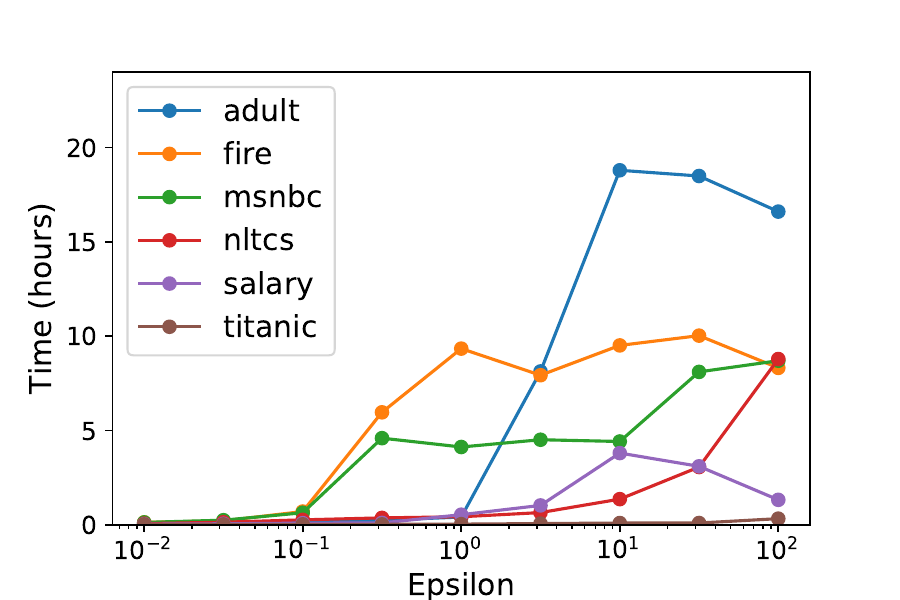}
\caption{Runtime of \aim on the \general workload.} \label{fig:runtime}
\end{figure}

\section{\pgm vs. Relaxed Projection}

\begin{figure*}[t!]
\centering
\subcaptionbox{\vspace{1em}$\epsilon = 0.1$}{\includegraphics[width=0.33\textwidth]{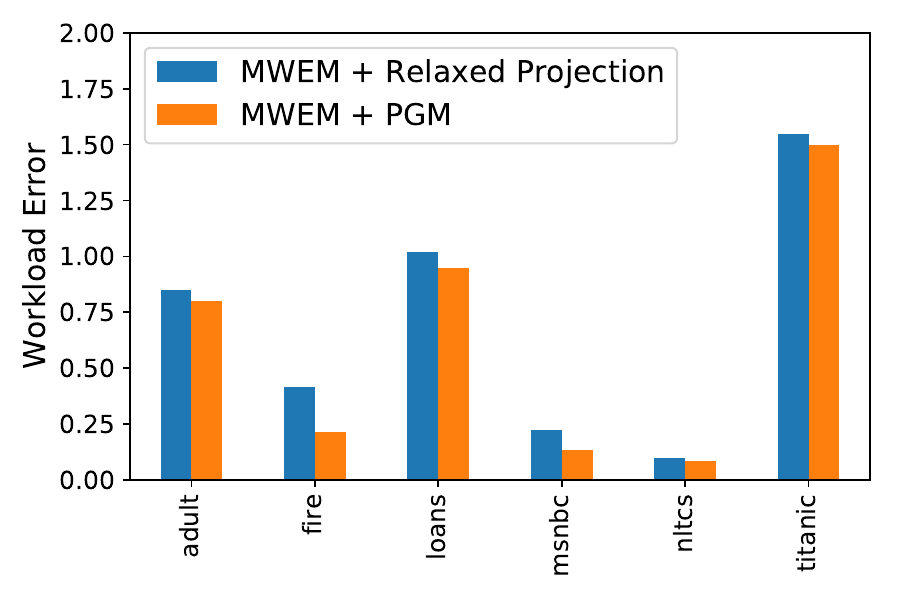}}
\subcaptionbox{$\epsilon = 1.0$}{\includegraphics[width=0.33\textwidth]{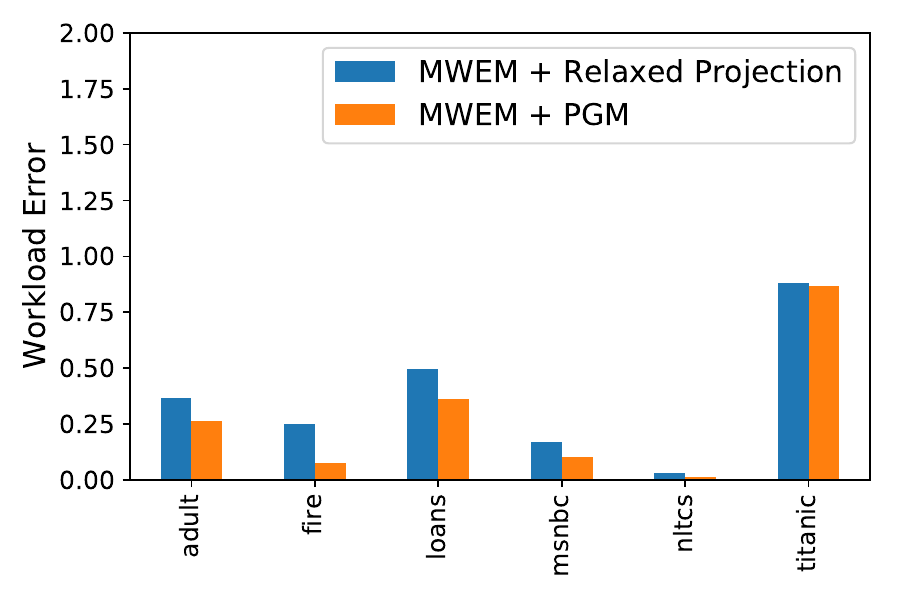}}
\subcaptionbox{$\epsilon = 10.0$}{\includegraphics[width=0.33\textwidth]{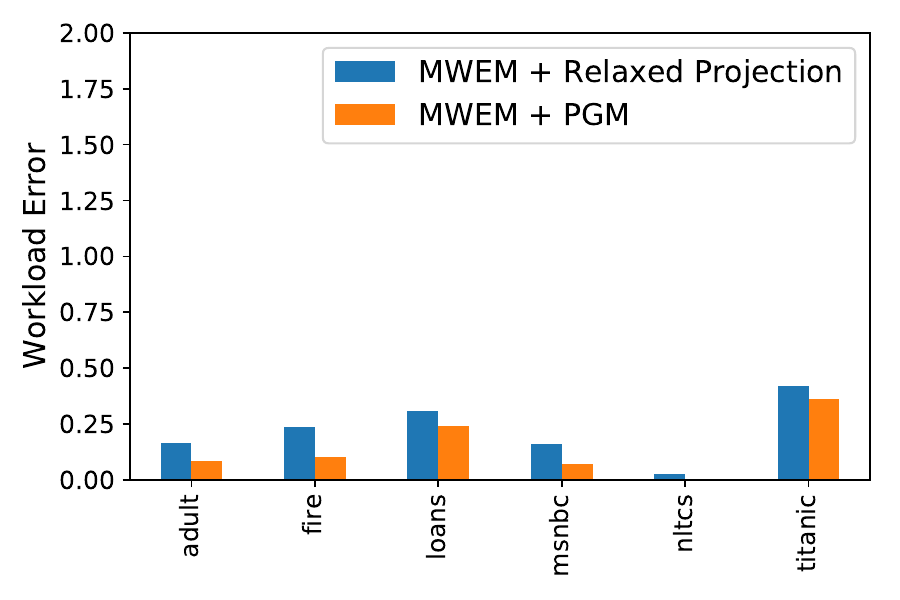}}
\caption{\label{fig:pgmmp} \mech{MWEM+Relaxed Projection} vs. \mwempgm on the \general workload.}
\end{figure*}

In this paper, we built \aim on top of \pgm, leveraging prior work for the \textbf{generate} step of the select-measure-generate paradigm.  \pgm is not the only method in this space, although it was the first general purpose and scalable method to our knowledge.  ``Relaxed Projection'' \cite{aydore2021differentially} is another general purpose and scalable method that solves the same problem, and could be used in place of \pgm if desired.  \rap, the main mechanism that utilizes this technique, did not perform well in our experiments.  However, it is not clear from our experiments if the poor performance can be attributed to the relaxed projection algorithm, or some other algorithmic design decisions.  In this section, we attempt to precisely pin down the differences between these two related methods, taking care to fix possible confounding factors.  We thus consider two mechanisms: \mwempgm, which is defined in \cref{alg:mwem}, and \mech{MWEM+Relaxed Projection} which is identical to \mwempgm in every way, except the call to \pgm is replaced with a call to the relaxed projection algorithm of Aydore et al. 

For this experiment, we consider the \general workload, and we run each algorithm for $T = 5, 10, \dots, 100$, with five trials for each hyper-parameter setting.  We average the workload error across the five trials, and report the minimum workload error across hyper-parameter settings in \cref{fig:pgmmp}.  Although the algorithms are conceptually very similar, \mwempgm consistently outperforms \mech{MWEM+Relaxed Projection}, across every dataset and privacy level considered.  The performance difference is modest in many cases, but significant on the \fire dataset.  

\mech{AP-PGM} \cite{mckenna2021relaxed} offers another alternative to \pgm for the generate step, and while it was shown to be an appealing alternative to \pgm in some cases, within the context of an \mwem-style algorithm, their own experiments demonstrate the superiority of \pgm.  

Generator networks \cite{liu2021iterative} offer yet another alternative to \pgm for the generate step.  To the best of our knowledge, no direct comparison between this approach and \pgm has been done to date, where confounding factors are controlled for. Conceptually, this approach is most similar to the relaxed projection approach, so we conjecture the results to look similar to those shown in \cref{fig:pgmmp}.

\revision{\section{Tuning Model Capacity} \label{sec:capacity}

\begin{figure*}[t!]
\centering
\subcaptionbox{Workload Error vs. Model Capacity}{\includegraphics[width=0.34\textwidth]{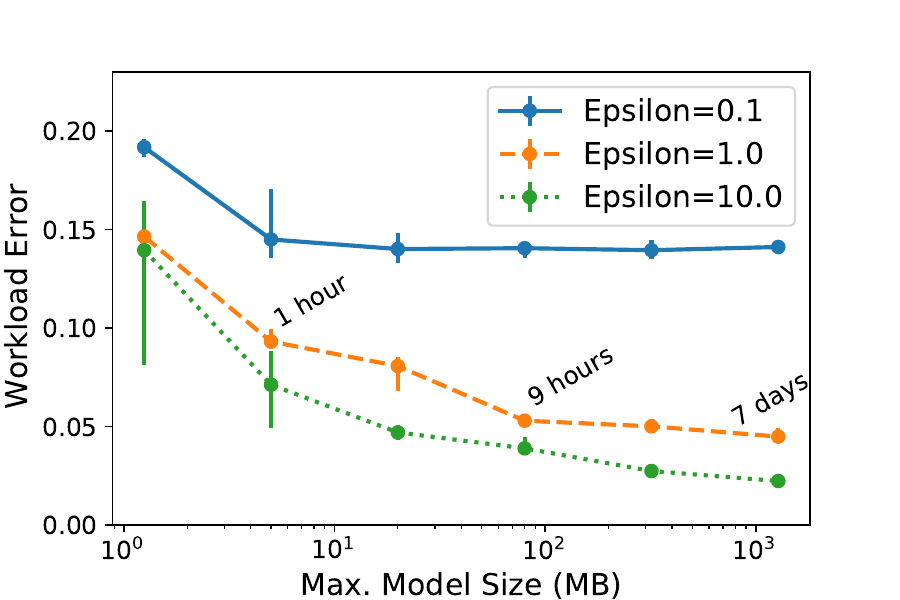}}
\subcaptionbox{Runtime vs. Model Capacity}{\includegraphics[width=0.34\textwidth]{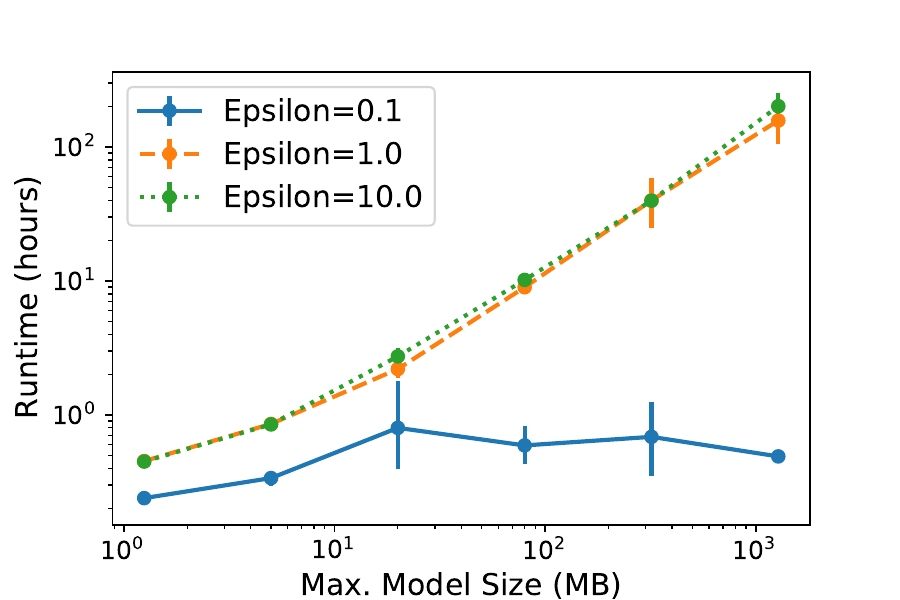}}
\caption{\label{fig:misc} Impact of the model capacity hyper-parameter on the error and runtime of \aim.}
\end{figure*}

In Line 12 of \aim (\cref{alg:aim}), we construct a set of candidates to consider in the current round based on an upper limit on \jtsize.  80 MB was chosen to match prior work,\footnote{Cai et al. \cite{cai2021data} limit the size of the \emph{largest} clique in the junction tree to have at most $10^7$ cells (80 MB with 8 byte floats), while we limit the \emph{overall} size of the junction tree.} but in general we can tune it as desired to strike the right accuracy / runtime trade-off.  Unlike other hyper-parameters, there is no ``sweet spot'' for this one: setting larger model capacities should always make the mechanism perform better, at the cost of increased runtime.  We demonstrate this trade-off empirically in \cref{fig:misc} (a-b).  For $\epsilon = 0.1, 1,$ and $10$, we considered model capacities ranging from 1.25 MB to 1.28 GB, and ran \aim on the \fire dataset with the \general workload.  Results are averaged over five trials, with error bars indicating the min/max runtime and workload error across those trials.  Our main findings are listed below:

\begin{enumerate}
\item As expected, runtime increases with model capacity, and workload error decreases with capacity.  The case $\epsilon=0.1$ is an exception, where both the plots level off beyond a capacity of $20 MB$.  This is because the capacity constraint is not active in this regime: \aim already favors small marginals when the available privacy budget is small by virtue of the quality score function for marginal query selection, so the model remains small even without the model capacity constraint.
\item Using the default model capacity and $\epsilon=1$ resulted in a 9 hour runtime.  We can slightly reduce error further, by about $13\%$, by increasing the model capacity to $1.28 GB$ and waiting 7 days.  Conversely, we can reduce the model capacity to $5 MB$ which increases error by about $75\%$, but takes less than one hour.  The law of diminishing returns is at play.
\end{enumerate}

Ultimately, the model capacity to use is a policy decision.  In real-world deployments, it is certainly reasonable to spend additional computational time for even a small boost in utility.
}

\revision{
\section{Adaptive Rounds Experiments} \label{sec:adaptive}

In \cref{sec:ablation}, we saw that the ``Adaptive Rounds + Budget Split'' element described in \cref{sec:aim} had an important role in the strong performance of \aim, offering $1.48\times$ improvement over the default number of rounds on average.  In this section, we expand that experiment, by comparing against alternatives that use a different number of rounds.  In particular, we run remove this element of \aim, and instead run the mechanism for fixed rounds $T$, where privacy budget is split evenly across rounds, and $T$ is varied from $\set{2,4,8,16,32,64,128,256,d}$.  For each choice of $T$, we run the mechanism across all experimental settings ($6$ datasets $\times$ $3$ workloads $\times$ $9$ privacy levels) and $5$ trials for each setting. We then compute the mean workload error across the five trials, and compare that to the mean workload error of \aim with adaptive rounds + budget split to calculate an ``improvement ratios''.  The distribution of improvment ratios for each $T$ across the $162$ experimental settings is shown in \cref{fig:adaptive}.  The geometric mean of improvement ratios is $5.07$ for $T=2$, $3.87$ for $T=4$, $2.36$ for $T=8$, $1.46$ for $T=16$, $1.13$ for $T=32$, $1.06$ for $T=64$, $1.17$ for $T=128$, $1.27$ for $T=256$, and $1.48$ for the defualt $T=d$.  Thus, it is better to use adaptive rounds + budget split than any fixed $T$.  If we chose the best $T$ in hindsight for each experimental setting, that is on average $1.1\times$ better than using adaptive rounds + budget split.  However, in general it would require spending significant privacy budget to optimize $T$ in practice \cite{liu2019private}, and this small improvement in utility would likely not be enough to warrant the high privacy cost of hyper-parameter optimization.  

\begin{figure}[H]
\includegraphics[width=0.9\linewidth]{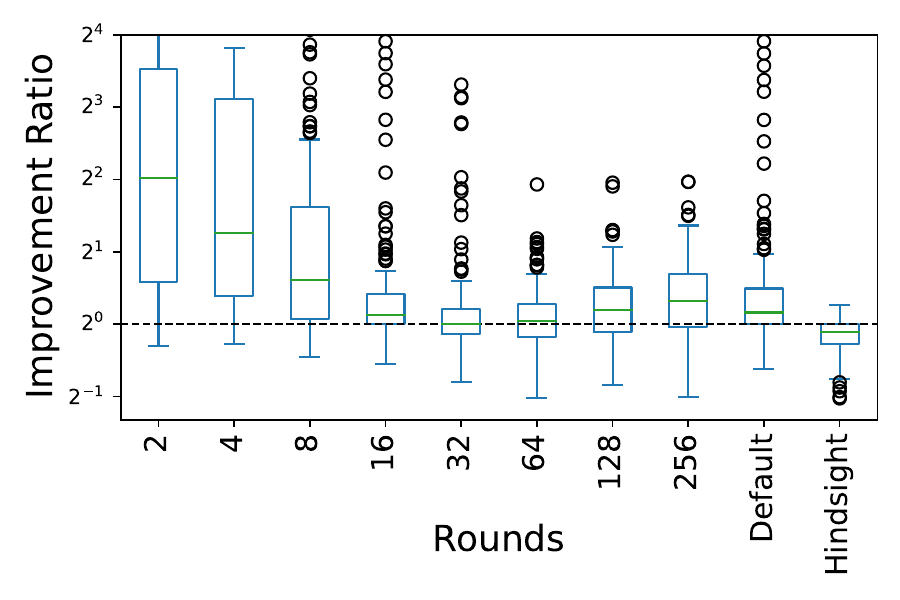}
\caption{\label{fig:adaptive} Adaptive rounds experiment}
\end{figure} }
\revision{\section{Sensitivity to Hyper-parameters} \label{sec:sensitivity}

In \cref{sec:experiments}, we evaluated both \aim and existing mechanisms using default hyper-parameter settings.  A natural question is to what extent the results would have changed if the hyper-parameters of competing mechanisms had been near-optimally tuned for each mechanism.  In this section, we explore this question, by evaluating three mechanisms (\mwempgm, \gem, and \rap) across a variety of hyper-parameter settings.  \gem and \rap were chosen because the reported results in their experimental evaluation was best performance on a grid of hyper-parameter settings.  \mwempgm was chosen because it is most similar to \aim.  The other mechanisms, like \privmrf, \privbayes, and \mst are not included in this experiment since these mechanisms were all evaluated with fixed hyper-parameter settings in their own experimental evaluation (and the effect of hyper-parameters were investigated in separate experiments, if applicable).  

For this experiment, we focus on the \adult dataset and the \general workload, but vary $\epsilon \in [0.01, 100]$.  We run each mechanism for $T = \set{2,4,8,16,32,64,75,128,256}$ rounds.  By default, \mwempgm runs for $d=15$ rounds, \gem runs for $75$ rounds, and \rap runs for $10$ rounds.   We used the same compute resources described in \cref{sec:experiments}, and used a $24$ hour time limit for $\mwempgm$ and $\gem$.  We used a smaller $4$ hour time limit for \rap, since that required significantly more RAM (64 GB) than the other mechanisms so only 2 executions of the mechanism could be run in parallel on the compute cluster used in experiments.  

\cref{fig:rounds} shows the performance of these mechanisms with varying numbers of rounds.  Note that none of the mechanisms completed for $256$ rounds, \gem did not complete for $128$ rounds and \rap did not complete for $64$ rounds.  This is because these mechanisms surpassed the time limit for the experiment.  For \mwempgm, the best hyper-parameter setting varied by $\epsilon$, with $T=8$ working well for $\epsilon \leq 0.3$, and $T \in \set{64, 75,128}$ working well for $\epsilon \geq 1$.  There is some loss in performance for using $T=16$ (the default is $T=15$), although the difference is small.  Note that state-of-the-art DP hyper-parameter selection algorithms come at a multiplicative $3 \times$ cost to the privacy parameter \cite{liu2019private, papernot2021hyperparameter}.  Thus, it is clear from \cref{fig:mwemablation} that it is better to use the default hyper-parameter setting than using a DP mechanism to optimize the hyper-parameter, especially for lower values of $\epsilon$.  For example, at $\epsilon=0.1$, \mwempgm with $T=16$ achieved a workload error of $0.62$.  The best workload error achieved at $\epsilon=0.032$ was only $0.86$ (for $T=8$), which would be worse than using the full private budget for the default hyper-parameter setting.  

The behavior of \gem demonstrated a similar trend, with $T=32$ working well for $\epsilon\leq 0.1$ and $T \in {64, 75}$ working well for $\epsilon \geq 0.32$.  In this case, the default setting was actually best in many cases, and in the settings where it was not optimal, it is still better than the alternative of optimizing the hyper-parameter with a DP mechanism, which comes at a $3 \times$ multiplicative privacy cost.  

The behavior of \rap  demonstrated the same trend, with larger values of $T$ working better for larger $\epsilon$, and vice-versa.

In short, while the default hyper-parameter settings are not optimal in hindsight, the amount of utility gained by optimizing hyper-parameters is not enough to warrant the high cost to privacy.  Moreover, even if we ignored the privacy cost of this hyper-parameter optimization, the utility gained would not be enough to substantially change our main findings from \cref{sec:experiments}.

\begin{figure*}[t!]
\centering
\subcaptionbox{\mwempgm \label{fig:mwemablation}}{\includegraphics[width=0.33\textwidth]{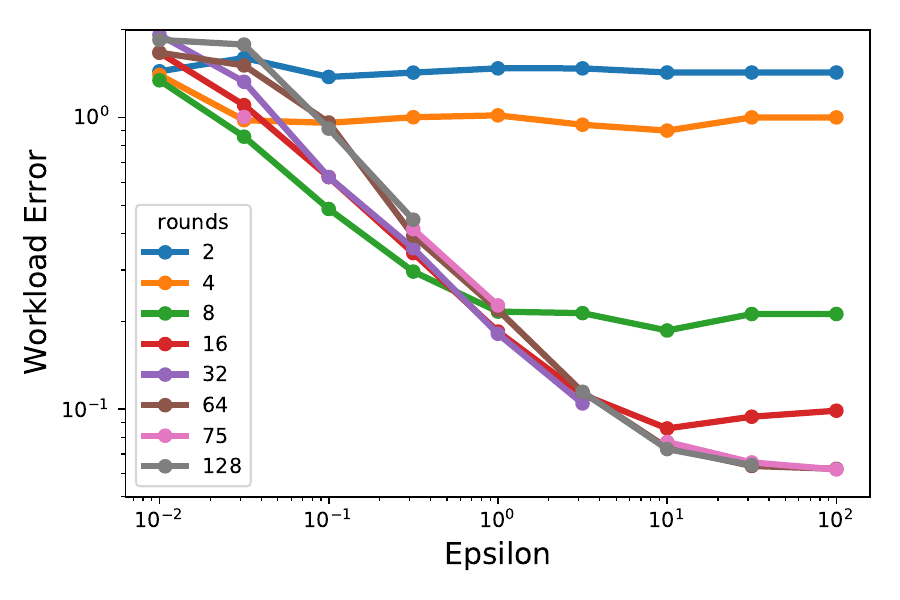}}
\subcaptionbox{\gem}{\includegraphics[width=0.33\textwidth]{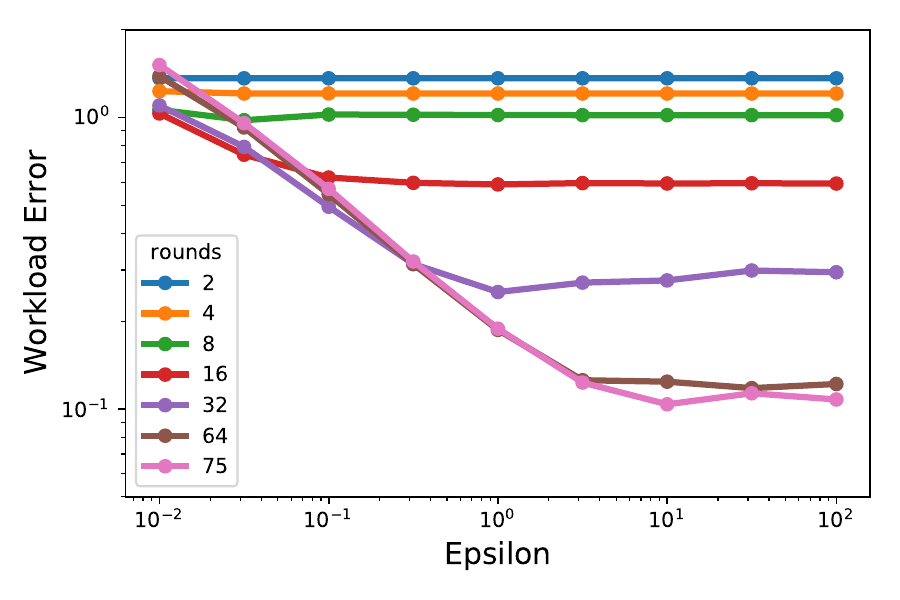}}
\subcaptionbox{\rap}{\includegraphics[width=0.33\textwidth]{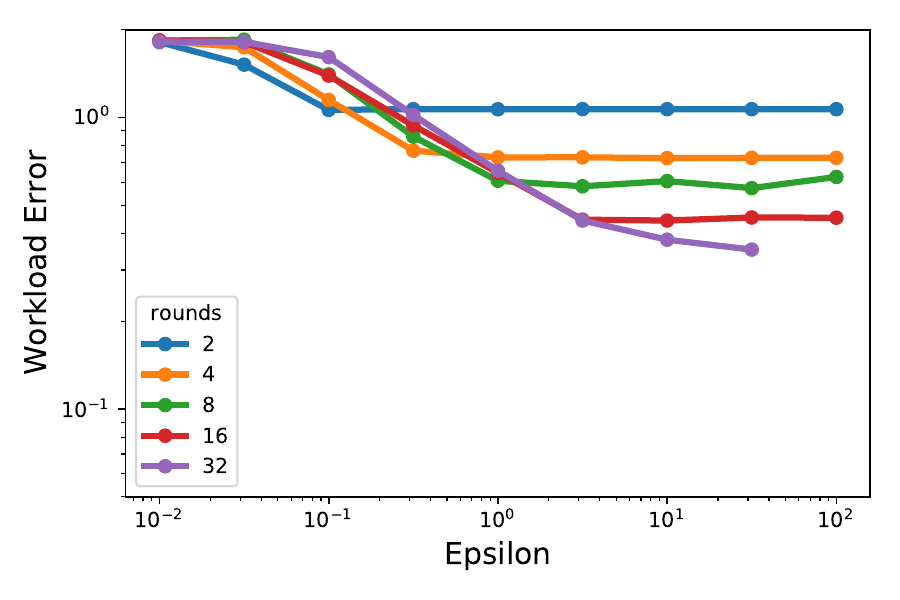}}
\caption{ \label{fig:rounds} Effect of the number of rounds on the performance of \mwempgm, \gem, and \rap. }
\end{figure*}
}

\revision{\section{Other Error Metrics} \label{sec:metrics}

In this work, we primarily focused on the $L_1$ workload error metric defined in \cref{def:error}, although other error metrics could have been used instead.  In \cref{fig:l2,fig:max} we visualize the error of \aim and other mechanism using an $L_2$ error metric as well as a max error metric.  The results for the $L_2$ error metric are similar to the standard $L_1$ metric, except \mwempgm outperforms \aim on the \salary dataset at $\epsilon \leq 1$.  On other datasets, the performance of \aim relative to competitors is roughly the same between the $L_1$ and $L_2$ metrics.  \aim is also often teh best-performing mechanism on the $L_{\inf}$ (max) error metric, although \gem does outperform it on both the \fire and \msnbc datasets.  It is not surprising that \gem outperforms \aim in some cases, since it specifically targets the $L_{\inf}$ error metric.  It is interesting that \aim still outperforms \gem and the other mechanisms on this error metric in other cases.  

\begin{figure*}[t!]
\centering
\includegraphics[width=\textwidth]{fig/legend}
\subcaptionbox{\vspace{1em}Adult}{\includegraphics[width=0.335\textwidth]{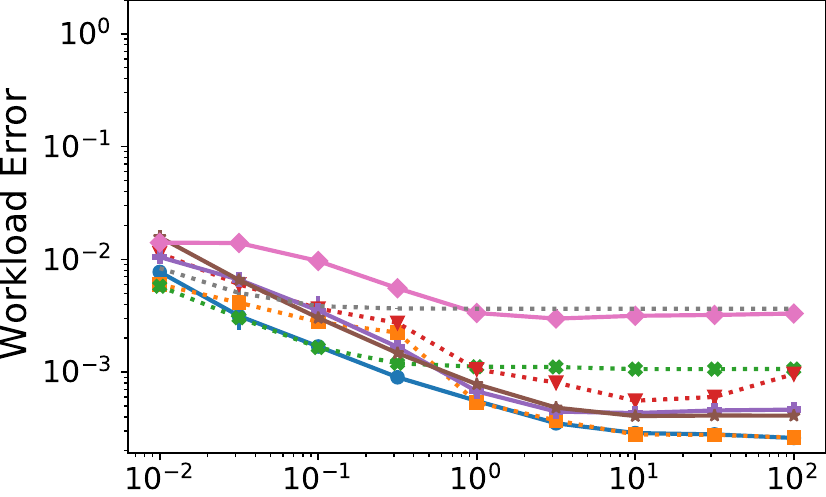}}
\subcaptionbox{Salary}{\includegraphics[width=0.318\textwidth]{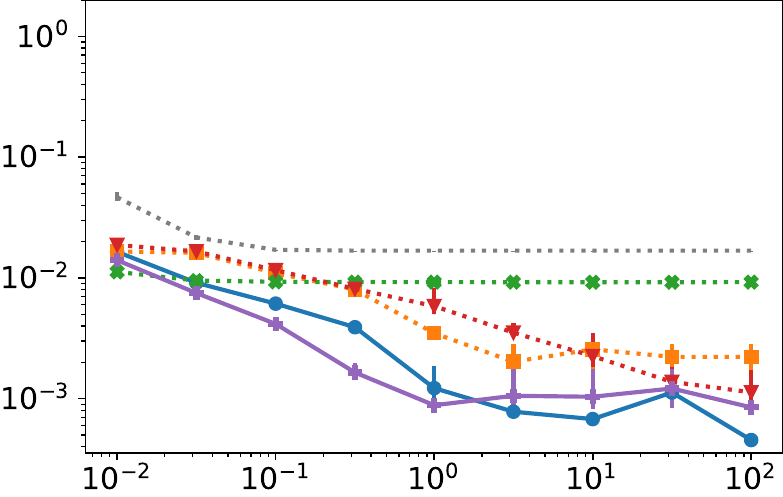}}
\subcaptionbox{MSNBC}{\includegraphics[width=0.318\textwidth]{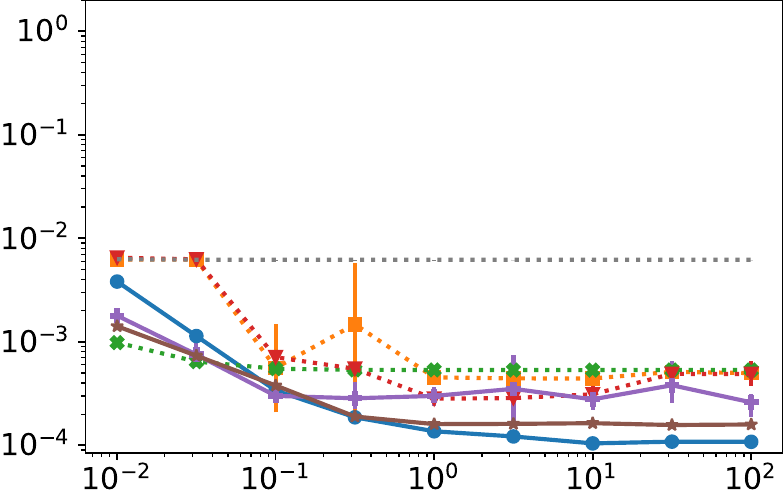}}
\subcaptionbox{Fire}{\includegraphics[width=0.335\textwidth]{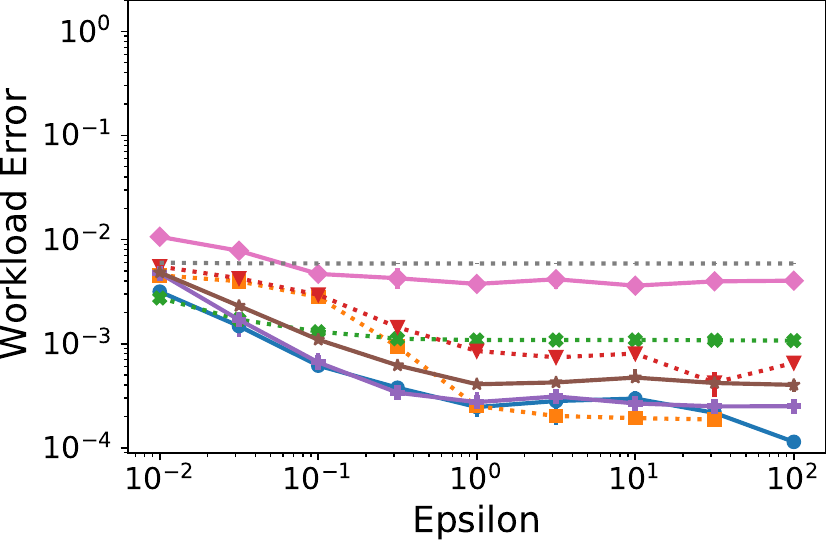}}
\subcaptionbox{NLTCS}{\includegraphics[width=0.318\textwidth]{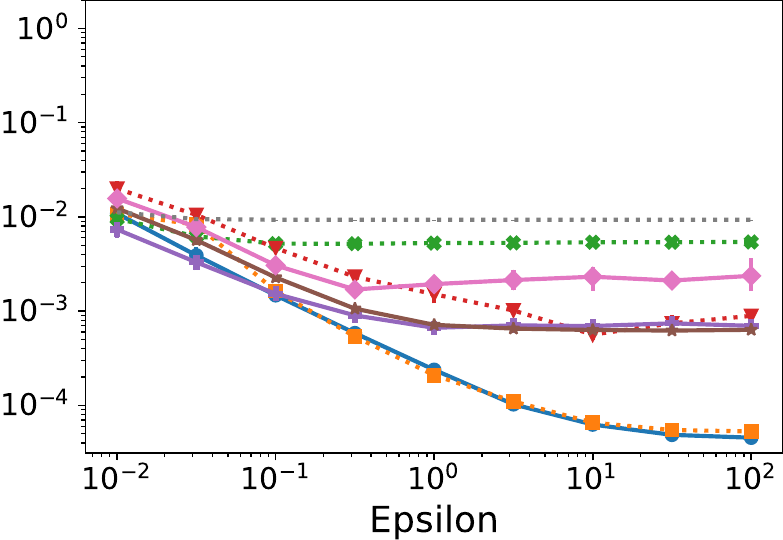}}
\subcaptionbox{Titanic}{\includegraphics[width=0.312\textwidth]{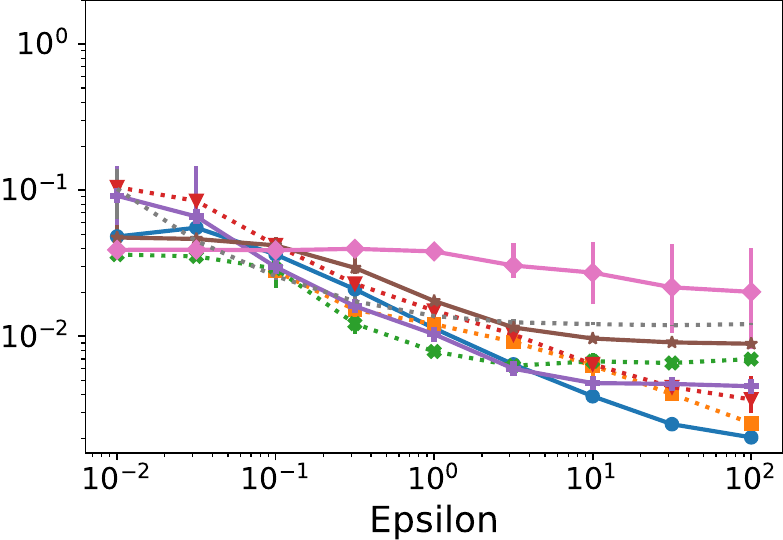}}
\caption{\label{fig:l2} $L_2$ workload error of competing mechanisms on the \general workload for $\epsilon = 0.01, \dots, 100$.}
\end{figure*}

\begin{figure*}[t!]
\centering
\includegraphics[width=\textwidth]{fig/legend}
\subcaptionbox{\vspace{1em}Adult}{\includegraphics[width=0.335\textwidth]{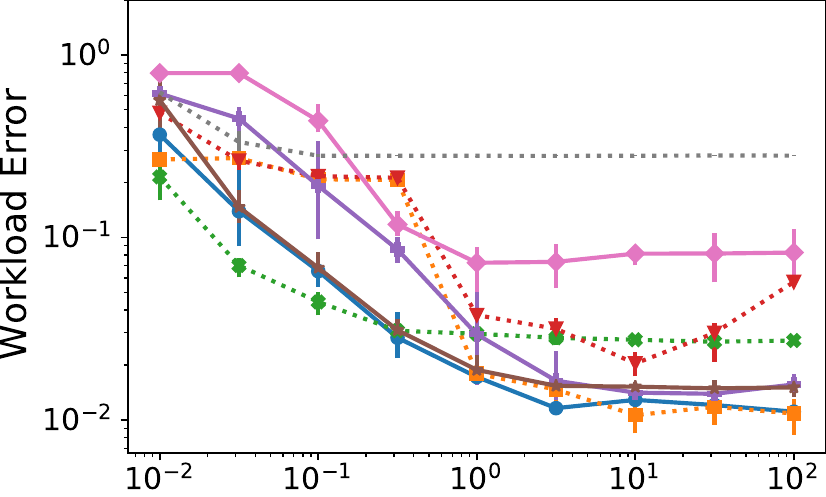}}
\subcaptionbox{Salary}{\includegraphics[width=0.318\textwidth]{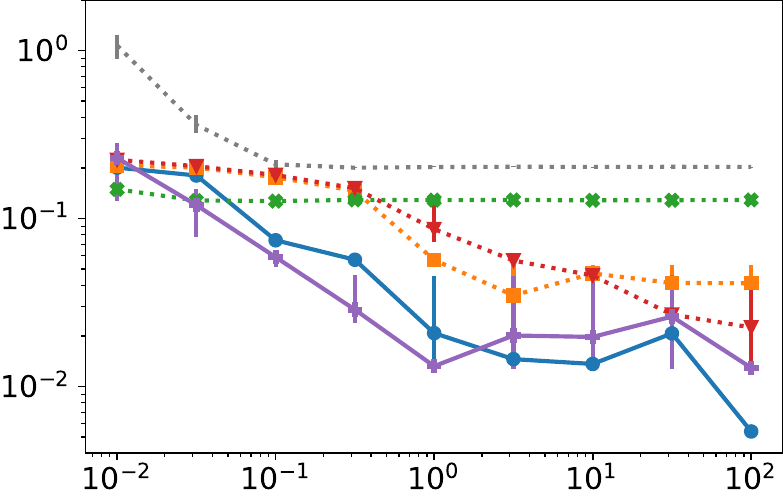}}
\subcaptionbox{MSNBC}{\includegraphics[width=0.318\textwidth]{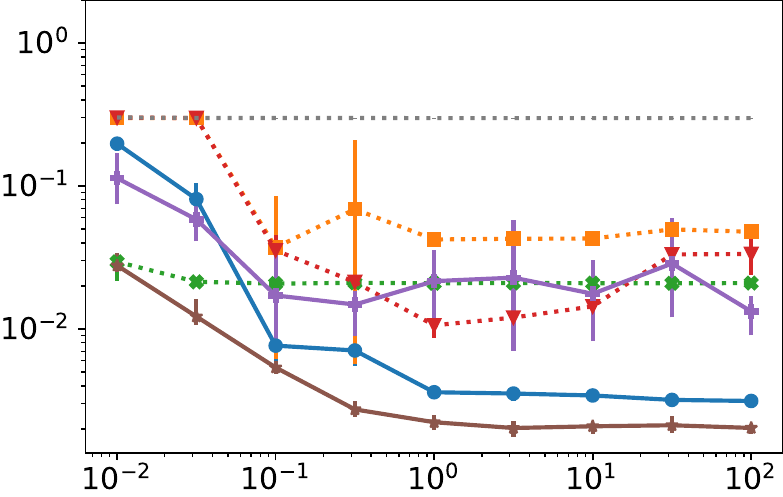}}
\subcaptionbox{Fire}{\includegraphics[width=0.335\textwidth]{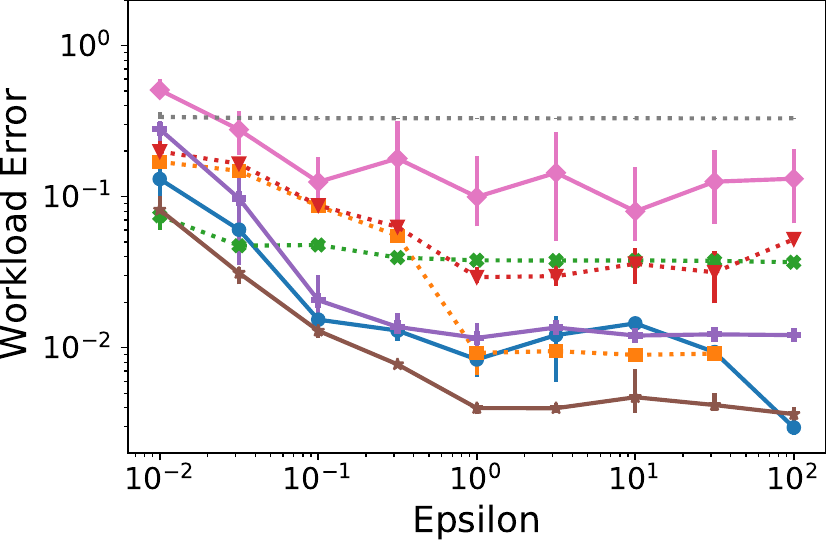}}
\subcaptionbox{NLTCS}{\includegraphics[width=0.318\textwidth]{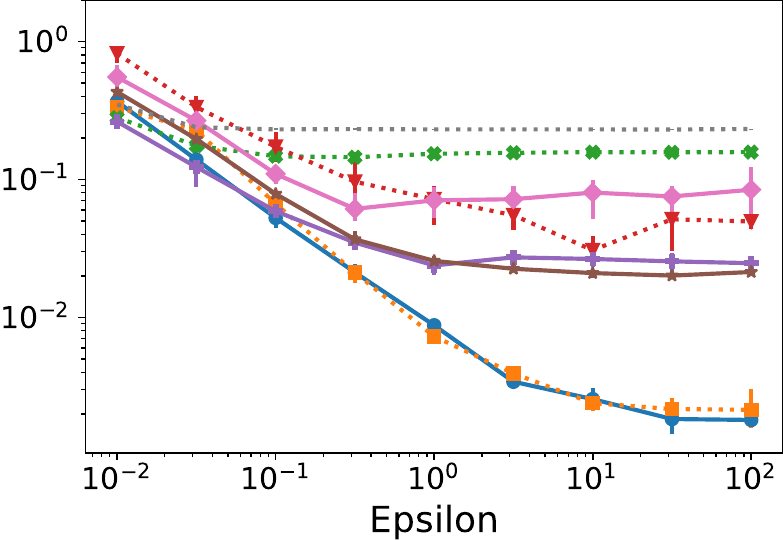}}
\subcaptionbox{Titanic}{\includegraphics[width=0.312\textwidth]{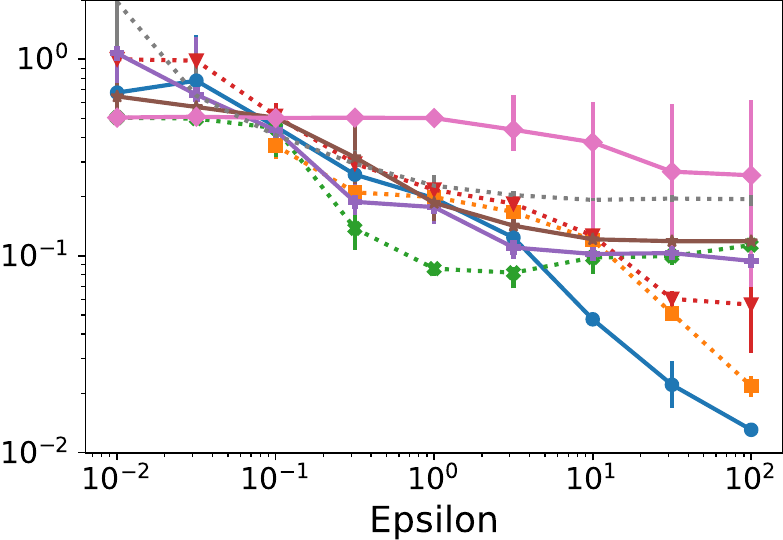}}
\caption{\label{fig:max} $L_{\inf}$ (max) workload error of competing mechanisms on the \general workload for $\epsilon = 0.01, \dots, 100$.}
\end{figure*}
}

\revision{
\section{Results on 2-way marginals} \label{sec:2way}

\begin{figure*}[t!]
\centering
\includegraphics[width=\textwidth]{fig/legend}
\subcaptionbox{\vspace{1em}Adult}{\includegraphics[width=0.335\textwidth]{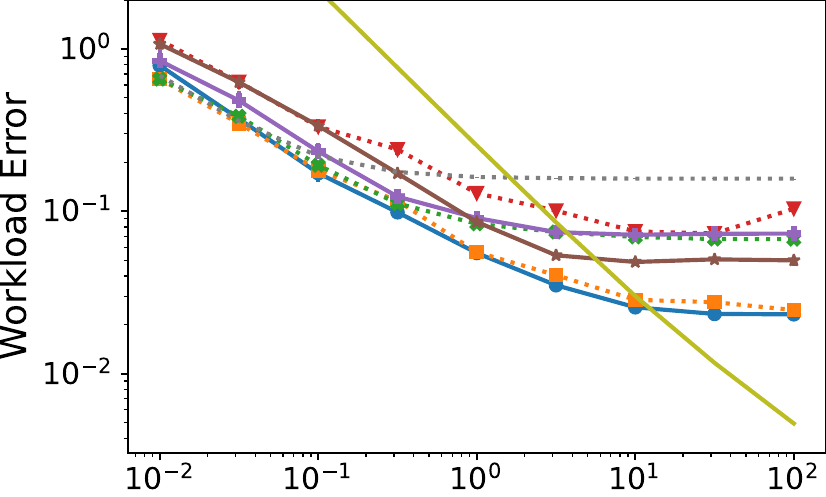}}
\subcaptionbox{Salary}{\includegraphics[width=0.318\textwidth]{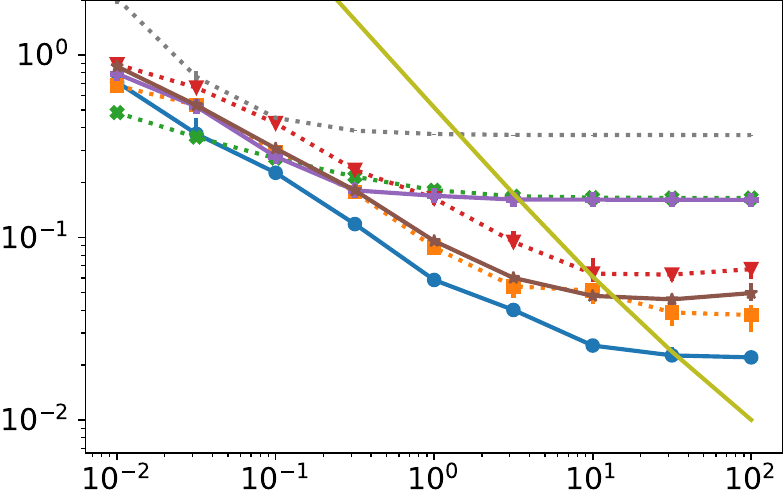}}
\subcaptionbox{MSNBC}{\includegraphics[width=0.318\textwidth]{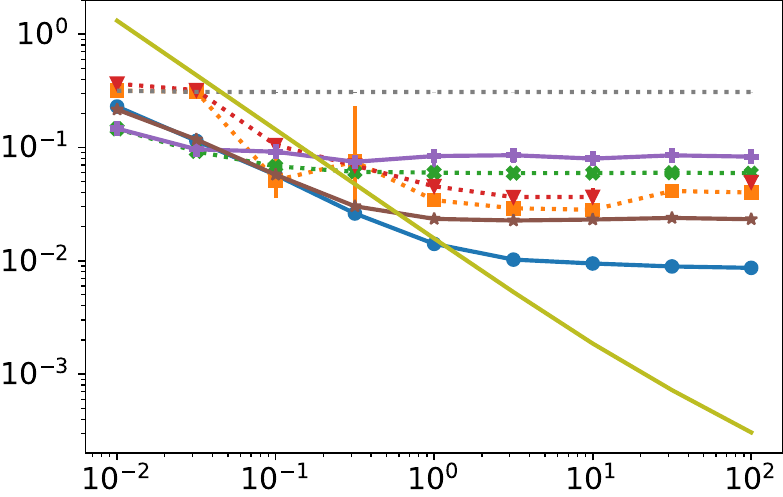}}
\subcaptionbox{Fire}{\includegraphics[width=0.335\textwidth]{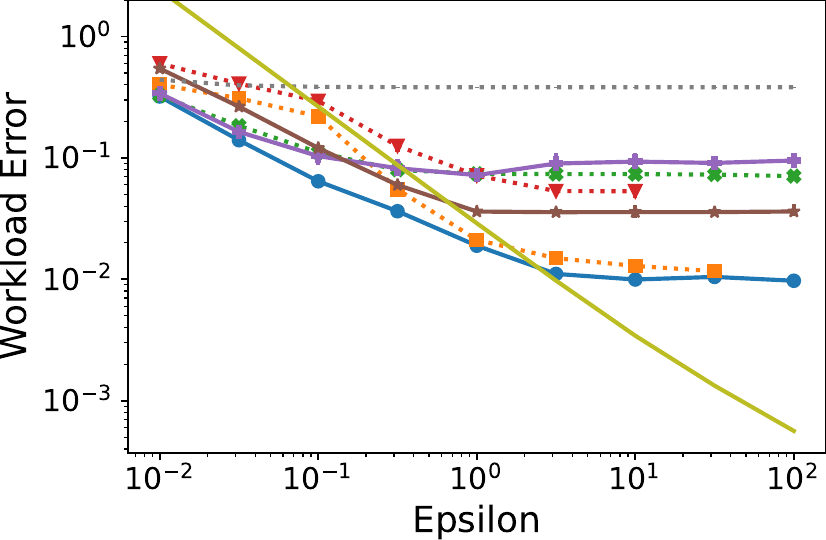}}
\subcaptionbox{NLTCS}{\includegraphics[width=0.318\textwidth]{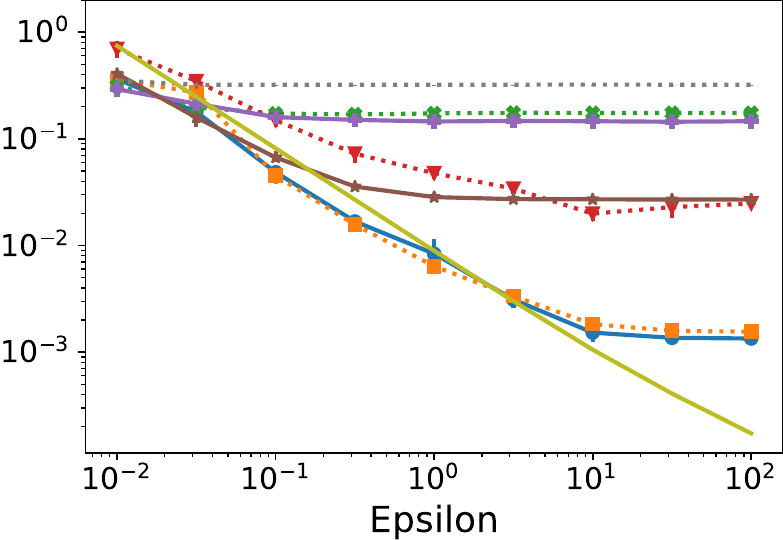}}
\subcaptionbox{Titanic}{\includegraphics[width=0.312\textwidth]{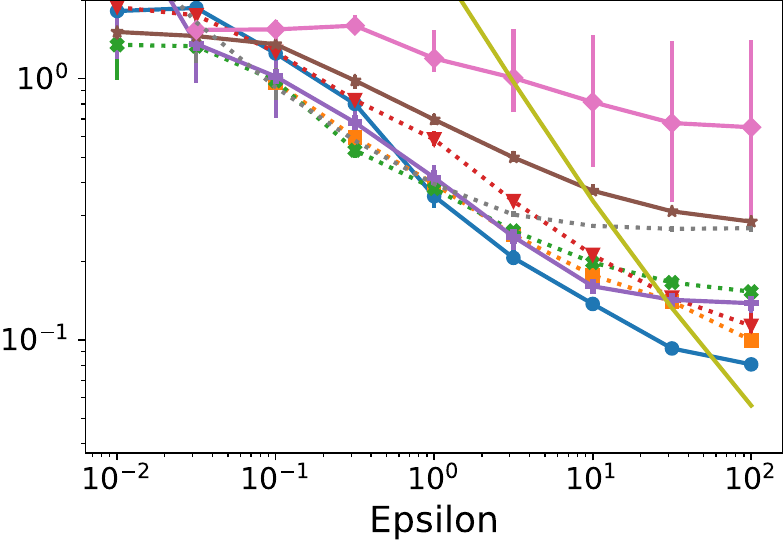}}
\caption{\label{fig:2way} workload error of competing mechanisms on the \workload{all-2way} workload for $\epsilon = 0.01, \dots, 100$.}
\end{figure*}

The workloads considered in our main experiments (\general, \target, and \weighted) consisted of different subsets of $3$-way marginal queries.  In this section, we provide additional experiments on the \workload{all-way} workload, which contains all 2-way marginal queries.   The experimental setup is exactly the same as described in \cref{sec:experiments}.  Results are shown in \cref{fig:2way}.  As before, \aim consistently outperforms all competing mechanisms on this workload.  Interestingly, \gaussian exhibits strong performance on this workload, especiall for higher values of $\epsilon$, and the \msnbc dataset.  While \gaussian was also strong in this regime for the other workloads, here it is surpassing \aim at $\epsilon=1$, whereas on the other workloads it was generally not competitive until $\epsilon \geq 31$.  This behavior can be attributed to the fact that \workload{all-2way} is a smaller workload than the others considered, and therefore the noise magnitude required to measure every query in the workload is not excessively large.  We remind the reader that unlike the other mechanisms, \gaussian is a baseline that \emph{does not} produce synthetic data --- it only produces noisy query answers.  Thus, even though it achieves lower workload error than the synthetic data mechanisms, it would not be a viable mechanism to use if synthetic data was needed.  

}

\end{document}